\documentclass[11pt, vi, twoside, singlespace]{mitthesis}
\usepackage{lgrind, graphicx, amssymb, amsmath, amsfonts, amsthm, dsfont}
\input{Qcircuit}

\begin{document}

\newcommand{\ud}{\mathrm{d}}
\newcommand{\braket}[2]{\langle #1|#2\rangle}
\newcommand{\Bra}[1]{\left<#1\right|}
\newcommand{\Ket}[1]{\left|#1\right>}
\newcommand{\Braket}[2]{\left< #1 \right| #2 \right>}
\renewcommand{\th}{^\mathrm{th}}
\newcommand{\tr}{\mathrm{Tr}}
\newcommand{\mx}{\begin{array}{l} \includegraphics[width=0.2in]{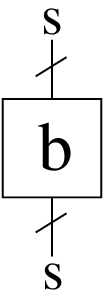} \end{array}}
\newcommand{\ttr}{\widetilde{\tr}}
\newcommand{\id}{\mathds{1}}
\newcommand{\cu}{\mathcal{U}}
\newcommand{\ca}{\mathcal{A}}
\newcommand{\hv}{\hat{V}}

\newtheorem{lemma}{Lemma}
\newtheorem{theorem}{Theorem}
\newtheorem{definition}{Definition}
\newtheorem*{prop1}{Proposition 1}
\newtheorem{proposition}{Proposition}
\newtheorem{corollary}{Corollary}
\newtheorem{A}{Proposition $A_n$}
\newtheorem{B}{Proposition $B_n$}
\newtheorem{problem}{Problem}

\title{Quantum Computation Beyond the Circuit Model}

\author{Stephen Paul Jordan}
\department{Department of Physics}
\degree{Doctor of Philosophy in Physics}
\degreemonth{May}
\degreeyear{2008}
\thesisdate{May 2008}


\supervisor{Edward H. Farhi}{Professor of Physics}

\chairman{Thomas J. Greytak}{Professor of Physics}

\maketitle



\cleardoublepage
\setcounter{savepage}{\thepage}
\begin{abstractpage}
The quantum circuit model is the most widely used model of quantum
computation. It provides both a framework for formulating quantum
algorithms and an architecture for the physical construction of
quantum computers. However, several other models of
quantum computation exist which provide useful alternative frameworks
for both discovering new quantum algorithms and devising new physical
implementations of quantum computers. In this thesis, I first present
necessary background material for a general physics audience and discuss
existing models of quantum computation. Then, I present three new
results relating to various models of quantum computation:
a scheme for improving the intrinsic fault tolerance of adiabatic
quantum computers using quantum error detecting codes, a proof that a
certain problem of estimating Jones polynomials is complete for the
one clean qubit complexity class, and a generalization of perturbative
gadgets which allows $k$-body interactions to be directly simulated
using 2-body interactions. Lastly, I discuss general principles
regarding quantum computation that I learned in the course of my
research, and using these principles I propose directions for future
research.

\end{abstractpage}


\cleardoublepage

\section*{Acknowledgments}

I first wish to thank the people who have participated most directly
in my formation as a quantum information scientist. At the top of this
list is Eddie Farhi, who I thank for offering me invaluable help on
matters both scientific and logistical, and for being a pleasure to
work with. Without him, my graduate school experience would not have
been as rich and rewarding as it was. I thank Peter Shor for being
always willing to chat about any scientific topic, being a frequent
collaborator, serving on my general exam and thesis committees, and
teaching quantum computation courses that I greatly enjoyed. In many
ways, I feel that he is like a second advisor to me. I thank Isaac
Chuang for serving on my general exam and thesis committees, teaching
quantum computation courses which I greatly benefitted from, and for
interesting conversations. I thank Seth Lloyd for interesting
conversations and for serving on my general exam committee. I thank
Al{\'a}n Aspuru-Guzik at Harvard for bringing me in on his chemical
dynamics project, and for spreading his enthusiasm for research. I
also thank each of these people for aiding me in obtaining a
postdoctoral position.

I wish to thank those who have collaborated with me directly on
papers, some of which form much of the content of this thesis: Eddie
Farhi, Peter Shor, Al\'an Aspuru-Guzik, Ivan Kassal, Peter Love,
Masoud Mohseni, David Yeung, and Richard Cleve. I am also grateful to
Richard Cleve for many long and interesting conversations about
quantum computation.  I thank John Preskill, Howard Barnum, Lov
Grover, Daniel Lidar, Joe Traub, and the people of Perimeter Institute
for inviting me to visit them, and having interesting conversations
when I did. I am grateful to many people for having interesting
conversations that shaped my thoughts on quantum computing, and for
being helpful in various ways. These include Scott Aaronson, Dave
Bacon, Jacob Biamonte, Andrew Childs, Wim van Dam, David DiVincenzo,
Pavel Etingof, Steve Flammia, Andrew Fletcher, Joe Giraci, Jeffrey
Goldstone, Daniel Gottesman, Sam Gutmann, Aram Harrow, Elham Kashefi,
Jordan Kerenidis, Ray Laflamme, Mike Mosca, Andrew Landahl, Debbie
Leung, Michael Levin, William Lopes, Carlos Mochon, Shay Mozes, Markus
Mueller, Daniel Nagaj, Ashwin Nayak, Robert Raussendorf, Mark Rudner,
Rolando Somma, Madhu Sudan, Jake Taylor, David Vogan, and Pawel
Wocjan.

I wish to thank the people and organizations who have supported me and
my research. I thank the society of presidential fellows for
supporting me during my first year at MIT. I thank Isaac Chuang and
Ulrich Becker for giving me the opportunity to TA for their junior lab
course during my second year. I thank Mark Heiligman, T.R. Govindan,
Mel Currie, and all the other people of ARO and DTO for supporting me
for the rest of my time at graduate school as a QuaCGR fellow. I thank
Senthil Todadri and Alan Guth for serving as academic advisors. I thank
the US Department of Energy for funding MIT's center for theoretical
physics, where I work, and I thank the administrative staff of CTP,
Scott Morley, Charles Suggs, and Joyce Berggren for being so helpful
and making CTP a functional and pleasant place to work.

Lastly, but by no means least importantly I wish to thank those people
who have contributed to this thesis indirectly by contributing to my
development in general: my parents Eric and Janet, my wife Sara,
my sisters Katherine and Elizabeth, my undergraduate research 
advisors, Moses Chan, Rafael Garcia, and Vincent Crespi, my high
school physics and chemistry teacher Kevin McLaughlin, and all of the
friends and teachers who I have been lucky enough to have.


\pagestyle{plain}

\tableofcontents

\chapter{Introduction}
\label{introchap}

\section{Classical Computation Preliminaries}
\label{classical_prelim}

This thesis is about quantum algorithms, complexity, and models of
quantum computation. In order to discuss these topics it
is necessary to use notations and concepts from classical computer
science, which I define in this section.

The ``big-O'' family of notations greatly aids in analyzing both
classical and quantum algorithms without getting mired in minor
details. Although it may seem like a trivial notation, it is the first
step in a chain of increasing abstraction which allows computer
scientists to analyze the general laws of computation which apply
whether the computer is using base 2 or base 20, and whether it is
made of transistors or tinker toys. By following this chain we will reach
the major open questions about complexity classes and their relations
to one another. The big-O notation is defined as follows. 

\begin{definition}
Given two functions $f(n)$ and $g(n)$, $f(n)$ is $O(g(n))$ if
there exist constants $n_0$ and $c > 0$ such that $|f(n)| < c |g(n)|$ for
all $n > n_0$.
\end{definition}

\begin{definition}
Given two functions $f(n)$ and $g(n)$, $f(n)$ is $\Omega(g(n))$ if
there exist constants $n_0$ and $c > 0$ such that $|f(n)| > c |g(n)|$
for all $n > n_0$.
\end{definition}

\begin{definition}
Given two functions $f(n)$ and $g(n)$, $f(n)$ is $\Theta(g(n))$ if it
is $O(g(n))$ and $\Omega(g(n))$.
\end{definition}

Thus, $O$ describes upper bounds, $\Omega$ describes lower bounds, and
$\Theta$ describes asymptotic behavior modulo an overall
multiplicative constant. The standard way to describe the efficiency
of an algorithm is to use ``big-O'' notation to describe the number of
computational steps the algorithm uses to solve a problem as a
function of number of bits of input. For example, the
standard method for multiplying two $n$-digit numbers that is taught in
elementary school uses $\Theta(n^2)$ elementary operations in which
individual digits are manipulated. By using big-O notation we can avoid
distinguishing between $2 n^2$ and $50 n^2$ which allows us to
disregard unnecessary details. 

Moving one level higher in abstraction, we reach computational
complexity classes. These are sets of problems solvable with a given
set of computational resources. For example, the complexity class P is
the set of problems solvable on a Turing machine in a number of steps
which scales polynomially in the number of bits of input $n$, that is,
with $O(n^c)$ steps for some constant $c$. Note that for a problem to
be contained in P, all problem instances of size $n$ must be solvable
in time $\mathrm{poly}(n)$, including highly atypical worst-case
instances.

Complexity classes are usually defined in terms of decision
problems. These are problems that admit a yes/no answer, such as the
problem of determining whether a given integer is prime. Many problems
are not of this form. For example, the problem of factoring
integers has an output which is a list of prime factors. However, it
turns out that in almost all cases, problems can be reduced with
polynomial overhead to decision versions. For example, consider the
problem where, given two numbers $a$ and $b$, you are asked to answer
whether $a$ has a prime factor smaller than $b$. Given a polynomial
time algorithm solving this problem, one can construct a polynomial
time algorithm for factoring using this algorithm as a
subroutine\footnote{This, like many reductions to decision problems,
  can be done using the process called binary search.}  Thus by
considering only decision problems (or more technically, the
associated languages), complexity theorists are simplifying things
without losing anything essential. Because problems are usually
equivalently hard to their decision versions, we will often gloss over
the distinction between the problems and their decision versions in
this thesis.

Some complexity classes describe models of computation which are
essentially realistic. P describes problems solvable in polynomial
time on Turing machines. Until recently, every plausible deterministic
model of universal computation has led to the same set of problems
solvable in polynomial time. That is, all models of universal
computation could be simulated with polynomial overhead by a Turing
machine, and vice versa. Thus the complexity class P was regarded as a
robust description of what could be efficiently computed
deterministically in the real world, which captures something
fundamental and is not just an artifact of the particular model of
computation being studied.

For example, in the standard formulation of a Turing machine, each
location on the tape can take two states. That is, it contains a
bit. If you instead allow $d$ states (a ``dit'') then the speedup is
only by a constant factor, which already disappears from our notice
when we use big-O notation. Furthermore, even parallel computation,
although useful in practice, does not generate a complexity
class distinct from P, provided one allows at most polynomially many
processors as a function of problem size.

One may ask why polynomial time is chosen as the definition of
efficiency. Certainly it would be a stretch to consider an algorithm
operating in time $n^{25}$ efficient, or an algorithm operating in time
$2^{\lceil 0.0001 \sqrt{n} \rceil}$ inefficient. There are several
reasons for using polynomial time as a mathematical formalization of
efficiency. First of all, it is mathematically convenient. It is a
robust definition which allows one to ignore many details of the
implementation. Furthermore, asymptotic complexity is more robust than
the complexity of small instances, which can be influenced by the
presence of lookup tables or other preprocessed information hidden in
the program. Secondly, it appears to do a good job of sorting the
efficient algorithms from the inefficient algorithms in practice. It
is rare to obtain a polynomial time classical algorithm with runtime
substantially greater than $n^3$ or a superpolynomial time algorithm
with runtime substantially less than $2^n$. Furthermore, whenever
polynomial time algorithms are found with large exponents, it usually
turns out that either the runtime in practice is much better than the
worst case theoretical runtime, or a more efficient algorithm is
subsequently found. 

Sometimes a problem not known to be in P seems to be efficiently
solvable in practice. This can happen either because the problem is
not in P but the worst case instances are hard to construct, or
because the problem actually is in P but the proof of this fact is
difficult. Linear programming provides an interesting and historically
important example of the latter. For this problem the best known
algorithm was for a long time the simplex method, which had
exponential worst-case complexity, but was generally quite efficient
in practice. An algorithm with polynomial time worst-case complexity
has since been discovered. 

One can also consider probabilistic computation. That is, one can give
the computer the ability to generate random bits, and demand only that
it give the correct answer to a problem with high probability. It is
clear that the set of problems solvable in this way contains P and
possibly goes beyond it. The standard formalization of this notion is
the complexity class BPP, which is defined as the set of decision problems
solvable on a probabilistic Turing machine with probability at least
$2/3$ of giving the correct answer. (BPP stands for Bounded-error
Probabilistic Polynomial-time.) Note that, like P, BPP is defined using
worst-case instances. The probabilities appear not by
randomizing over problem instances, but by randomizing over the random
bits used in the probabilistic algorithm. The probability $2/3$
appearing in the definition of BPP may appear arbitrary, and in
addition, not very high. However, choosing any other fixed probability
strictly between $1/2$ and $1$ yields the same complexity class. This
is because one can amplify the success probability arbitrarily by
running the algorithm multiple times and taking the majority vote.

Prior to the discovery of quantum computation, no plausible model of
computation was known which led to a larger complexity class than
BPP. Just as all plausible models of classical deterministic
computation turned out to be equivalent up to polynomial overhead, the
same was true for classical probabilistic computation. Furthermore, it
is now generally suspected that BPP=P. In practice, randomized
algorithms usually work just fine if the random bits are replaced by
pseudorandom bits, which although generated by deterministic
algorithms, pass most naive tests of randomness (e.g. the various means
and correlations come out as one would expect for random bits, obvious
periodicities are absent, and so forth). The conjecture that P=BPP is
currently unproven, and finding a proof is a major open problem in
computer science. (Until recently, the problem of primality testing
was known to be in BPP but not P, increasing the plausibility that the
classes are distinct. However, a deterministic polynomial-time
algorithm for this problem was recently discovered.) The notion that
BPP captures the power of polynomial time computation in the real
world was eventually formalized as the strong Church-Turing thesis,
which states:
\\ \\
\emph{Any ``reasonable'' model of computation can be efficiently
  simulated on a probabilistic Turing machine.}
\\ \\
If BPP=P then dropping the word ``probabilistic'' results in an
equivalent claim. The strong Church-Turing thesis is named in
reference to the original Church-Turing thesis, which states
\\ \\
\emph{Any ``reasonable'' model of computation can be simulated on a
  Turing machine.}
\\ \\
The original Church-Turing thesis is a statement only about what is
computable and what is not computable, where no limit is made on the
amount of time or memory which can be used. Thus we have reached a
very high level of abstraction at which runtime and memory
requirements are ignored completely. It is perhaps not intuitively
obvious that with unlimited resources there is anything one
cannot compute. The fact that uncomputable functions exist was a
profound realization with an simple proof. One can see by
Cantor diagonalization that the set of all decision problems, \emph{i.e.}
the set of all maps from bitstrings to $\{0,1\}$, is uncountably
infinite. In contrast, the set of all computer programs, which can be
represented as bitstrings, is only countably infinite. Thus,
computable functions make up an infinitely sparse subset of all
functions. 

This leaves the question of whether one can find a natural function
which is uncomputable. In 1936, Alan Turing showed that the problem of
deciding whether a given program terminates is undecidable. This can
be proven by the following simple \emph{reductio ad absurdum}. Suppose
you had a program $A$, which takes two inputs, a program, and the data
on which the program is to act. $A$ then answers whether the given
program halts when run on the given data. One could use
$A$ as a subroutine to construct another program $B$ that takes a
single input, a program. $B$ determines whether the program halts when
given itself as the data. If the answer is yes, then $B$
jumps into an infinite loop, and if the answer is no then $B$
halts. By operating $B$ on itself one thus arrives at a
contradiction.

As we shall see in section \ref{algorithms}, quantum computers provide
the first significant challenge to the strong Church-Turing thesis. In
general, it is difficult to prove that one model of computation is
stronger than another. By discovering an algorithm, one can show that
a given model of computation can solve a certain problem with a
certain number of computational steps. However, in most cases it is
not known how to show that no efficient algorithm on a given model of
computation exists for a given problem. In 1994, Peter Shor discovered
a quantum algorithm which can factor any $n$-bit number in $O(n^3)$
time\cite{Shor_factoring}. There is no proof that this problem cannot
be solved in polynomial time by a classical computer. However, no
polynomial time classical algorithm for factoring has ever been
discovered, despite being studied since at least 200BC (\emph{cf.}
sieve of Eratosthenes). Furthermore, factoring has been well-studied
in modern times because a polynomial time algorithm for factoring
would allow the decryption of the RSA public key cryptosystem, which
is used ubiquitously for electronic transactions. The fact that
quantum computers can factor efficiently and classical computers can't
is one of the strongest pieces of evidence that quantum computers are
more powerful than classical computers.

A second piece of evidence for the power of quantum computers is that
quantum algorithms are known which can efficiently simulate the time
evolution of many-body quantum systems, a task which classical
computers apparently cannot perform despite decades of effort along
these lines. The search for polynomial time classical algorithms to
simulate quantum systems has been intense because they would have
large economic and scientific impact. For example, they would greatly
aid in the design of new materials and medicines, and could aid in the
understanding of mysterious condensed-matter systems, such as
high-temperature superconductors.

Quantum computers do not provide a challenge to the original
Church-Turing thesis. As we shall see in section \ref{quantum_prelim},
the behavior of a quantum computer can be completely predicted by
muliplying together a series of exponentially large unitary
matrices. Thus any problem solvable on a quantum computer in
polynomial time is solvable on a classical computer in exponential
time. Hence quantum computers cannot solve problems such as the
halting problem.

Other complexity classes relating to realistic classical models of
computation have been defined. (See \cite{Papadimitriou} for
overview.) These are weaker than BPP. The most
important of these for the purposes of this thesis are L and NC1. L
stands for Logarithmic space, and NC stands for Nick's Class. L is
the set of problems solvable using only logarithmic memory (other than
the memory used to store the input). The class NC1 is the set of
problems solvable using classical circuits of logarithmic
depth. Similarly, NC2 is the set of problems solvable in depth
$O(\log^2(n))$, and so on. Roughly speaking, NC1 can be identified
as those problems in P which are highly parallelizable. For a detailed
explanation of why this is a reasonable intepretation of this
complexity class see \cite{Papadimitriou}. For an illustration of the
meaning of circuit depth see figure \ref{depth}.

\begin{figure}
\begin{center}
\includegraphics[width=0.3\textwidth]{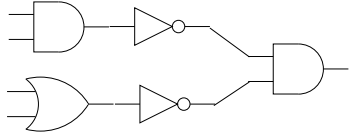}
\caption{\label{depth} The most common way to measure the complexity
  of a circuit is the number of gates, in this case five. However, one
  can also measure the depth. In this example, the circuit is three
  layers deep. The number of gates corresponds to the number of steps
  in the corresponding sequential algorithm. The depth corresponds to
  the number of steps in the corresponding parallel algorithm, since
  gates within the same layer can be performed in parallel.}
\end{center}
\end{figure}

I have described NC1 using logic circuits. The classical complexity classes
such as P and BPP can also be defined using logic circuits such
as the one shown in figure \ref{depth}. A given circuit takes a fixed
number of bits of input (four in the circuit of figure
\ref{depth}). Thus, an algorithm for a given problem corresponds to an
infinite family of circuits, one for each input size. It is tempting to
suggest that P consists of exactly those problems which can be solved
by a family of circuits in which the number of gates scales
polynomially with the input size. However, this is not quite
correct. The problem is that we have not specified how the circuits
will be generated. It is unreasonable to specify an algorithm by an
infinitely long description containing the circuit for each possible
input size. Such arbitrary families of circuits are called
``nonuniform''.

The set of problems solvable by polynomial size nonuniform circuits
may be much larger than P, because one can precompute the answers to
the problem and hide them in the circuits. One can even ``solve''
uncomputable problems this way. A uniform family of circuits is one
such that given an input size $n$, one can efficiently generate a
description of the corresponding circuit. One may for example demand
that a fixed Turing machine can produce a description of the circuit
corresponding to $n$, given $n$ as an input. In practice, a family of
circuits is usually described informally, such that it is easily
apparent that it is uniform. The set of decision problems efficiently
solvable by a uniform family of polynomial size circuits is exactly P.

While discussing circuits, it bears mentioning that the set of gates
used in figure \ref{depth}, namely AND, OR, and NOT, are
universal. That is, any function from $n$ bits to one bit (here we are
again restricting to decision problems out of convenience and not
necessity) can be computed using some circuit constructed from these
elements. The proof of this is fairly easy, and the interested reader
may work it out independently. A solution is given in appendix
\ref{classical_universality}. Note that the number of possible
functions from $n$ bits to one bit is $2^{2^n}$. In contrast the
number of possible circuits with $n$ gates is singly exponential in
$n$. Thus, most of the functions on $n$ bits must have exponentially
large circuits. Both this universality result and this counting
argument have quantum analogues, which are discussed in subsequent
sections.

For this thesis, P, BPP, L, and NC1 are a sufficient set of realistic
classical complexity classes to be familiar with. We'll now move on to
describe a few of the more fanciful classes. These describe models of
computation which are not realistic and classes of problems not
necessarily expected to be efficiently solvable in the real world. The
most important of these is NP. NP stands for Nondeterministic
Polynomial-time. Loosely speaking, it is the set of problems whose
solutions are verifiable in polynomial time. NP contains P, because if
you have a polynomial time algorithm for correctly solving a problem,
you can always verify a proposed solution in polynomial time by simply
computing the solution yourself.

More precisely, NP is defined in terms of witnesses (also sometimes
called proofs or certificates). These are simply bitstrings which
certify the correctness of an answer to a given problem. NP is the set
of decision problems such that there exists a polynomial time
algorithm (called the verifier), such that if the answer to the
instance is yes, there exist a biststring of polynomial length (the
witness), which the verifier accepts. If the answer to the instance is
no, then the verifier will reject all inputs. This definition can be
illustrated using Boolean satisfiability, which is a cannonical
example of a problem in NP. The problem of Boolean satisfiability is,
given a Boolean formula on $n$ variables, determine whether there is
some assignment of true/false to these variables which makes the
Boolean formula true. The witness in this case is a string of $n$ bits
listing the true/false values of each of the variables. The verifier
simply has to substitute these values in and evaluate the Boolean
formula, a task easily doable in polynomial time.

NP apparently does not correspond to the set of problems efficiently
solvable using any realistic model of computation. Why then would
anyone study NP? One reason is that, although there is clearly more
practical interest in understanding which problems are efficiently
solvable, there is certainly some appeal at least philosophically, in
knowing which problems have efficiently verifiable solutions. Perhaps
the most important motivation, however, is that by introducing a
strange model of computation such as nondeterministic Turing machines,
we gain a tool for classifying the difficulty of computational
problems.

When faced with a difficult computational problem, it is very
difficult to know whether one's inability to find an efficient
algorithm is fundamental or merely a failure of imagination. How can
one know whether it is time to give up, or whether the solution around
the next corner? Complexity classes give us two handles on the
difficulty of a computational problem: containment and
hardness. Containment is the more straightforward of the two. If a
problem is contained in a given complexity class, then it can be
solved by the corresponding model of computation. In a sense this
gives an upper bound on the problem's difficulty. The less obvious
concept is hardness. In computer science, ``hardness'' is a technical
term with a precise meaning different from its common usage. If a
problem is hard for a given complexity class, this means that any
problem in that class is reducible to an instance of that problem. For
example, if a problem is NP-hard, it means that any problem contained
in NP can be reduced to an instance of that problem in polynomial time
and with at most polynomial increase in problem size. Thus, up to
polynomial factors, an NP-hard problem is at least as hard as any
problem in NP. If one could solve that problem in polynomial time,
then one could solve all NP problems in polynomial time.

It not obvious that NP-hard problems exist. After all, how could one
ever show that every single problem in NP reduces to a given problem?
We don't even know what all the problems in NP are! We'll use Boolean
satisfiability as an example to see how it is in fact possible to
prove that a problem is NP-hard. As discussed earlier, logic circuits
made from AND, OR, and NOT gates form a universal model of
computation, equal in power (up to polynomial factors) to the Turing
machine model. (See appendix \ref{classical_universality}.) Thus the
verifier for a problem in NP can be constructed as a logic circuit
from such gates. Such a logic circuit corresponds directly to a
Boolean formula made from AND, OR, and NOT. This formula will be
satisfiable if and only if there exists some input (the witness) which
causes the verifier to accept. Thus we have proven that Boolean
satisfiability is NP-hard. Given this fact, one can then prove the
NP-hardness of other problems by reductions of Boolean satisfiability
to other problems. Boolean satisfiability has the property that it is
both contained in NP and it is NP-hard. Such problems are called 
NP-complete. In a well-defined sense, NP-complete problems are the
hardest problems in NP. Furthermore, if one specifies a problem and
says it is complete for class X, then that statement uniquely defines
complexity class X.

Boolean satisfiability is not the only NP-complete problem. In fact,
there are now hundreds of NP-complete problems known. (See
\cite{Garey_Johnson} for a partial catalog of these.) Remarkably,
experience has shown that if a well-defined computational problem
resists all attempts to find polynomial time classical solution, it
almost always turns out to be NP-hard. There are only a few problems
currently known which are believed to be neither in P nor
NP-hard. These include factoring, discrete logarithm, graph
isomorphism, and approximating the shortest vector in a lattice. If a
problem is NP-hard, this is taken as evidence that the problem is not
solvable in polynomial time. If it were, then all of NP would be
solvable in polynomial time. This is considered unlikely, becaue it
seems contrary to experience that verifying the solution to a problem is
fundamentally no harder than finding the solution. Furthermore, it
seems unlikely that all those hundreds of NP-complete problems really
do have polynomial-time solutions which were never discovered despite
tremendous effort by very smart people over long periods of time. On
the other hand, there is no proof that all of NP is not solvable in
polynomial time. This is the famous P vs. NP problem which, for
various reasons\footnote{In addition to the failed attempts by many
  smart people to find a proof that P $\neq$ NP, there are additional
  reasons to believe that finding a proof should be hard. Namely,
  theorems have now been proven which show that the most
  natural methods for proving whether P is equal to NP are irrefutably
  doomed from the start\cite{Razborov_Rudich}.} is thought to be very
difficult.

NP is not the only complexity class based on a non-realistic model of
computation. Another important class is coNP. This is the set of
problems which have witnesses for the no instances. In other words,
these problems are the complements of the problems in NP. NP and coNP
overlap but are believed to be distinct. The problem of factoring
integers is known to be contained in both NP and coNP. This is one
reason factoring is not believed to be NP-complete. If it were then NP
would be contained in coNP. Graph isomorphism is also suspected to be contained
in the intersection of NP and coNP. MA is the probabilistic version of
NP, where the verifier is a BPP machine rather than a P
machine. PSPACE is the set of problems solvable using polynomial
memory. Polynomial space is a very powerful model of computation. The
class PSPACE contains both NP and coNP and is believed to be strictly
larger than either. \#P is like NP except to answer a \#P problem one
must count the number of witnesses rather than just answering
whether any witnesses exist. \#P is therefore not a decision class. To
make comparisons between \#P and decision classes one often uses
$P^{\#P}$, which is the set of problems solvable by a polynomial time
machine with access to an ``oracle'' which at any timestep can be queried
to solve a \#P problem. Many more complexity classes have been defined
(see \cite{complexity_zoo}). However the ones described above will
suffice for this thesis.

As is apparent from the preceeding discussion, many
complexity-theoretic results are founded on widely accepted
conjectures, such as the conjecture that P is not equal to NP. This is
perhaps an unfamiliar situation. These conjectures are neither proven
mathematical facts, nor are they the familiar sort of empirical facts
based on physical experiments. They are instead empirical facts based
on mathematical evidence. How can one assign probability of
correctness to mathematical conjectures? Does it even make sense to do
so\footnote{To give a more specific example, suppose
  you conjectured that P $\neq$ NP. Then you proposed various
  polynomial time algorithms for NP-hard problems. Whether each of
  these algorithms work depends on various calculations the result of
  which are not obvious \emph{a priori}. Upon performing the calculations,
  one finds in every case that they come out in just such a way that
  the polynomial-time algorithms for the NP-hard problems fail. Can
  one somehow use Bayesian reasoning in this case, regarding the calculations
  as experiments and their outcomes as evidence in favor of the
  conjecture P $\neq$ NP?}? These are interesting philosophical
questions, but to my knowledge unresolved ones. In any case, they are
beyond the scope of this thesis. In practice the conjecture that P is
not equal to NP is almost universally believed by the relevant
experts. Many other similar complexity-theoretic conjectures are often
also considered be well-founded, although not necessarily as much so
as P $\neq$ NP.

In the presence of all this conjecturing, it is worth mentioning that
some relationships between complexity classes are known with
certainty. One thing that is known in general is that the class defined
by a space bound of $n$ is contained in the class defined by a time
bound of $2^n$. This is because any algorithm running for time longer
than $2^n$ with only $n$ bits of memory necessarily revisits a state
it has already been in, and is therefore in an infinite loop. Thus any
problem solvable in logarithmic space is solvable in polynomial time,
and any problem solvable in polynomial space is solvable in
exponential time. In general, containments are easier to prove the
separations. For example, it is trivial to show that P is contained in
NP, but nobody has ever succeeded in showing that NP is larger than
P. An exception to this is that separations are not hard to prove
between classes of the same \emph{type}. For example, it is proven
that exponential time (EXP) is a strictly larger class than polynomial
time (P), and polynomial space (PSPACE) is a strictly larger class
than logarithmic space (L). In fact, it is even possible to prove that
there exist some problems solvable in time $O(n^4)$ not solvable in
time $O(n^2)$. This is done using a method called
diagonalization\cite{Papadimitriou}. However, the argument is
essentially non-constructive, and it is generally not known how to prove
unconditional lower bounds on the amount of time needed to solve a
given problem.

\begin{figure}
\begin{center}
\includegraphics[width=0.6\textwidth]{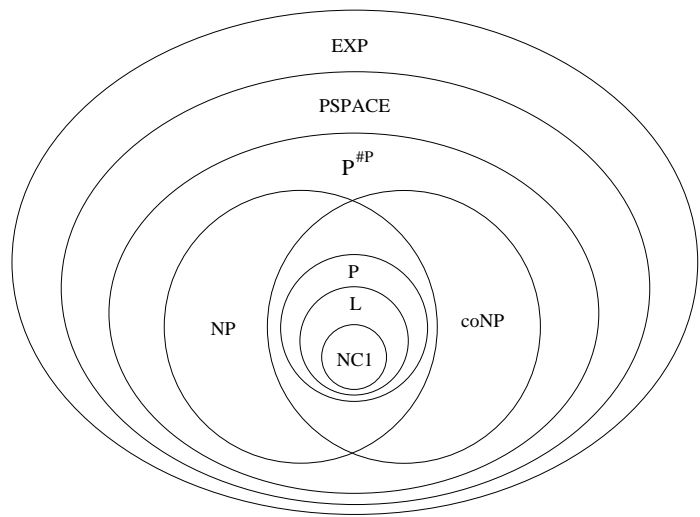}
\caption{ This diagram summarizes known and conjectured relationships
  between the classical complexity classes discussed in this
  section. All of the containments shown have been
  proven. However, none of the containments have been proven strict
  other than P is a strict subset of EXP and  L is a strict subset of
  PSPACE.}
\end{center}
\end{figure}

All of the complexity theory described so far has been about problems
where the input is given as a string of bits. However, one can also
imagine providing the input in the form of an oracle. An oracle is a
subroutine whose code is hidden. One then computes some property of
the oracle by making queries to it. For example, the oracle might
implement some function $f: \{1,2,\ldots,m \} \to \{1,2,\ldots,n\}$,
and we want to compute $\sum_{x=1}^m f(x)$. We can do this by querying
the oracle $m$ times, once for each value of $x$, and summing up the
results. It is also clear that this cannot be done by querying the
oracle fewer than $m$ times. This demonstrates a very nice feature of
the oracular setting, which is that it is often possible to prove
lower bounds on the number of queries necessary for computing a given
property.

The oracular model of computation is artificial, in the sense that we
have artificially prohibited access to the code implementing the
oracle. However, in many settings it seems unlikely that examining the
source code would help. Even simple functions that can be written down
using a small number of algebraic symbols often lack analytical
antiderivatives, and to find the definite integral there seems to be
nothing better to do than evaluate the function at a series of points
and use the trapezoid rule or other similar techniques. This is exactly the
oracular case. Similarly, Newton's method for finding roots, and
gradient descent methods for finding minima are both oracular
algorithms. If the function is implemented by some large and
complicated numerical calculation then it seems even more likely that
for finding integrals, derivatives, extrema, and so on, there is
nothing better to be done than simply querying the function at
various points and performing computations with the resulting
data. For these reasons, and because query complexity is much more
easily analyzed than computational complexity, the oracular setting is
an important area of study in both classical and quantum computation.

\section{Quantum Computation Preliminaries}
\label{quantum_prelim}

Because this is a physics thesis, I'll assume familiarity with
quantum mechanics. Many standard books exist on the subject
\cite{Cohen-Tannoudji, Messiah, Sakurai, Griffiths}. However, the
emphasis in these books is not necessarily placed on the aspects of
quantum mechanics which are most necessary for quantum computing. A
nice brief quantum-computing oriented introduction to quantum
mechanics is given in the second chapter of \cite{Nielsen_Chuang}.

To reason about quantum computers, one needs a mathematical model of
them. In fact, as I will argue in this thesis, it is helpful to have
several mathematical models of quantum computers. The most
widely used model of quantum computation is the quantum circuit model,
and I will now describe it. 

The first concept needed to define a quantum circuit is the
qubit. Physically, a qubit is a two state quantum mechanical
system, such as a spin-$1/2$ particle. As such, its state is given by
a normalized vector in $\mathbb{C}^2$. One normally imagines
doing quantum computation by performing unitary operations on an array of
qubits. One could of course use $d$-state systems with $d >
2$. Using $d$-dimensional units (called qudits) generally
results in only a speedup by a constant factor, which will not even be
noticed if one is using big-O notation. Since it makes no difference
algorithmically, people almost always choose the lowest dimensional
nontrivial systems for simplicity, and these are qubits. This is
analogous to the classical case. In addition to their physical
interpretation, qubits have meaning as the basic unit of quantum
information. This meaning arises from the study of quantum
communication, sometimes known as quantum Shannon theory. Quantum
Shannon theory will not be discussed in this thesis. For this see
\cite{Aramsthesis, Nielsen_Chuang}.

Next, we need some way of acting upon qubits. Upon thinking about the
Coulombic forces between charged particles, the gravitational forces
between massive objects, the interaction between magnetic dipoles, and
so forth, one sees that most interactions appearing in nature are
pairwise. That is, the total energy of a configuration of $n$
particles is of the form
\[
\sum_{i,j = 1}^n E_{ij}
\]
where $E_{ij}$ depends only on the states of particles $i$ and
$j$. This carries over into quantum mechanical systems. Thus,
one expects only to directly enact operations on single
qubits or pairs of qubits. As a simple model of quantum
computation, one may suppose that one can apply arbitrary unitary
operations on individual qubits and pairs of qubits. A quantum
computation then consists of a polynomially long sequence of such
operations. From an algorithmic point of view this is considered to be
a perfectly acceptable definition of a quantum computer. The
individual one-qubit and two-qubit unitaries are called gates, by
analogy to the classical logic gates. The entire sequence of unitaries
is called a quantum circuit.

In the quantum circuit model, the input to the computation (the
problem instance) is the initial state of the qubits prior to being
acted upon by the series of unitaries. Since human minds are
apparently classical, the problems we wish to solve are
classical. Thus, we will only consider problems whose inputs and
outputs are classical bitstrings. We can choose two orthogonal states
of a given qubit as corresponding to classical 0 and 1. These states
form a basis for the Hilbert space of the qubit, known as the
computational basis. The computational basis states of a qubit are
conventionally labelled $\ket{0}$ and $\ket{1}$. The $2^n$ states
obtained by putting each qubit into $\ket{0}$ or $\ket{1}$ form the
computational basis basis for the $2^n$-dimensional Hilbert space of
the entire system. Rather than labelling these states by
\[
\begin{array}{c}
\ket{0} \otimes \ket{0} \otimes \ldots \otimes \ket{0} \\
\ket{0} \otimes \ket{0} \otimes \ldots \otimes \ket{1} \\
\vdots \\
\ket{1} \otimes \ket{1} \otimes \ldots \otimes \ket{1}
\end{array}
\]
it is conventional to simply write them as
\[
\begin{array}{c}
\ket{00\ldots0} \\
\ket{00\ldots1} \\
\vdots \\
\ket{11\ldots1}
\end{array}.
\]
The input to the computation is the computational basis state
corresponding to the classical bitstring which specifies the problem
instance. The output of the computation is the result of a measurement
in the computational basis.

This is not a matter of mere notation. Arbitrary quantum states are
hard to produce, and measurements in arbitrary bases are hard to
perform. The special feature of the computational basis is that it is
a basis if tensor product states, that is, in every computational
basis state, the qubits are completely unentangled. Such states are
easy to generate, since the qubits need only be put into their states
individually without interacting them. Similarly, the measurement at
the end can be performed by measuring the qubits one by one.

\begin{figure}
\begin{center}
\includegraphics[width=0.6\textwidth]{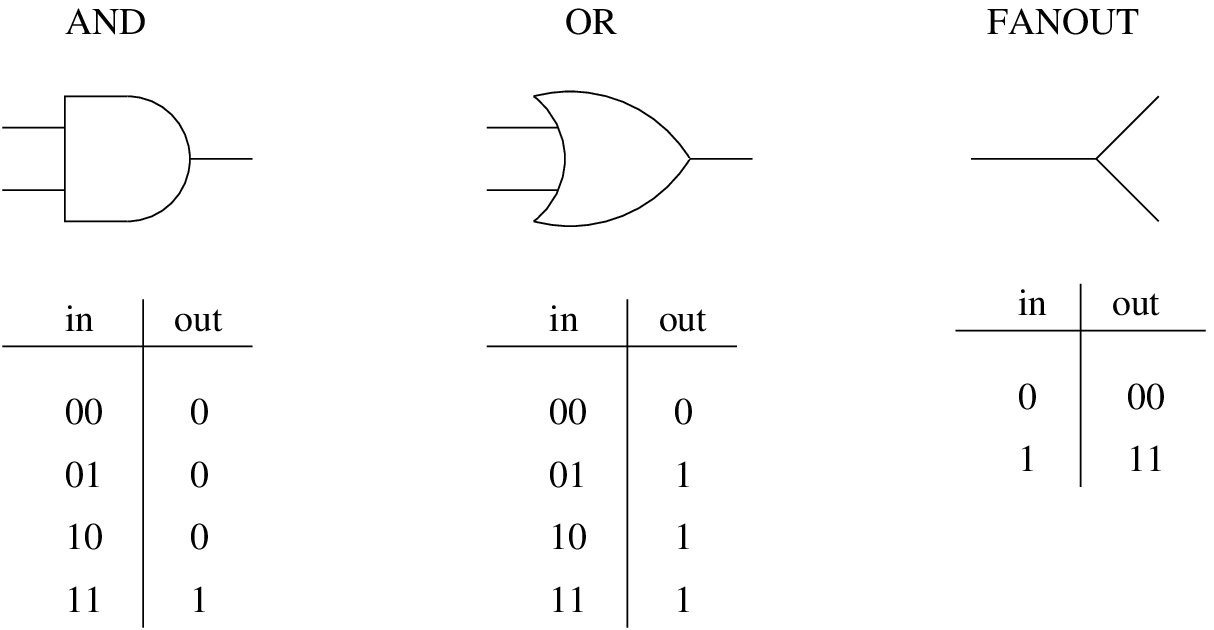}
\caption{\label{bools} These gates are universal for classical
  computation. Below each gate is the corresponding ``truth table''
  giving the dependence of output on input.}
\end{center}
\end{figure}

Classically, any Boolean function can be constructed using only
AND, NOT, and FANOUT, as shown in figure \ref{bools}. Thus, this set
of gates are said to be universal. Similarly, the set of two-qubit
quantum gates is universal in the sense that any unitary on $n$-qubits
can be constructed as a product of such gates. In general, this can
require exponentially many gates, much like the classical case. The
proof that two-qubit gates are universal is given in detail in
\cite{Nielsen_Chuang}, so I will only sketch it here. The approach is
to first show that the set of two level unitaries are universal. A
two-level unitary is one which interacts two basis states unitarily
and leaves all other states untouched, as shown below.
\[
\left[ \begin{array}{cccccccc}
1 &        &        &   &        &    &        &  \\
  & \ddots &        &   &        &    &        &  \\
  &        & u_{11} &   & u_{12} &    &        &  \\
  &        &        & 1 &        &    &        &  \\
  &        & u_{21} &   & u_{22} &    &        &  \\
  &        &        &   &        & 1  &        &  \\
  &        &        &   &        &    & \ddots &  \\
  &        &        &   &        &    &        & 1
\end{array}
\right]
\]
Given any arbitrary $2^n \times 2^n$ unitary, one can left-multiply by
a sequence of two-level unitaries to eliminate off-diagonal matrix
elements one by one. This process is somewhat analogous to Gaussian
elimination. At the same time, one can also ensure that the remaining
diagonal elements are all equal to 1. That is, for any unitary $U$ on
$n$ qubits, there is some sequence of two-level unitaries $U_m U_{m-1}
\ldots U_2 U_1$ such that
\[
U_m U_{m-1} \ldots U_2 U_1 U = I.
\]
Thus, for any $U$, there is a product of two-level unitaries equal to
$U^{-1}$, which shows that two-level unitaries are universal. 

For any given pair of basis states $\ket{x},\ket{y}$, one can construct
the two level unitary that acts on them according to 
\[
U = \left[ \begin{array}{cc}
u_{11} & u_{12} \\
u_{21} & u_{22} \\
\end{array} \right]
\]
by conjugating the single-qubit gate for $U$ with a matrix $U_{\pi}$
that permutes the basis so that 
$U_{\pi} \ket{x}$ and $U_{\pi} \ket{y}$ differ on a single bit. It is
a simple exercise to show that $U_{\pi}$ can always be constructed by
a sequence of controlled-not (CNOT) gates. CNOT is a two-qubit gate
that act on two-qubits according to:
\begin{equation}
\label{CNOT_def}
\left[ \begin{array}{cccc}
1 & 0 & 0 & 0 \\
0 & 1 & 0 & 0 \\
0 & 0 & 0 & 1 \\
0 & 0 & 1 & 0
\end{array} \right] 
\begin{array}{c}
00 \\
01 \\
10 \\
11
\end{array}.
\end{equation}
The bitstrings on the right label the four computational basis states
of the two qubits. The controlled-not gets its name from the fact that
a NOT gate is applied to the second bit (the target bit) only if the
first bit (the control bit) is 1.

Although this gate universality result is a very nice first step, it
is still not fully satisfying. The set of two-qubit gates ($4 \times
4$ unitary matrices) forms a continuum. An infinite number of bits
would be necessary to exactly specify particular gate. The same goes
for one-qubit gates ($2 \times 2$ unitary matrices). However, this is
a surmountable problem. The reason is that small deviations from the
desired gate will cause only small probability of error in the final
measurement. This is because the deviations from the desired state
caused by each gate add at most linearly,  which we no show, following
\cite{Nielsen_Chuang}. 

Suppose we wish to perform the gate $V$ followed by the gate
$U$. In reality we perform imprecise versions of these, $V'$ followed by
$U'$. We'll quantify the error introduced by the imprecise gates by
\[
E(U'V') \equiv \max_{\braket{\psi}{\psi} = 1} \| U'V' \ket{\psi} - UV
\ket{\psi} \|,
\]
which equals
\[
= \max_{\braket{\psi}{\psi} = 1} \| (U V \ket{\psi} - U V' \ket{\psi})
+ (U V' \ket{\psi} - U' V' \ket{\psi}) \|.
\]
By the triangle inequality this is at most
\[
\leq \max_{\braket{\psi}{\psi} = 1} \| U V \ket{\psi} - U V' \ket{\psi} \|
+ \| U V' \ket{\psi} - U' V' \ket{\psi} \|.
\]
By unitarity, this is at most
\[
\leq \max_{\braket{\psi}{\psi} = 1} \| V \ket{\psi} - V' \ket{\psi} \| +
\max_{\braket{\phi}{\phi} = 1} \| U \ket{\phi} - U' \ket{\phi} \|
\]
\[
= E(V') + E(U').
\]
Thus
\begin{equation}
\label{error_add}
E(U'V') \leq E(V') + E(U').
\end{equation}

By equation \ref{error_add}, one sees that it is not necessary to
obtain higher than polynomial accuracy in the gates in order to
implement quantum circuits of polynomial size. Hence only
logarithmically many bits are needed to specify a gate. This result
can be improved upon in two ways. First, it turns out that it is
unnecessary to have even a polynomially large set of gates. Instead,
arbitrary one and two qubit gates can always be constructed with
polynomial accuracy using a sequence of logarithmically many gates
chosen from some finite set of \emph{universal} quantum
gates. This result is known as the Solovay-Kitaev theorem, which we
state formally below. Universal sets of quantum gates are known with
as few as two gates. Secondly, the fault tolerance threshold theorem
shows (among other things) that it is in fact unnecessary to achieve
higher than constant accuracy in implementing each gate. Fault
tolerance thresholds are discussed in section \ref{FT}.

The following is a formal statement of the Solovay-Kitaev theorem
adapted from\cite{Kitaev}.
\begin{theorem}[Solovay-Kitaev]
Suppose matrices $U_1, \ldots, U_r$ generate a dense subgroup in
$SU(d)$. Then, given a desired unitary $U \in SU(d)$, and a precision
parameter $\delta > 0$, there is an algorithm to find a product $V$ of
$U_1, \ldots, U_r$ and their inverses such that $\| V - U \| \leq
\delta$. The length of the product and the runtime of the algorithm
are both polynomial in $\log(1/\delta)$.
\end{theorem}
Combining this with the universality of two-qubit unitaries, one sees
that any set of one-qubit and two-qubit gates that generates a dense
subgroup of $SU(4)$ is universal for quantum computation. A convenient
universal set of quantum gates is the CNOT, Hadamard, and $\pi/8$
gates. The CNOT gate we have encountered already in equation
\ref{CNOT_def}. The Hadamard gate is
\[
H = \frac{1}{\sqrt{2}} \left[ \begin{array}{cc}
1 & 1 \\
1 & -1
\end{array} \right],
\]
and the $\pi/8$ gate is
\[
T = \left[ \begin{array}{cc}
e^{-i \pi/8} & 0 \\
0            & e^{i \pi/8}
\end{array} \right]. 
\]

Although two-qubit gates are universal, it does not follow that
arbitrary unitaries can be constructed efficiently from two-qubit
gates. In fact, even the set of $2^n \times 2^n$ permutation matrices
(corresponding to reversible computations) is doubly exponentially
large, whereas the set of polynomial size quantum circuits is only
singly exponentially large, given any discrete set of gates. Thus,
some unitaries on $n$-qubits require exponentially many gates to
construct as a function of $n$. 

We have now seen that using a discrete set of quantum gates we can
construct arbitrary unitaries, although some $n$-qubit unitaries
require exponentially many gates. This is in some sense a universality
result. However, what we are really interested in is computational
universality. At present it is not yet obvious that one can even
efficiently perform universal classical computation with such a set of
gates. However, it turns out that this is indeed possible. It would be
surprising if this were not possible, since classical physics, upon
which classical computaters are based, is a limiting case of quantum
physics. Nevertheless showing how to specifically implement classical
computation with a quantum circuit is not trivial. The essential
difficulty is that the standard sets of universal classical gates
include gates which lose information. For example, the AND gate takes
two bits of input and produces only a single bit of output. There is
no way of deducing what the input was just by reading the output. In
contrast, the quantum mechanical time evolution of a closed system is
unitary and therefore never loses any information.

The solution to this conundrum actually predates the field of quantum
computation and goes by the name of reversible circuits. It turns out
that universal classical computation can be achieved using gates that
have the same number of output bits as input bits, and which
furthermore never lose any information. That is, the map of inputs to
outputs is injective. These are called reversible gates. The CNOT gate
described in equation \ref{CNOT_def} is an example of a classical
reversible gate which has truth table
\[
\begin{array}{c}
00 \to 00 \\
01 \to 01 \\
10 \to 11 \\
11 \to 10
\end{array}.
\]
By itself, CNOT is not universal. However, the Fredkin gate, or
controlled SWAP is. This gate has the truth table
\[
\begin{array}{c}
000 \to 000 \\
001 \to 001 \\
010 \to 010 \\
011 \to 011 \\
100 \to 100 \\
101 \to 110 \\
110 \to 101 \\
111 \to 111
\end{array}.
\]
The second pair of bits are swapped only if the first bit is 1. As
shown in figure \ref{Fredkin}, AND, NOT, and FANOUT can all be
implemented using the Fredkin gate. In standard non-reversible
classical circuits one normally takes FANOUT for granted,
considering it to be achieved by splitting a wire. In the context of
reversible computing one must be more careful. The FANOUT operation
requires the use of an additional bit initialized to the 0 state to
take the copied value of the bit undergoing FANOUT. In fact, each of
the constructions shown in figure \ref{Fredkin} require initialized
work bits, known as ancilla bits. This is a generic feature of
reversible computation because ``garbage'' bits cannot be erased and
instead are simply carried to the end of the computation.

\begin{figure}
\begin{center}
\includegraphics[width=0.5\textwidth]{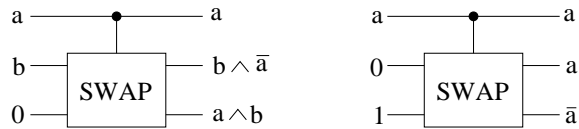}
\caption{\label{Fredkin} The left circuit uses a controllod-SWAP
  (\emph{i.e.} a Fredkin gate) to achieve AND using one ancilla bit
  initialized to zero. The right circuit uses a Fredkin gate to
  achieve NOT and FANOUT using two ancilla bits initialized to zero
  and one. AND, NOT, and FANOUT are universal for classical
  computation, thus classical computation can be performed reversibly
  using Fredkin gates and ancilla bits.}
\end{center}
\end{figure}

Because AND, NOT, and FANOUT can each be constructed from a single
Fredkin gate, it follows that taking classical circuits and making
them reversible incurs only constant overhead. Thus, the set of
problems solvable in polynomial time on uniform families of reversible
circuits is exactly P. On a quantum computer, a reversible 3-qubit
gate such as the Fredkin gate corresponds to an 3-qubit quantum gate
which is an $8 \times 8$ permutation matrix, permuting the basis
states in accordance with the gate's truth table. Hence reversible
computation is efficiently achievable on quantum computers. Because of
this generic construction, current research on quantum algorithms
focuses on quantum algorithms which beat the best classical
algorithms. Quantum algorithms matching the performance of classical
algorithms can always be achieved using reversible circuits.

For the purpose of quantum computation it is often important to remove
the garbage qubits accumulated at the end of a reversible computation,
because these can destroy the interference needed in quantum
algorithms. It is always possible to remove the garbage bits by first
performing the reversible computation, then using CNOT gates to copy
the result into a register of ancilla bits initialized to zero, and
then reversing the computation, as illustrated in figure
\ref{uncomputation}. The process of reversing the computation is known
as uncomputation.

\begin{figure}
\begin{center}
\includegraphics[width=0.9\textwidth]{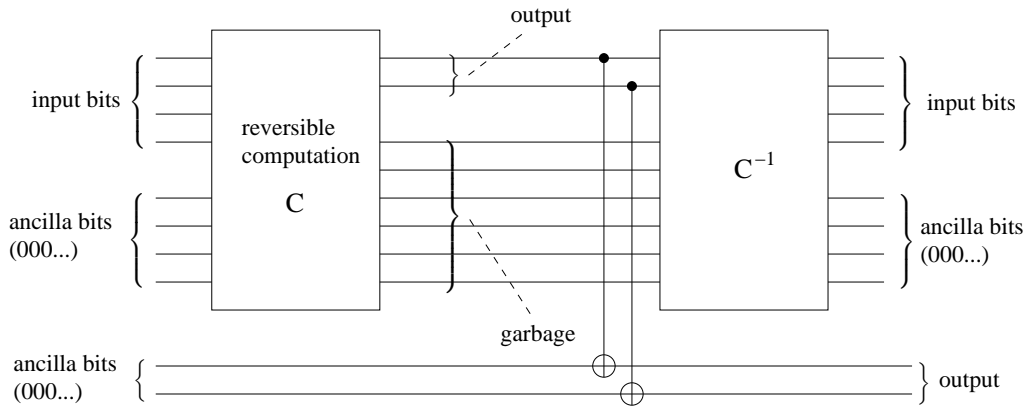}
\caption{\label{uncomputation} Garbage bits can be reset to zero. This
  is done by first performing the computation, then copying the output
  into an ancilla register, then using the inverse computation to
  ``uncompute'' the garbage bits.}
\end{center}
\end{figure}

The quantum circuit model is used as the standard definition of
quantum computers. The class of problems solvable in polynomial time with
quantum circuits is called BQP, which stands for Bounded-error
Quantum Polynomial-time. The initial state given to the quantum
circuit must be a computational basis state corresponding to a
bitstring encoding the problem instance, plus optionally a supply of
polynomially many ancilla qubits initialized to $\ket{0}$. The output
is obtained by measuring a single qubit in the computational
basis. BQP is a class of decision problems, and the measurement
outcome is considered to be yes or no depending on whether the
measurement yields one or zero. A decision problem belongs to BQP 
if there exists a uniform family of quantum circuits whose number of
gates scales polynomially with the input size $n$, such that the
output is correct with probability at least 2/3 for every problem
instance. BQP is thus the quantum analogue of BPP. 

A family of quantum circuits is considered to be uniform if the
circuit for any given $n$ can be generated in $\mathrm{poly}(n)$ time
by a classical computer. Allowing the family of circuits to be
generated by a quantum computer does not increase the power of the
model. This is a consequence of the principle of deferred measurement,
as discussed in appendix \ref{deferred}.

Because probabilities arise naturally in quantum mechanics, most
studies of quantum computation focus on probabilistic computations and
complexity classes. Deterministic quantum computation can certainly be
defined, and some quantum algorithms succeed with probability one
while still achieving a speedup over classical
computation\footnote{For example, the Bernstein-Vazirani algorithm
  achieves this.}. However, restricting to deterministic quantum
algorithms seems somewhat artificial. Most of the literature on
quantum algorithms and complexity assumes the probabilistic setting by
default, as does this thesis.

Recall that MA is the probabilistic version of NP. That is, it is the
class of problems whose YES instances have probabilistically
verifiable witnesses. There are two quantum analogues to MA, depending
on whether the witnesses are classical or quantum. The set of decision
problems whose solutions are efficiently verifiable on a quantum
computer given a classical bitstring as a witness is called QCMA. The
set of decision problems whose solutions are efficiently verifiable on
a quantum computer given a quantum state as a witness is called
QMA. Many important physical problems are now known to be QMA
complete, such as computing the ground state energy of arbitrary
Hamiltonians made from two-body interactions\cite{Kempe}, and
determining the consistency of a set of density
matrices\cite{Liu}. One can also define space bounded quantum
computation. BQPSPACE is the class of problems solvable with bounded
error on a quantum computer with polynomial space and unlimited
time. Perhaps surprisingly, BQPSPACE = PSPACE
\cite{Watrous_space}. (As an aside, NPSPACE = PSPACE \cite{Savitch}!)

The class of problems solvable by logarithmic depth quantum circuits
is called BQNC1. This class is potentially relevant for physical
implementation of quantum computers because if quantum gates can be
performed in parallel, then the BQNC1 computations can be carried out in
logarithmic time. This greatly reduces the time one needs to
maintain the coherence of the qubits. Interestingly, an approximate
quantum Fourier transform can be done using a logarithmic depth
quantum circuit. As a result, factoring can be done with polynomially
many uses of logarithmic depth quantum circuits, followed by a
polynomial amount of classical
postprocessing\cite{Cleve_parallel}.  

As mentioned previously, it is easy to see that problems solvable in
classical space $f(n)$ are solvable in classical time $2^{f(n)}$
ecause there are only $2^{f(n)}$ states that the computer can be
in. Thus, after $2^{f(n)}$ steps the computer must reenter a previously used
state and repeat itself. Quantum mechanically the situation is
different. For any fixed $\epsilon \ll 1$, in a Hilbert space of
dimension $d$ one can fit exponentially many nonoverlapping patches of
size $\epsilon$ as a function of $d$. (We could define a patch of size
$\epsilon$ centered at $\ket{\psi}$ as $\{ \ket{\phi} : \| \ket{\phi}
- \ket{\psi} \| < \epsilon \}$.) Thus there are doubly exponentially
many  reasonably distinct states of $n$ qubits. Hence there is not an
analogous argument to show that problems solvable in quantum space
$f(n)$ are solvable in quantum time
$\mathrm{exp}(f(n))$. Nevertheless, this statement is true. It can be
proven using the previously described universality construction based
on two level unitaries. Working through the construction in detail one
finds that any $2^n \times 2^n$ unitary can be constructed from
$O(2^{2n})$ two level unitaries, and any 2-level unitary on the
Hilbert space of $n$ qubits can be achieved using $O(n^2)$ CNOT gates
plus one arbitrary single-qubit gate. Thus no computation on
$n$-qubits can require more than $O(2^{2n} n^2)$ gates. (However,
finding the appropriate gate sequence may be difficult.)

\section{Quantum Algorithms}
\label{algorithms}

\subsection{Introduction}

By now it is well-known that quantum computers can
solve certain problems much faster than the best known classical
algorithms. The most famous example is that quantum computers can
factor $n$-bit numbers in time polynomial in $n$\cite{Shor_factoring},
whereas no known classical algorithm can do this. The quantum
algorithm which achieves this is known as Shor's factoring
algorithm. As discussed in section \ref{quantum_prelim}, a quantum
algorithm can be defined as a uniform family of quantum circuits, and
the running time is the number of gates as a function of number of
bits of input.

Quantum algorithms can be categorized into two types based on the
method by which the problem instance is given to the quantum
computer. The most most obvious and fundamental way to provide the
input is as a bitstring. This is how the input is provided to the factoring
algorithm. The second way of providing the input to a quantum
algorithm is through an oracle. The oracular setting is very much
analogous to the classical oracular setting, with the additional
restriction  that the oracle must be unitary. Any classical oracle can
be made unitary by the general technique of reversible computation, as
shown in figure \ref{reversible_oracle}.

\begin{figure}
\begin{center}
\includegraphics[width=0.6\textwidth]{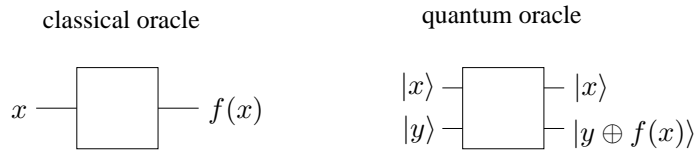}
\caption{\label{reversible_oracle} Quantum oracles must be
  unitary. One can always achieve this by using separate input and
  output registers. The input register is left unchanged and the
  output is added into the output register bitwise modulo 2. If the
  input register $y$ is initialized $0000\ldots$, than after applying
  the oracle it will contain $f(x)$.}
\end{center}
\end{figure}

The second most famous quantum algorithm is oracular. The oracle
implements the function $f:\{0,1,\ldots,N\} \to \{0,1\}$ defined by
\[
f(x) = \left\{ 
\begin{array}{ll}
1 & \textrm{if $x = w$} \\
0 & \textrm{otherwise}
\end{array}
\right.
\]
The task is to find the ``winner'' $w$. Classically, the only way to
do this with guaranteed success is to query all $N$ values of
$x$. Even on average, one needs to query $N/2$ values. On a quantum 
computer this can be achieved using $O(\sqrt{N})$
queries\cite{Grover_search}. The algorithm which achieves this is
known as Grover's searching algorithm. The queries made to the oracle
are superpositions of multiple inputs. Quantum computers cannot solve
this problem using fewer than $\Omega(\sqrt{N})$
queries\cite{BBBV}. Brute-force searching is a common subroutine
in classical algorithms. Thus, many classical algorithms can be sped
up by using Grover search as a subroutine. Furthermore, quantum
algorithms achieve quadratic speedups for searching in the presence
of more than one winner\cite{BBHT98}, evaluating sum of an
arbitrary function\cite{BBHT98,BHT98,Mos98}, finding the global
minimum of arbitrary function\cite{DH96,NW99}, and approximating definite
integrals\cite{integration}. These algorithms are based on Grover's
search algorithm.

From a complexity point of view, a quantum algorithm provides an upper
bound on the quantum complexity of a given problem. It is also
interesting to look for lower bounds on the quantum complexity
problems, or in other words upper limits on the power of quantum
computers. The techniques for doing so are very different in the
oracular versus nonoracular settings.

In the oracular setting, several powerful methods are known for
proving lower bounds on the quantum query complexity of
problems\cite{Ambainis_adversary, BBCMW98}. The $\Omega(\sqrt{N})$
lower bound for searching is one example of this. For some oracular
problems it is proven that quantum computers do not offer any speedup
over classical computers beyond a constant factor. For example,
suppose we are given an oracle computing an arbitrary function of the
form $f:\{0,1,\ldots,N\} \to \{0,1\}$, and we wish to compute the
parity 
\[
\bigoplus_{x=1}^N f(x).
\]
Both quantum and classical computers need $\Omega(N)$ queries to
achive this\cite{parity}.

In the non-oracular setting there are essentially\footnote{One can
  prove very weak statements such as the fact that most problems
  cannot be solved in less than the time it takes to read the entire
  input. Also certain extremely difficult problems, such as optimally
  playing generalized chess, are EXP-complete. These problems provably
  are not in P.}  no known techniques for proving lower bounds on quantum (or
classical) complexity. One can see however that quantum computers
cannot achieve superexponential speedups over classical computers
because they can be classically simulated with exponential
overhead. Extending this reasoning, it is clear that one could show by
diagonalization\cite{Papadimitriou} that for any polynomial $p(n)$
there exist problems in EXP which cannot be solved on a quantum
computer in time less than $p(n)$. A different type of upper bound on
the power of quantum computers is that $\textrm{BQP} \in P^{\# P}$, as
shown in\cite{Bernstein_Vazirani}.

Arguably the most important class of known quantum algorithms from a
practical point of view are those for quantum simulation. The problem
of simulating quantum systems has great economic and scientific
significance. Many problems, such as the design of new drugs and
materials, and understanding condensed matter systems such as high
temperature superconductors, would likely be much easier if quantum
many-body systems could be efficiently simulated. It seems that this
cannot be done on classical computers because the dimension of the
Hilbert space grows exponentially with the number of degrees of
freedom. Thus, even writing down the wavefunction would require
exponential resources.

In contrast to classical computers, it is generally believed that
standard quantum computers can efficiently simulate all
nonrelativistic quantum systems. That is, the number of gates and
number of qubits needed to simulate a system of $n$ particles for time
$t$ should scale polynomially in $n$ and $t$. The essential reason for
this is that Hamiltonians arising in nature generally consist of
few-body interactions. Few-body interactions can be simulated using
few-body quantum gates via the Trotter formula. The exact form of the
few-body interactions is irrelevant due to gate
universality. Furthermore, even if a Hamiltonian is not a sum of
few-body terms, it can still be efficiently simulated provided that
each row of the matrix has at most polynomially many nonzero entries
and these entries can be computed efficiently. Methods for quantum
simulation are described in
\cite{Feynman, Childs_thesis, Zalka_sim, Wiesner_sim, Aharonov_Tashma,
  Abrams_sim, Cleve_sim, Kassal_sim, Lidar_sim}. If a physical
system were discovered that could not be simulated in polynomial time
by a quantum computer, and that systems could be reliably
controlled, then it could presumably be used to construct a
computer more powerful standard quantum computers. Currently, it is
not fully known whether relativistic quantum field theory can be
efficiently simulated by quantum computers. In fact, the task of
formulating a well-defined mathematical theory of computation based on
quantum field theory appears to be difficult.

Not all quantities that arise in the study of physics
are easily computable using quantum computers. For example, finding
the ground energy of an arbitrary local Hamiltonian is
QMA-hard\cite{Kempe}, and evaluating the partition function of
the classical Potts model is \#P-hard\footnote{The Potts model
  partition function is a special case of the Tutte polynomial, as
  discussed in \cite{Aharonov3}. It was shown in \cite{Jaeger}
  that exact evaluation of the Tutte polynomial at all but a few
  points is \#P-hard.}. Therefore it is unlikely that these problems
can be solved in general on a quantum computer in polynomial
time. It is perhaps not surprising that some partition functions
cannot be efficiently evaluated, because partition functions are not
directly measurable by physical means, and thus not computable by the
simulation of a physical process. In contrast, information about the
eigenenergies of physical systems can be measured by spectroscopy. The
problem is, for some systems, the time needed to cool them into the
ground state may be extremely long. Correspondingly, on a quantum
computer, the energy of a given eigenstate can be efficiently
determined to polynomial precision by the method of phase
estimation (see appendix \ref{phase_estimation}), but there may be no
efficient method to prepare the ground state.

Several other quantum algorithms are known. A list of known quantum
algorithms is given below. I have attempted to be comprehensive,
although there are probably a few oversights. By known results
regarding reversible computation, any classical algorithm can be
implemented on a quantum computer with only constant overhead. Thus, I
only list quantum algorithms achieving a speedup over the fastest
known classical algorithm. Furthermore, any quantum circuit solves the
problem of computing its own output. Thus to keep the list meaningful,
I include only quantum algorithms achieving a speedup for a problem
that could have been stated prior to the concept of quantum
computation (although not all of these problems necessarily
were). Most quantum algorithms in the literature meet this
criterion.\\ \\

\subsection{Algebraic and Number Theoretic Problems}

\noindent
\begin{minipage}[c]{\textwidth}
\noindent
\textbf{Algorithm:} Factoring \\
\textbf{Type:} Non-oracular \\
\textbf{Speedup:} Superpolynomial \\
\textbf{Description:} Given an $n$-bit integer, find the prime
factorization. The quantum algorithm of Peter Shor solves this in
$\mathrm{poly}(n)$ time\cite{Shor_factoring}. The fastest known classical 
algorithm requires time superpolynomial in $n$. This algorithm breaks the
RSA cryptosystem. At the core of this algorithm is order finding,
which can be reduced to the Abelian hidden subgroup problem. \\ \\
\end{minipage}

\noindent
\begin{minipage}[c]{\textwidth}
\noindent
\textbf{Algorithm:} Discrete-log \\
\textbf{Type:} Non-oracular \\
\textbf{Speedup:} Superpolynomial \\
\textbf{Description:} We are given three $n$-bit numbers $a$, $b$, and
$N$, with the promise that $b = a^s \mod N$ for some $s$. The task is
to find $s$. As shown by Shor\cite{Shor_factoring}, this can be achieved
on a quantum computer in $\mathrm{poly}(n)$ time. The fastest known
classical algorithm requires time superpolynomial in $n$. See also Abelian
hidden subgroup.
\\ \\
\end{minipage}

\noindent
\begin{minipage}[c]{\textwidth}
\noindent
\textbf{Algorithm:} Pell's Equation \\
\textbf{Type:} Non-oracular \\
\textbf{Speedup:} Superpolynomial \\
\textbf{Description:} Given a positive nonsquare integer $d$, Pell's
equation is $x^2 - d y^2 = 1$. For any such $d$ there are infinitely
many pairs of integers $(x,y)$ solving this equation. Let $(x_1,y_1)$
be the pair that minimizes $x+y\sqrt{d}$. If $d$ is an $n$-bit integer
(\emph{i.e.} $0 \leq d < 2^n$), then $(x_1,y_1)$ may in general require
exponentially many bits to write down. Thus it is in general
impossible to find $(x_1,y_1)$ in polynmial time. Let $R =
\log(x_1+y_1 \sqrt{d})$. $\lfloor R \rceil$ uniquely identifies $(x_1,y_1)$. As shown
by Hallgren\cite{Hallgren_Pell}, given a $n$-bit number $d$, a quantum
computer can find $\lfloor R \rceil$ in $\mathrm{poly}(n)$ time. No polynomial
time classical algorithm for this problem is known. Factoring reduces
to this problem. This algorithm breaks the Buchman-Williams
cryptosystem. See also Abelian hidden subgroup.
\\ \\
\end{minipage}

\noindent
\begin{minipage}[c]{\textwidth}
\noindent
\textbf{Algorithm:} Principal Ideal \\
\textbf{Type:} Non-oracular \\
\textbf{Speedup:} Superpolynomial \\
\textbf{Description:} We are given an $n$-bit integer $d$ and an
invertible ideal $I$ of the ring $\mathbb{Z}[\sqrt{d}]$. $I$ is a
principal ideal if there exists $\alpha \in \mathbb{Q}(\sqrt{d})$
such that $I = \alpha \mathbb{Z}[\sqrt{d}]$. $\alpha$ may be
exponentially large in $d$. Therefore $\alpha$ cannot in general even
be written down in polynomial time. However, $\lfloor \log \alpha 
\rceil$ uniquely identifies $\alpha$. The task is to determine whether
$I$ is principal and if so find $\lfloor \log \alpha \rceil$. As shown
by Hallgren, this can be done in polynomial time on a quantum
computer\cite{Hallgren_Pell}. Factoring reduces to solving Pell's
equation, which reduces to the principal ideal problem. Thus the
principal ideal problem is at least as hard as factoring and therefore
is probably not in P. See also Abelian hidden subgroup.
\\ \\
\end{minipage}

\noindent
\begin{minipage}[c]{\textwidth}
\noindent
\textbf{Algorithm:} Unit Group \\
\textbf{Type:} Non-oracular \\
\textbf{Speedup:} Superpolynomial \\
\textbf{Description:} The number field $\mathbb{Q}(\theta)$ is said to
be of degree $d$ if the lowest degree polynomial of which $\theta$ is
a root has degree $d$. The set $\mathcal{O}$ of elements of
$\mathbb{Q}(\theta)$ which are roots of monic polynomials in
$\mathbb{Z}[x]$ forms a ring, called the ring of integers of
$\mathbb{Q}(\theta)$. The set of units (invertible elements) of the
ring $\mathcal{O}$ form a group denoted $\mathcal{O}^*$. As shown by
Hallgren \cite{Hallgren_unit}, for any $\mathbb{Q}(\theta)$ of fixed
degree, a quantum computer can find in polynomial time a set of
generators for $\mathcal{O}^*$, given a description of $\theta$. No
polynomial time classical algorithm for this problem is known. See
also Abelian hidden subgroup.\\ \\
\end{minipage}

\noindent
\begin{minipage}[c]{\textwidth}
\noindent
\textbf{Algorithm:} Class Group \\
\textbf{Type:} Non-oracular \\
\textbf{Speedup:} Superpolynomial \\
\textbf{Description:} The number field $\mathbb{Q}(\theta)$ is said to
be of degree $d$ if the lowest degree polynomial of which $\theta$ is
a root has degree $d$. The set $\mathcal{O}$ of elements of
$\mathbb{Q}(\theta)$ which are roots of monic polynomials in
$\mathbb{Z}[x]$ forms a ring, called the ring of integers of
$\mathbb{Q}(\theta)$. For a ring, the ideals modulo the prime ideals
form a group called the class group. As shown by
Hallgren\cite{Hallgren_unit}, a quantum computer can find in
polynomial time a set of generators for the class group of the ring of
integers of any constant degree number field, given a description of
$\theta$. No polynomial time classical algorithm for this problem is
known. See also Abelian hidden subgroup.\\ \\
\end{minipage}

\noindent
\begin{minipage}[c]{\textwidth}
\noindent
\textbf{Algorithm:} Hidden Shift \\
\textbf{Type:} Oracular \\
\textbf{Speedup:} Superpolynomial \\
\textbf{Description:} We are given oracle access to some function
$f(x)$ on a domain of size $N$. We know that $f(x) = g(x+s)$ where $g$
is a known function and $s$ is an unknown shift. The hidden shift
problem is to find $s$. By reduction from Grover's problem it is clear
that at least $\sqrt{N}$ queries are necessary to solve hidden shift
in general. However, certain special cases of the hidden shift problem
are solvable on quantum computers using $O(1)$ queries. In particular,
van Dam \emph{et al.} showed that this can be done if $f$ is a
multiplicative character of a finite ring or
field\cite{vanDam_shift}. The previously discovered shifted Legendre 
symbol algorithm\cite{vanDam_Legendre, vanDam_weighing} is subsumed as
a special case of this, because the Legendre symbol $\left(
\frac{x}{p} \right)$ is a multiplicative character of
$\mathbb{F}_p$. No classical algorithm running in time
$O(\mathrm{polylog}(N))$ is known for these problems. Furthermore, the quantum
algorithm for the shifted Legendre symbol problem breaks certain classical
cryptosystems\cite{vanDam_shift}. \\ \\
\end{minipage}

\noindent
\begin{minipage}[c]{\textwidth}
\noindent
\textbf{Algorithm:} Gauss Sums \\
\textbf{Type:} Non-oracular \\
\textbf{Speedup:} Superpolynomial \\
\textbf{Description:} Let $\mathbb{F}_q$ be a finite
field. The elements other than zero of $\mathbb{F}_q$ form a group
$\mathbb{F}_q^\times$ under multiplication, and the elements of
$\mathbb{F}_q$ form an (Abelian but not necessarily cyclic) group
$\mathbb{F}_q^+$ under addition. We can choose some representation
$\rho^\times$ of $\mathbb{F}_q^\times$ and some representation
$\rho^+$ of $\mathbb{F}_q^+$. Let $\chi^\times$ and $\chi^+$ be the
characters of these representations. The Gauss sum corresponding to
$\rho^\times$ and $\rho^+$ is the inner product of these characters:
$\sum_{x \neq 0 \in \mathbb{F}_q} \chi^+(x) \chi^\times(x)$. As shown
by van Dam and Seroussi\cite{vanDam_Gauss}, Gauss sums can be
estimated to polynomial precision on a quantum computer in polynomial
time. Although a finite ring does not form a group under
multiplication, its set of units does. Choosing a representation for
the additive group of the ring, and choosing a representation for the
multiplicative group of its units, one can obtain a Gauss sum over the
units of a finite ring. These can also be estimated to polynomial
precision on a quantum computer in polynomial
time\cite{vanDam_Gauss}. No polynomial time classical algorithm for
estimating Gauss sums is known. Furthermore, discrete log reduces to
Gauss sum estimation. \\ \\
\end{minipage}

\noindent
\begin{minipage}[c]{\textwidth}
\noindent
\textbf{Algorithm:} Abelian Hidden Subgroup \\
\textbf{Type:} Oracular \\
\textbf{Speedup:} Exponential \\
\textbf{Description:} Let $G$ be a finitely generated Abelian group,
and let $H$ be some subgroup of $G$ such that $G/H$ is finite. Let $f$
be a function on $G$ such that for any $g_1,g_2 \in G$, $f(g_1) =
f(g_2)$ if and only if $g_1$ and $g_2$ are in the same coset of
$H$. The task is to find $H$ (\emph{i.e.} find a set of generators for
$H$) by making queries to $f$. This is solvable on a quantum computer
using $O(\log |G|)$ queries, whereas classically $\Omega(|G|)$ are
required. This algorithm was first formulated in full generality by
Boneh and Lipton in \cite{BL95}. However, proper attribution of this
algorithm is difficult because, as described in chapter 5 of
\cite{Nielsen_Chuang}, it subsumes many historically important quantum
algorithms as special cases, including Simon's algorithm, which was
the inspiration for Shor's period finding algorithm, which forms the
core of his factoring and discrete-log algorithms. The Abelian hidden
subgroup algorithm is also at the core of the Pell's equation,
principal ideal, unit group, and class group algorithms. In certain
instances, the Abelian hidden subgroup problem can be solved using a
single query rather than $\log(|G|)$, see \cite{Beaudrap}.\\ \\
\end{minipage}

\noindent
\begin{minipage}[c]{\textwidth}
\noindent
\textbf{Algorithm:} Non-Abelian Hidden Subgroup \\
\textbf{Type:} Oracular \\
\textbf{Speedup:} Exponential \\
\textbf{Description:} Let $G$ be a finitely generated group, and let
$H$ be some subgroup of $G$ that has finitely many left cosets.  Let
$f$ be a function on $G$ such that for any $g_1,g_2 \in G$, 
$f(g_1) =  f(g_2)$ if and only if $g_1$ and $g_2$ are in the same left
coset of $H$. The task is to find $H$ (\emph{i.e.} find a set of generators for
$H$) by making queries to $f$. This is solvable on a quantum computer
using $O(\log(|G|)$ queries, whereas classically $\Omega(|G|)$ are
required\cite{Ettinger, Hallgren_Russell}. However, this 
does not qualify as an efficient quantum algorithm because in general,
it may take exponential time to process the quantum states obtained
from these queries. Efficient quantum algorithms for the hidden
subgroup problem are known for certain specific non-Abelian
groups\cite{RB_NAHS, IMS_NAHS, MRRS_NAHS, IlG_NAHS, BCvD_NAHS,
  CKL_NAHS, ISS_NAHS, CP_NAHS, ISS2_NAHS, FIMSS_NAHS, G_NAHS,
  CvD_NAHS}. A slightly outdated survey is given in
\cite{Survey_NAHS}. Of particular interest are the symmetric group and
the dihedral group. A solution for the symmetric group would solve
graph isomorphism. A solution for the dihedral group would solve
certain lattice problems\cite{Regev_lattice}. Despite much effort, no
polynomial-time solution for these groups is known. However,
Kuperburg\cite{Kuperberg} found a time $O(2^{C \sqrt{\log N}})$
algorithm for finding a hidden subgroup of the dihedral group
$D_N$. Regev subsequently improved this algorithm so that it uses not
only subexponential time but also polynomial
space\cite{Regev_dihedral}.
\\ \\
\end{minipage}

\subsection{Oracular Problems}

\noindent
\begin{minipage}[c]{\textwidth}
\noindent
\textbf{Algorithm:} Searching \\
\textbf{Type:} Oracular \\
\textbf{Speedup:} Polynomial \\
\textbf{Description:} We are given an oracle with $N$ allowed
inputs. For one input $w$ (``the winner'') the corresponding output is
1, and for all other inputs the corresponding output is 0. The task is
to find $w$. On a classical computer this requires $\Omega(N)$
queries. The quantum algorithm of Lov Grover achieves this using
$O(\sqrt{N})$ queries\cite{Grover_search}.This has algorithm has
subsequently been generalized to search in the presence of multiple
``winners''\cite{BBHT98}, evaluate the sum of an arbitrary
function\cite{BBHT98,BHT98,Mos98}, find the global minimum of an
arbitrary function\cite{DH96,NW99}, and approximate definite
integrals\cite{integration}. The generalization of Grover's algorithm
known as amplitude estimation\cite{Amplitude} is now an
important primitive in quantum algorithms. Amplitude estimation forms
the core of most known quantum algorithms related to collision finding
and graph properties.\\ \\
\end{minipage}

\noindent
\begin{minipage}[c]{\textwidth}
\noindent
\textbf{Algorithm:} Bernstein-Vazirani \\
\textbf{Type:} Oracular \\
\textbf{Speedup:} Polynomial \\
\textbf{Description:} We are given an oracle whose input is $n$ bits
and whose output is one bit. Given input $x \in \{0,1\}^n$, the output
is $x \odot h$, where $h$ is the ``hidden'' string of $n$ bits, and
$\odot$ denotes the bitwise inner product modulo 2. The task is to
find $h$. On a classical computer this requires $n$ queries. As shown
by Bernstein and Vazirani\cite{Bernstein_Vazirani}, this can be
achieved on a quantum computer using a single query. Furthermore, one
can construct a recursive version of this problem, called recursive
Fourier sampling, such that quantum computers require exponentially
fewer queries than classical computers\cite{Bernstein_Vazirani}. \\ \\
\end{minipage}

\noindent
\begin{minipage}[c]{\textwidth}
\noindent
\textbf{Algorithm:} Deutsch-Josza \\
\textbf{Type:} Oracular \\
\textbf{Speedup:} Polynomial \\
\textbf{Description:} We are given an oracle whose input is $n$ bits
and whose output is one bit. We are promised that out of the $2^n$
possible inputs, either all of them, none of them, or half of them
yield output 1. The task is to distinguish the balanced case (half of
all inputs yield output 1) from the constant case (all or none of the
inputs yield output 1). It was shown by Deutsch\cite{Deutsch} that for
$n=1$, this can be solved on a quantum computer using one query,
whereas any deterministic classical algorithm requires two. This was
historically the first well-defined quantum algorithm achieving a
speedup over classical computation. The generalization to arbitrary
$n$ was developed by Deutsch and Josza in
\cite{Deutsch_Josza}. Although probabilistically easy to solve with
$O(1)$ queries, the Deutsch-Josza problem has exponential worst case
deterministic query complexity classically. \\ \\
\end{minipage}

\noindent
\begin{minipage}[c]{\textwidth}
\noindent
\textbf{Algorithm:} NAND Tree \\
\textbf{Type:} Oracular \\
\textbf{Speedup:} Polynomial \\
\textbf{Description:} A NAND gate takes two bits of input and produces
one bit of output. By connecting together NAND gates, one can thus
form a binary tree of depth $n$ which has $2^n$ bits of input and
produces one bit of output. The NAND tree problem is to evaluate the
output of such a tree by making queries to an oracle which stores the
values of the $2^n$ bits and provides any specified one of them upon
request. Farhi \emph{et al.} used a continuous time quantum walk model
to show that a quantum computer can solve this problem using
$O(2^{0.5n})$ time whereas a classical computer requires
$\Omega(2^{0.753n})$ time\cite{Farhi_NAND}. It was soon shown that this
result carries over into the conventional model of circuits and
queries\cite{Childs_Jordan}. The algorithm was subsequently
generalized for NAND trees of varying fanin and noniform
depth\cite{ragged}, and to trees involving larger gate
sets\cite{Reichardt_Spalek}, and MIN-MAX trees \cite{Cleve_tree}.\\ \\
\end{minipage}

\noindent
\begin{minipage}[c]{\textwidth}
\noindent
\textbf{Algorithm:} Gradients \\
\textbf{Type:} Oracular \\
\textbf{Speedup:} Polynomial \\
\textbf{Description:} We are given a oracle for computing some
smooth function $f:\mathbb{R}^d \to \mathbb{R}$. The inputs and
outputs to $f$ are given to the oracle with finitely many bits of
precision. The task is to estimate $\nabla f$ at some specified point
$\mathbf{x}_0 \in \mathbb{R}^d$. As I showed in \cite{Jordan_gradient},
a quantum computer can achieve this using one query, whereas a
classical computer needs at least $d+1$ queries. In \cite{Bulger},
Bulger suggested potential applications for optimization
problems\cite{Bulger}. As shown in appendix \ref{quadratic_forms},
a quantum computer can use the gradient algorithm to find the minimum of a
quadratic form in $d$ dimensions using $O(d)$ queries, whereas, as
shown in \cite{Yao}, a classical computer needs at least $\Omega(d^2)$
queries. \\ \\
\end{minipage}

\noindent
\begin{minipage}[c]{\textwidth}
\noindent
\textbf{Algorithm:} Ordered Search \\
\textbf{Type:} Oracular \\
\textbf{Speedup:} Constant \\
\textbf{Description:} We are given oracle access to a list of $N$ numbers
in order from least to greatest. Given a number $x$, the task is to
find out where in the list it would fit. Classically, the best
possible algorithm is binary search which takes $\log_2 N$
queries. Farhi \emph{et al.} showed that a quantum computer can
achieve this using $0.53 \log(N)$ queries\cite{FGGS99}. Currently, the
best known deterministic quantum algorithm for this problem uses
$0.433 \log_2 N$ queries. A lower bound of $\frac{1}{\pi} \log_2 N$
quantum queries has been proven for this problem\cite{Childs_Lee}. In
\cite{Ben_Or_Search}, a randomized quantum algorithm is given whose
expected query complexity is less than $\frac{1}{3} \log_2 N$.\\ \\
\end{minipage}

\noindent
\begin{minipage}[c]{\textwidth}
\noindent
\textbf{Algorithm:} Graph Properties \\
\textbf{Type:} Oracular \\
\textbf{Speedup:} Polynomial \\
\textbf{Description:} A common way to specify a graph is by an oracle,
which given a pair of vertices, reveals whether they are connected by
an edge. This is called the adjacency matrix model. It generalizes
straightforwardly for weighted and directed graphs. Building  on
previous work \cite{DH96,prev2,prev3}, D\"urr \emph{et al.}
\cite{Durr_graphs} show that the quantum query complexity of finding a
minimum spanning tree of weighted graphs, and deciding connectivity
for directed and undirected graphs have $\Theta(n^{3/2})$ quantum
query complexity, and that finding lowest weight paths has $O(n^{3/2}
\log^2 n)$ quantum query complexity. Berzina \emph{et al.}
\cite{Berzina} show that deciding whether a graph is bipartite can be
achieved using $O(n^{3/2})$ quantum queries. All of these problems are
thought to have $\Omega(n^2)$ classical query complexity. For many of
these problems, the quantum complexity is also known for the case
where the oracle provides an array of neighbors rather than entries of
the adjacency matric\cite{Durr_graphs}. See also triangle finding.\\ \\
\end{minipage}

\noindent
\begin{minipage}[c]{\textwidth}
\noindent
\textbf{Algorithm:} Welded Tree \\
\textbf{Type:} Oracular \\
\textbf{Speedup:} Exponential \\
\textbf{Description:} Some computational problems can be phrased in
terms of the query complexity of finding one's way through a
maze. That is, there is some graph $G$ to which one is given oracle
access. When queried with the label of a given node, the oracle
returns a list of the labels of all adjacent nodes. The task is,
starting from some source node (\emph{i.e.} its label), to find the
label of a certain marked destination node. As shown by Childs
\emph{et al.}\cite{Childs_weld}, quantum computers can exponentially
outperform classical computers at this task for at least some
graphs. Specifically, consider the graph obtained by joining together
two depth-$n$ binary trees by a random ``weld'' such that all nodes
but the two roots have degree three. Starting from one root, a quantum
computer can find the other root using $\mathrm{poly}(n)$ queries,
whereas this is provably impossible using classical queries.\\ \\
\end{minipage}

\noindent
\begin{minipage}[c]{\textwidth}
\noindent
\textbf{Algorithm:} Collision Finding \\
\textbf{Type:} Oracular \\
\textbf{Speedup:} Polynomial \\
\textbf{Description:} Suppose we are given oracle access to a two to
one function $f$ on a domain of size $N$. The collision problem is to
find a pair $x,y \in \{1,2,\ldots,N\}$ such that $f(x)=f(y)$. The
classical randomized query complexity of this problem is
$\Theta(\sqrt{N})$, whereas, as shown by Brassard \emph{et al.}, a
quantum computer can achieve this using $O(N^{1/3})$
queries\cite{Brassard_collision}. Buhrman \emph{et al.} subsequently
showed that a quantum computer can also find a collision in an
arbitrary function on domain of size $N$, provided that one exists,
using $O(N^{3/4} \log N)$ queries\cite{Buhrman_collision}, whereas the
classical query complexity is $\Theta(N \log N)$.  The decision version
of collision finding is called element distinctness, and also has $\Theta(N
\log N)$ classical query complexity. Ambainis subsequently improved
upon\cite{Brassard_collision}, achieving a quantum query complexity of
$O(N^{2/3})$ for element distinctness, which is optimal, and extending
to the case of $k$-fold collisions\cite{Ambainis_distinctness}. Given
two functions $f$ and $g$, each on a domain of size $N$, a claw is a
pair $x,y$ such that $f(x) = g(y)$. A quantum computer can find claws
using $O(N^{3/4} \log N)$ queries\cite{Buhrman_collision}. \\ \\
\end{minipage}

\noindent
\begin{minipage}[c]{\textwidth}
\noindent
\textbf{Algorithm:} Triangle Finding \\
\textbf{Type:} Oracular \\
\textbf{Speedup:} Polynomial \\
\textbf{Description:} Suppose we are given oracle access to a
graph. When queried with a pair of nodes, the oracle reveals whether
an edge connects them. The task is to find a triangle (\emph{i.e.} a
clique of size three) if one exists. As shown by Buhrman \emph{et al.}
\cite{Buhrman_collision}, a quantum computer can accomplish this using
$O(N^{3/2})$ queries, whereas it is conjectured that
classically one must query all $\binom{n}{2}$ edges. Magniez \emph{et
  al.} subsequently improved on this, finding a triangle with
$O(N^{13/10})$ quantum queries\cite{Magniez_triangle}.
\end{minipage}

\noindent
\begin{minipage}[c]{\textwidth}
\noindent
\textbf{Algorithm:} Matrix Commutativity \\
\textbf{Type:} Oracular \\
\textbf{Speedup:} Polynomial \\
\textbf{Description:} We are given oracle access to $k$ matrices, each
of which are $n \times n$. Given integers $i,j \in \{1,2,\ldots,n\}$,
and $x \in \{1,2,\ldots,k\}$ the oracle returns the $ij$ matrix element
of the $x\th$ matrix. The task is to decide whether all of these $k$
matrices commute. As shown by Itakura\cite{Itakura}, this can be
achieved on a quantum computer using $O(k^{4/5}n^{9/5})$ queries,
whereas classically this requires $O(k n^2)$ queries. \\ \\
\end{minipage}

\noindent
\begin{minipage}[c]{\textwidth}
\noindent
\textbf{Algorithm:} Hidden Nonlinear Structures \\
\textbf{Type:} Oracular \\
\textbf{Speedup:} Exponential \\
\textbf{Description:} Any Abelian groups $G$ can be visualized as
a lattice. A subgroup $H$ of $G$ is a sublattice, and the cosets of $H$
are all the shifts of that sublattice. The Abelian hidden subgroup
problem is normally solved by obtaining superposition over a random
coset of the Hidden subgroup, and then taking the Fourier transform so
as to sample from the dual lattice. Rather than generalizing to
non-Abelian groups (see non-Abelian hidden subgroup), one can instead
generalize to the problem of identifying hidden subsets other than
lattices. As shown by Childs \emph{et al.}\cite{Childs_nonlinear} this
problem is efficiently solvable on quantum computers for certain
subsets defined by polynomials, such as spheres. Decker \emph{et al.}
showed how to efficiently solve some related problems
in\cite{Wocjan_nonlinear}.\\ \\ 
\end{minipage}

\noindent
\begin{minipage}[c]{\textwidth}
\noindent
\textbf{Algorithm:} Order of Blackbox Group \\
\textbf{Type:} Oracular \\
\textbf{Speedup:} Exponential \\
\textbf{Description:} Suppose a finite group $G$ is given oracularly
in the following way. To every element in $G$, one assigns a
corresponding label. Given an ordered pair of labels of group
elements, the oracle returns the label of their product. The task is to
find the order of the group, given the labels of a set of
generators. Classically, this problem cannot be solved using
$\mathrm{polylog}(|G|)$ queries even if $G$ is Abelian. For Abelian
groups, quantum computers can solve this problem using
$\mathrm{polylog}(|G|)$ queries by reducing it to the Abelian hidden
subgroup problem, as shown by Mosca\cite{Mosca_thesis}. Furthermore,
as shown by Watrous\cite{Watrous_solvable}, this problem can be solved
in $\mathrm{polylog}(|G|)$ queries for any solvable group. \\ \\ 
\end{minipage}

\subsection{Approximation and BQP-complete Problems}

\noindent
\begin{minipage}[c]{\textwidth}
\noindent
\textbf{Algorithm:} Quantum Simulation \\
\textbf{Type:} Non-oracular \\
\textbf{Speedup:} Exponential \\
\textbf{Description:} It is believed that for any physically realistic
Hamiltonian $H$ on $n$ degrees of freedom, the corresponding time
evolution operator $e^{-i H t}$ can be implemented using
$\mathrm{poly}(n,t)$ gates. Unless BPP=BQP, this problem is not
solvable in general on a classical computer in polynomial time. Many
techniques for quantum simulation have been developed for different
applications\cite{Childs_thesis, Zalka_sim, Wiesner_sim,
  Aharonov_Tashma, Abrams_sim, Cleve_sim, Kassal_sim, Lidar_sim}. The
exponential complexity of classically simulating quantum systems led
Feynman to first propose that quantum computers might outperform
classical computers on certain tasks\cite{Feynman}.  \\ \\
\end{minipage}

\noindent
\begin{minipage}[c]{\textwidth}
\noindent
\textbf{Algorithm:} Jones Polynomial \\
\textbf{Type:} Non-oracular \\
\textbf{Speedup:} Exponential \\
\textbf{Description:} As shown by Freedman\cite{Freedman, Freedman2},
\emph{et al.}, finding a certain additive approximation to the Jones
polynomial of the plat closure of a braid at $e^{i 2 \pi/5}$ is a
BQP-complete problem. This result was reformulated and extended to
$e^{i 2 \pi/k}$ for arbitrary $k$ by Aharonov \emph{et
  al.}\cite{Aharonov1, Aharonov2}. Wocjan and Yard further generalized
this, obtaining a  quantum algorithm to estimate the HOMFLY
polynomial\cite{Wocjan}, of which the Jones polynomial is a special
case. Aharonov \emph{et al.} subsequently showed that quantum
computers can in polynomial time estimate a certain additive
approximation to the even more general Tutte polynomial for planar
graphs\cite{Aharonov3}. The hardness of the additive
approximation obtained in \cite{Aharonov3} is not yet fully
understood. As discussed in chapter \ref{Jones} of this thesis, the
problem of finding a certain additive approximation to the Jones
polynomial of the trace closure of a braid at $e^{i 2 \pi/5}$ is
DQC1-complete. \\ \\
\end{minipage}

\noindent
\begin{minipage}[c]{\textwidth}
\noindent
\textbf{Algorithm:} Zeta Functions \\
\textbf{Type:} Non-oracular \\
\textbf{Speedup:} Superpolynomial \\
\textbf{Description:} As shown by Kedlaya\cite{Kedlaya}, quantum
computers can determine the zeta function of a genus $g$ curve over a
finite field $\mathbb{F}_q$ in time polynomial in $g$ and $\log
q$. No polynomial time classical algorithm for this problem is known. 
More speculatively, van Dam has conjectured that due to a connection
between the zeros of zeta functions and the eigenvalues of certain
quantum operators, quantum computers might be able to efficiently
approximate the number of solutions to equations over finite
fields\cite{vanDam_zeros}. Some evidence supporting this conjecture is
given in \cite{vanDam_zeros}.
\\ \\
\end{minipage}

\noindent
\begin{minipage}[c]{\textwidth}
\noindent
\textbf{Algorithm:} Weight Enumerators \\
\textbf{Type:} Non-oracular \\
\textbf{Speedup:} Exponential \\
\textbf{Description:} Let $C$ a code on $n$ bits,
\emph{i.e.} a subset of $\mathbb{Z}_2^n$. The weight enumerator of
$C$ is $S_C(x,y) = \sum_{c \in C} x^{|c|} y^{n-|c|}$, where $|c|$ 
denotes the Hamming weight of $c$. Weight enumerators have many uses
in the study of classical codes. If $C$ is a linear code, it can be
defined by $C = \{c: Ac = 0\}$ where $A$ is a matrix over
$\mathbb{Z}_2$. In this case $S_C(x,y) = \sum_{c:Ac=0} x^{|c|}
y^{n-|c|}$. Quadratically signed weight enumerators (QWGTs) are a
generalization of this: 
$S(A,B,x,y) = \sum_{c:Ac=0} (-1)^{c^T B c} x^{|c|} y^{n-|c|}$. Now 
consider the following special case. Let $A$ be an $n \times n$
matrix over $\mathbb{Z}_2$ such that $\mathrm{diag}(A) = I$. Let
$\mathrm{lwtr}(A)$ be the lower triangular matrix resulting from
setting all entries above the diagonal in $A$ to zero. Let $l,k$ be
positive integers. Given the promise that
$|S(A,\mathrm{lwtr}(A),k,l)| \geq \frac{1}{2} (k^2+l^2)^{n/2}$, the
problem of determining the sign of $S(A,\mathrm{lwtr}(A),k,l)$ is 
BQP-complete, as shown by Knill and Laflamme in \cite{Knill_QWGT}. The
evaluation of QWGTs is also closely related to the evaluation of
Ising and Potts model partition functions\cite{Lidar_Ising,
  Geraci_QWGT1, Geraci_QWGT2, Geraci_exact}.\\ \\
\end{minipage}

\noindent
\begin{minipage}[c]{\textwidth}
\noindent
\textbf{Algorithm:} Simulated Annealing \\
\textbf{Type:} Non-oracular \\
\textbf{Speedup:} Polynomial \\
\textbf{Description:} In simulated annealing, one has a series of
Markov chains defined by stochastic matrices $M_1, M_2,\ldots
,M_n$. These are slowly varying in the sense that their limiting
distributions $\pi_1, \pi_2, \ldots, \pi_n$ satisfy $|\pi_{t+1} -
\pi_t| < \epsilon$ for some small $\epsilon$. These distributions can often
be though of as thermal distributions at successively lower
temperatures. If $\pi_1$ can be easily prepared then by applying this
series of Markov chains one can sample from $\pi_n$. Typically, one
wishes for $\pi_n$ to be a distribution over good solutions to some
optimization problem. Let $\delta_i$ be the gap between the largest
and second largest eigenvalues of $M_i$. Let $\delta = \min_i
\delta_i$. The run time of this classical algorithm is proportional to
$1/\delta$. Building upon results of Szegedy\cite{Szegedy}, Somma
\emph{et al.} have shown\cite{Somma} that quantum computers can sample
from $\pi_n$ with a runtime proportional to $1/\sqrt{\delta}$.\\ \\
\end{minipage}

\noindent
\begin{minipage}[c]{\textwidth}
\noindent
\textbf{Algorithm:} String Rewriting \\
\textbf{Type:} Non-oracular \\
\textbf{Speedup:} Exponential \\
\textbf{Description:} String rewriting is a fairly general model of
computation. String rewriting systems (sometimes called grammars) are
specified by a list of rules by which certain substrings are allowed
to be replaced by certain other substrings. For example, context free
grammars, are equivalent to the pushdown automata. In
\cite{Wocjan_strings}, Janzing and Wocjan showed that
a certain string rewriting problem is PromiseBQP-complete. Thus
quantum computers can solve it in polynomial time, but classical
computers probably cannot. Given three strings $s$, $t$, and $t'$, and
a set of string rewriting rules satisfying certain promises, the
problem is to find a certain approximation to the difference between
the number of ways of obtaining $t$ from $s$ and the number of ways of
obtaining $t'$ from $s$. Similarly, certain problems of approximating
the difference in number of paths between pairs of vertices in a graph, and
difference in transition probabilities between pairs of states in a
random walk are also BQP-complete\cite{Wocjan_walks}.\\ \\
\end{minipage}

\noindent
\begin{minipage}[c]{\textwidth}
\noindent
\textbf{Algorithm:} Matrix Powers \\
\textbf{Type:} Non-oracular \\
\textbf{Speedup:} Exponential \\
\textbf{Description:} Quantum computers have an exponential advantage
in approximating matrix elements of powers of exponentially large
sparse matrices. Suppose we are have an $N \times N$ symmetric
matrix $A$ such that there are at most $\mathrm{polylog}(N)$ nonzero
entries in each row, and given a row index, the set of nonzero entries
can be efficiently computed. The task is, for any $1 < i < N$, and any
$m$ polylogarithmic in $N$, to approximate $(A^m)_{ii}$, the $i\th$
diagonal matrix element of $A^m$. The approximation is additive to
within $b^m \epsilon$, where $b$ is a given upper bound on $\|A \|$
and $\epsilon$ is of order $1/\mathrm{polylog}(N)$. As shown by
Janzing and Wocjan, this problem is PromiseBQP-complete, as is the
corresponding problem for off-diagonal matrix
elements\cite{Wocjan_matrix}. Thus, quantum computers can solve it in
polynomial time, but classical computers probably cannot.\\ \\ 
\end{minipage}

\noindent
\begin{minipage}[c]{\textwidth}
\noindent
\textbf{Algorithm:} Verifying Matrix Products \\
\textbf{Type:} Non-oracular \\
\textbf{Speedup:} Polynomial \\
\textbf{Description:} Given three $n \times n$ matrices, $A$, $B$, and
$C$, the matrix product verification problem is to decide whether
$AB=C$. Classically, the best known algorithm achieves this in time
$O(n^2)$, whereas the best known classical algorithm for matrix
multiplication runs in time $O(n^{2.376})$. Ambainis
\emph{et al.} discovered a quantum algorithm for this problem with
runtime $O(n^{7/4})$ \cite{Ambainis_matrix}. Subsequently, Buhrman and
\v Spalek improved upon this, obtaining a quantum algorithm for this
problem with runtime $O(n^{5/3})$ \cite{Buhrman_matrix}. This latter
algorithm is based on results regarding quantum walks that were proven
in\cite{Szegedy}. \\ \\
\end{minipage}

\subsection{Commentary}

As noted above, some of these algorithms break existing
cryptosystems. I mention this not because I care about breaking
cryptosystems, but because public key cryptosystems serve as a useful
indicator of the general consensus regarding the computational
difficulty of certain mathematical problems. For each known public key
cryptosystem, the security proof rests on an assumption that a certain
mathematical problem cannot be solved in polynomial time. As discussed
in section \ref{classical_prelim}, nobody knows how to prove that
these problems cannot be solved in polynomial time. However, some of
these problems, such as factoring, have resisted many years of
attempts at polynomial time solution. Thus, many people consider
it to be a safe assumption that factoring is hard. People and
corporations also effectively wager money on this, as nearly all
monetary transactions on the internet are encoded using the RSA
cryptosystem, which is based on the assumption that factoring is
hard to solve classically.

Many of the problems solved by these quantum algorithms may seem
somewhat esoteric. Upon hearing that quantum computers can approximate
Tutte polynomials or solve Pell's equation, one may ask ``Why should I
care?''. One answer to this is that mathematical algorithms sometimes
have applications which are not discovered until long after the
algorithm itself. A deeper answer is that complexity theory has shown
that the ability to solve hard computational problems is to some
degree a fungible resource. That is, many hard problems reduce to one
another with polynomial overhead. By finding a polynomial time
solution for one hard problem, one obtains polynomial time solutions
for a class of problems that can appear unrelated. An interesting
example of this is the LLL algorithm, for the apparently esoteric
problem of finding a basis of short vectors for a lattice. This has
subsequently found application in cryptography, error correction, and
finding integer relations between numbers. LLL and subsequent variants
for integer relation finding have even found use in computer assisted
mathematics. Their achievements include, among other things, the
discovery of a new formula for the digits of $\pi$ such that any digit
of $\pi$ can be calculated using a constant amount of computation
without having to calculate the preceeding digits\cite{pi_digits}!

In light of such history, and in light of known properties of
computational complexity, it makes sense to search for quantum
algorithms that provide polynomial speedups for problems of direct
practical relevance, and to try to find exponential speedups for
\emph{any} problem.

\section{What makes quantum computers powerful?}

The quantum algorithms described in section \ref{algorithms} establish
rigorously that quantum computers can solve some problems using far
fewer queries than classical computers, and establish convincingly
that quantum computers can solve certain problems using far fewer
computational steps than classical computers. It is natural to ask
what aspect of quantum mechanics gives quantum computers their
extra computational power. At first glance this appears to be a vague
and ill-posed question. However, we can approach this question in a
concrete way by taking away different aspects of quantum mechanics one
at a time, and seeing whether the resulting models of computation
retain the power of quantum computers.

One necessary ingredient for the power of quantum computing is the
exponentially high-dimensional Hilbert space. If we take this away,
then the resulting model of computation can be simulated in polynomial
time by a classical compute. To simulate the action of each
gate, one would need only to multiply the state vector by a unitary matrix
of polynomial dimension. Interestingly, classical optics can be
described by a formalism that is nearly identical to quantum
mechanics, the only difference being that the amplitudes are a
function of three spatial dimensions rather than an arbitrary number
of degrees of freedom. As described in appendix \ref{optical}, this
analogy was fruitful in that it led me to discover a quantum
algorithm for estimating gradients faster than is possible
classically\cite{Jordan_gradient}. Intuitions from optics were
apparently also used in the development of quantum algorithms for the
identification of hidden nonlinear structures\cite{Childs_nonlinear}.

We have seen that an exponentially large state space is a necessary
ingredient for the power of quantum computers. However, the states of
a probabilistic computer live in a vector space of exponentially high
dimension too. The state of a probabilistic computer with $n$ bits is
a vector in the $2^n$ dimensional space of probability distributions
over its possible configurations. The essential difference is that
quantum systems can exhibit interference due to the cancellation of
amplitudes. In contrast probabilities are all positive and cannot
interfere.

In light of this comparison between quantum and probabilistic
computers, it is natural to ask whether it is necessary that the
amplitudes be complex for quantum computers to retain their
power. After all, real amplitudes can still interfere as long as they
are allowed to be both positive and negative. It turns out that real
amplitudes are sufficient to obtain BQP\cite{Shi_real}. The proof of
this is based on the following simple idea. Take an arbitrary state of
$n$-qubits
\[
\ket{\psi} = \sum_{x=0}^{2^n-1} a_x \ket{x}.
\]
One can encode it by the following real state on $n+1$ qubits
\[
\ket{\psi_\mathbb{R}} = \sum_{x=0}^{2^n-1} \mathrm{Re}(a_x)
  \ket{x}\ket{0} + \mathrm{Im}(a_x) \ket{x}\ket{1}.
\]
As shown in \cite{Shi_real}, for each quantum gate, an equivalent
version on the encoded states can be efficiently constructed. As a
result, arbitrary quantum computations can be simulated using only
real amplitudes.

It is often said that entanglement is a key to the power of quantum
computers. A completely unentangled pure state on $n$-qubits is always
of the form 
\[
\ket{\psi_1} \otimes \ket{\psi_2} \otimes \ldots \otimes \ket{\psi_n}
\]
where $\ket{\psi_1},\ldots,\ket{\psi_n}$ are each single-qubit
states. Each of these states can be described by a pair of complex
amplitudes. Thus, the unentangled states are described by $2n$ complex
numbers in contrast to arbitrary states which in general require
$2^n$. Hence, it is not surprising that quantum computers must use
entangled states in order to obtain speedup over classical
computation.

\begin{figure}
\begin{center}
\includegraphics[width=0.6\textwidth]{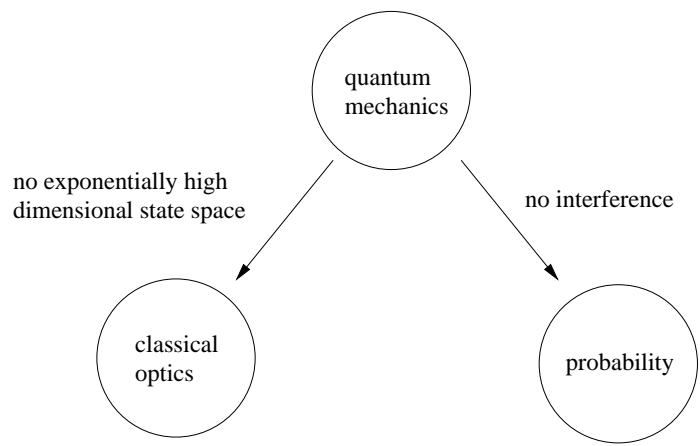}
\caption{\label{powers} A diagram of (loose) conceptual relationships
  between quantum mechanics, classical optics, and
  probability. Correspondingly, by taking away interference from
  quantum computers, one is left with the power of probabilistic
  computers, and by taking away the exponentially high dimensional
  space of quantum states, one is left with the power of optical
  computing. Interestingly, the Fourier transform is an important
  primitive in both quantum and optical computing.}
\end{center}
\end{figure}

Both interference and an exponentially high-dimensional state space
seem to be necessary to the power of quantum
computation. Nevertheless, there are classes of quantum processes 
which involve both of these characteristics yet can be simulated
classically in polynomial time. Certain quantum states admit concise
group-theoretic description. The Pauli group $P_n$ on $n$ qubits is
the group of $n$-fold tensor products of the four Pauli matrices
$\{X,Y,Z,I\}$ with phases of $\pm 1$ and $\pm i$.
\[
\begin{array}{cc}
X = \left[ \begin{array}{cc}
0 & 1 \\
1 & 0
\end{array} \right]
&
Y = \left[ \begin{array}{cc}
0 & -i \\
i & 0
\end{array} \right] \\
 & \\
Z = \left[ \begin{array}{cc}
1 & 0 \\
0 & -1
\end{array} \right] &
I = \left[ \begin{array}{cc}
1 & 0 \\
0 & 1
\end{array} \right]
\end{array}
\]
The states stabilized by subgroups of the Pauli group are called
stabilizer states. Any stabilizer state on $n$ qubits can be concisely
described using $\mathrm{poly}(n)$ bits by listing a set of generators
for its stabilizer subgroup.  The Clifford group is the normalizer of
the Pauli group. As discussed in \cite{Gottesman_pra}, it is generated
by CNOT, Hadamard, and
\[
S = \left[ \begin{array}{ll}
1 & 0 \\
0 & i
\end{array} \right]
\]
Applying a Clifford group operation to a stabilizer
state results in another stabilizer state. Thus quantum circuits made
from gates in the Clifford group can be efficiently simulated on a
classical computer\cite{Gottesman_thesis}, provided the initial state
is a stabilizer state. This is possible even though stabilizer
states can be highly entangled and can involve both positive and
negative amplitudes. This result is known as the Gottesman-Knill
theorem.

In addition, many quantum states with limited but nonzero entanglement
can be concisely described using the matrix product state (MPS) and
projected entangled pair state (PEPS) formalisms. Matrix product
states can have amplitudes of all phases. Nevertheless, processes on
MPS and PEPS with limited entaglement can be efficiently simulated on
classical computers\cite{MPS, PEPS}. 

\section{Fault Tolerance}
\label{FT}

Quantum computation is not the first model of physical computation to
offer an apparent exponential advantage over standard digital
computers. Certain analog circuits and even mechanical devices have
seemed to achive exponential speedups for some problems. However,
closer inspection has always shown that these devices depend on
exponential precision to operate, thus the speedups offered are
physically unrealistic (see introduction of \cite{Shor_factoring} and
references therein). This is one reason why we rarely see discussion of
analog computers today.

One of the most sensible objections raised in the early days of quantum
computing was that quantum computers might be a form of analog
computer, dependent on exponential precision in order to achieve
exponential speedup. The discussion of section \ref{quantum_prelim}
shows that this is not true. To perform a computation on
$\mathrm{poly}(n)$ gates, one needs only to perform each gate with
$1/\mathrm{poly}(n)$ precision. However, from a practical point of
view, this seems not entirely satisfactory. Presumably the precision
achievable in the laboratory is limited. Thus even if the precision
necessary for computations of size $n$ is only $1/\mathrm{poly}(n)$,
the achievable computations will be limited to some maximum size. In
addition to gate imperfections, errors can arise from stray couplings
to the environment, which is ignored in the analysis of section
\ref{quantum_prelim}.

The threshold theorem shows that both of these problems are solvable
in principle. More precisely, the threshold theorem shows that if
errors are below a certain fixed threshold, then quantum computations
of unlimited length can be carried out reliably(see chapter 10 of
\cite{Nielsen_Chuang}). This is achieved by encoding the quantum
information using quantum error correcting codes, and continually
correcting errors as they occur throughout the computation. It does
not matter whether the error arises from gate imperfections or from
stray influence from the environment, as long as the total error rate
is below the fault tolerance threshold.

Any operator on $n$ qubits can be uniquely decomposed as a linear
combination of $n$-fold tensor products of the Pauli matrices
$\{X,Y,Z,I\}$. The number of non-identity Pauli matrices in a given
tensor product is called its weight. If the Pauli decomposition of an
operator consists only of tensor products of weight at most $k$, then
the operator is said to be $k$-local.

The essence of quantum error correction is to take advantage of the
fact that errors encountered are likely to be of low Pauli
weight. Suppose for example, that each qubit gets flipped in the
computational basis (\emph{i.e.} acted on by an $X$ operator) with
probability $p$. Then, the probability that the resulting error is of
weight $k$ is of order $p^k$. If $p$ is small, then with high
probability the errors will be of low weight. The error model
described here is an essentially classical one, but as discussed
in\cite{Nielsen_Chuang}, the same conclusion carries through when
considering errors other than bitflips, coherent superpositions rather
than probabilistic mixtures of corrupted and uncorrupted states, and
errors arising from persistent perturbations to the control
Hamiltonian rather than discrete ``kicks''.

An $[n,k]$ quantum code is a $2^k$-dimensional subspace of the
$2^n$-dimensional Hilbert space of $n$ qubits. Thus, it encodes $k$
logical qubits using $n$ physical qubits. Let $P$ be the projector on
to the code. Suppose that there is a discrete set of possible errors
$\{E_1,E_2,\ldots,E_m\}$ that we wish to correct. Then for
error correction it suffices for  $E_1 P, E_2 P, \ldots,E_m P$ to be
mutually orthogonal, because then the errors can be distinguished and
hence corrected. Of course, in the quantum setting, the
possible errors may form a continuum, but as discussed
in\cite{Nielsen_Chuang}, this problem can be avoided by making an
appropriate measurement to collapse the system into one of a discrete
set of errors.

As an alternative to the active correction of errors, schemes have
been proposed in which the physical system from which the quantum
computer is constructed has intrinsic resistance to errors. One of the
earliest examples of this is the Kitaev's quantum memory based on
toric codes\cite{Kitaev_memory}. The toric code is an $[l^2,2]$ code
defined on an  $l \times l$ square lattice of qubits on a torus. It
has the property that any error of weight less than $l$ is
correctable. Furthermore, one can construct a 4-local Hamiltonian
with a 4-fold degenerate ground space equal to the code space. The
Hamiltonion provides an energy penalty against any error of weight $l$
or less. Thus, if the ambient temperature is small compared to this
energy penalty, the system is unlikely to get kicked out of the
ground space. Furthermore, $c$-local error terms in the Hamiltonian
only cause splitting of the ground space degeneracy at order $l/c$ in
perturbation theory. Topological quantum computation is closely
related to toric codes and is also a promising candidate for
intrinsically robust quantum computation\cite{Nayak}.

The active schemes of quantum error correction generally yield very
low fault tolerance thresholds which are difficult to achieve
experimentally. Furthermore, the amount of overhead incurred by the
error correction process can be very large if the noise only slightly
below the threshold. The passive schemes of error protection may
reduce or eliminate the need for costly active error correction. In
chapter \ref{FT_AQC}, I investigate the fault tolerance of adiabatic
quantum computers and find that such passive error protection schemes
show promise for the adiabatic model of quantum computation. Although
adiabatic quantum computation has attractive features for
experimentalists, particularly regarding solid state qubits, no
threshold theorem for the adiabatic model of quantum computation is
currently known. I see the establishment of a threshold theorem for
adiabatic quantum computers as a major open problem in the theory of
quantum fault tolerance.

\section{Models of Quantum Computation}

In section \ref{quantum_prelim} I discussed the
universality of quantum circuits, and the reasons to believe that no
model of quantum computation is more powerful than the quantum circuit
model. That is, no discrete nonrelativistic quantum system is capable
of efficiently solving problems outside of BQP. Furthermore, as
discussed in section \ref{FT}, quantum circuits can be fault tolerant in
the sense that they can accurately perform arbitrarily long
computations provided the error rate is below a certain
threshold. Thus, in principle, the quantum circuit model is the only
model we need for the both study quantum algorithms, and the physical
implementation of quantum computers. In practice however, for both the
development of new quantum algorithms and the physical construction of
quantum computers it has proven useful to have alternative models of
quantum computation.

\subsection{Adiabatic}

In the adiabatic model of quantum computation, one starts with an
initial Hamiltonian with an easy to prepare ground state, such as
$\ket{0}^{\otimes n}$. Then, the Hamiltonian is slowly varied until it
reaches some final Hamiltonian whose ground state encodes the solution
to some computational problem. The adiabatic theorem shows that if the
Hamiltonian is varied sufficiently slowly and the energy gap between
the ground state and first excited state is sufficiently large, then
the system will track the instantaneous ground state of the
time-varying Hamiltonian. More precisely, suppose the Hamiltonian is
$H(t)$, and the evolution is from $t=0$ to $t=T$. Let $\gamma(t)$ be
the gap between the ground energy and first excited energy at time
$t$. Let $\gamma = \min_{0 \leq t \leq T} \gamma(t)$. Then the
necessary runtime to ensure high overlap of the final state with the
final ground state scales as $1/\mathrm{poly}(\gamma)$. A rough
analysis\cite{Messiah} suggests that the runtime should in fact scale
as $1/\gamma^2$. However, it is not clear that this holds as a rigorous
theorem for all cases. Nevertheless, rigorous versions of the
adiabatic theorem are known. For example, in appendix
\ref{adiabatic_theorem} we reproduce an elegant proof due to Jeffrey
Goldstone of the following theorem:

\begin{theorem}
Let $H(s)$ be a finite-dimensional twice differentiable Hamiltonian on
$0 \leq s \leq 1$ with a nondegenerate ground state $\ket{\phi_0(s)}$
separated by an energy energy gap $\gamma(s)$. Let $\ket{\psi(t)}$ be
the state obtained by Schr\"odinger  time evolution with Hamiltonian
$H(t/T)$ starting with state $\ket{\phi_0(0)}$ at $t = 0$. Then, with
appropriate choice of phase for $\ket{\phi_0(t)}$,
\[
\begin{array}{l}
\| \Ket{\psi(T)} - \Ket{\phi_0(T)} \| \leq
\frac{1}{T} \left[ \frac{1}{\gamma(0)^2} \left\| \frac{\ud H}{\ud s} \right\|_{s=0}
+ \frac{1}{\gamma(1)^2} \left\| \frac{\ud H}{\ud s} \right\|_{s=1}
+ \int_0^1 \ud s \left( \frac{5}{\gamma^3} \left\| \frac{\ud H}{\ud s}
\right\|^2 + \frac{1}{\gamma^2} \left\| \frac{\ud^2 H}{\ud s^2}
\right\| \right) \right]. 
\end{array}
\]
\end{theorem}

Schr\"odinger's equation shows that, for any constant $g$, the
time-dependent Hamiltonian $gH(gt)$ yields the same time evolution
from time 0 to $T/g$ that $H(t)$ yields from 0 to $T$. Thus, the
running time of an adiabatic algorithm would not appear to be well
defined. However, in any experimental realization there will be a
limit to the magnitude of the fields and couplings. Thus it is
reasonable to limit the norm of each local term in $H(t)$. Such a
restriction enables one to make statements about how the running time
of an adiabatic algorithm scales with some measure of the problem
size. An alternative convention is to simply normalize $\| H(t) \|$
to 1. 

Adiabatic quantum computation was first proposed as a method to solve
combinatorial optimization problems\cite{Farhi_adiabatic}. The
spectral gap, and hence the runtime, of the proposed adiabatic
algorithms for combinatorial optimization remain unknown. Quantum
circuits can simulate adiabatic quantum computers with polynomial
overhead using standard techniques of quantum simulation. In 
\cite{Aharonov_adiabatic} it was shown that adiabatic quantum
computers can simulate arbitrary quantum circuits with polynomial
overhead. Thus, up to a polynomial factor, adiabatic quantum computers
are equivalent to the quantum circuit model. In other words, the set
of problems solvable in polynomial time by adiabatic quantum computers
is exactly BQP.

In \cite{Aharonov_adiabatic}, Aharonov \emph{et al.} present a
construction for doing universal quantum computation with a 5-local
Hamiltonian. The minimum eigenvalue gap is proportional to $1/g^2$,
where $g$ is the number of gates in the circuit being
simulated. Assuming quadratic scaling of runtime with the inverse gap,
this implies a quartic overhead. They also show how to achieve
universal adiabatic quantum computation with 3-local Hamiltonians and
a runtime of $O(1/g^{14})$. Using the perturbative gadgets of
\cite{Kempe} this can be reduced to 2-local with further overhead in
runtime. These runtimes were subsequently greatly improved. Using
a clever construction of Nagaj and Moses\cite{Nagaj}, one can achieve
universal adiabatic quantum computation using a 3-local Hamiltonian
with a gap of order $1/g^2$ throughout the
computation\footnote{Surprisingly, the authors do not explicitly state
  in \cite{Nagaj} that their construction can be used for this
  purpose.}.

When using adiabatic quantum computation as a method of devising
algorithms rather than as an architecture for building quantum
computers, one can consider simulatable Hamiltonians, which are a
larger class than physically realistic Hamiltonians. As shown
in\cite{Aharonov_Tashma, Cleve_sim}, sparse Hamiltonians can be
efficiently simulated on quantum circuits even if they are not
local. Furthermore, if the adiabatic algorithm runs in time $T$ then,
the simulation can be accomplished in time $T^{1+1/k}$ using a $k\th$
order Suzuki-Trotter formula.

Several reasons have been proposed for why adiabatic quantum computers
might be easier to physically implement than standard quantum
computers. The standard architecture for physically implementing quantum
computation is based on the quantum circuit model. Each gate is
performed by applying a pulse to the relevant qubits. For example, in
an ion trap quantum computer, a laser pulses are used to manipulate the
electronic state of ions. In any such pulse-based scheme, it takes a
large bandwidth to transmit the control pulses to the qubits. This 
therefore leaves a large window open for noise to enter the system and
disturb the qubits. In contrast, in an adiabatic quantum computer, all
the control is essentially DC, and therefore much of the noise other
than that at extremely low frequencies can be filtered
out\cite{DiVincenzo_personal, Rose_personal}. Secondly, as a
consequence of the adiabatic theorem, if the Hamiltonian $H(t)$ drifts
off course during the computation, then the adiabatic algorithm will
still succeed provided that the initial and final Hamiltonians are
correct, and adiabaticity is maintained. Furthermore, dephasing in the
eigenbasis of $H(t)$ causes no decrease in the success
probability. Lastly, $H(t)$ is applied constantly. If the minimum
energy gap between the ground and first excited states is $\gamma$ and
the ambient temperature $kT$ is less than $\gamma$, then the system
will be unlikely to get thermally excited out of its ground
state\cite{Childs_robustness}. Unfortunately, in most adiabatic
algorithms, $\gamma$ scales inversely with the problem size,
apparently necessitating progressively lower temperatures to solve
larger problems. A technique for getting around this problem is
discussed in chapter \ref{FT_AQC}.

\subsection{Topological}

Topological quantum computation is a model of quantum computation
based on the braiding of a certain type of quasiparticles called
anyons, which can arise in quasi-two-dimensional many-body
systems. The energy of the system depends only on how many
quasiparticles are present. By adiabatically dragging these particles
around one another in two dimensions and back to their original
locations, one may incur a Berry's phase. In certain systems, this
phase has the special property that it depends only on the topology of
the path and not on the geometry. The phase induced by the braiding of
$n$ identical particles will thus be a representation of the
$n$-strand braid group. Particles with this property are called
anyons. If the space of $n$-particle states is $d$-fold degenerate,
then by braiding the particles around each other, one can move the
system within this degenerate space. The ``phase'' in this case is a
$d$-dimensional unitary representation of the $n$-strand braid
group. The representation can thus be non-Abelian, in which case the
particles are said to be non-Abelian anyons.

Not all representations of the braid group correspond to anyons that
are physically realized. This is because in addition to winding around
each other, anyons can be fused. For example, consider the
Abelian representation of the braid group where the clockwise swapping
of a pair of particles induces a phase of $e^{i\phi}$. Now, we may fuse
a pair of anyons into a bound pair, which can be thought of as another
species of anyon. Winding two of these clockwise around each other
must induce a phase of $e^{i 4 \phi}$, because each anyon in each pair
has wound around each anyon in the other pair. The non-Abelian case is
analogous but more complicated. Such fusion rules and the condition
that the theory be purely topological create constraints on which
representations of the braid group can arise from braiding of
anyons. A set of braiding rules and fusion rules satisfying the
consistency constraints is called a topological quantum field theory
(TQFT). Topological quantum field theories can also be formulated in
the more traditional language of Lagrangians and path
integrals. However, in this thesis I will not need to use the
Lagrangian formulation of TQFTs.

Topological quantum field theories have been well studied by both
mathematicians and physicists. One remarkable result is that the
complete set of consistency relations between braiding and fusion are
completely captured by just two identities, known as the pentagon and
hexagon identities\cite{Nayak}. Despite this progress, a full classification of
topological quantum field theories is not known. However, several
interesting and nontrivial examples of quantum field theories are
known. Of particular interest for quantum computing is the TQFT whose
particles are called Fibonacci anyons. A set of $n$ Fibonoacci anyons
lives in a degenerate eigenspace whose dimension is $f_{n+1}$, the
$(n+1)\th$ Fibonacci number. Freedman \emph{et al.}
showed\cite{Freedman} that the representation of the braid group
induced by the braiding of Fibonacci anyons is dense in $SU(f_{n+1})$,
and furthermore that quantum circuits on $n$ qubits with
$\mathrm{poly}(n)$ gates can be efficiently simulated by a braid on
$\mathrm{poly}(n)$ Fibonacci anyons with $\mathrm{poly}(n)$
crossings. This is made possible by the fact that $f_{n+1}$ is
exponential in $n$, and that the Fibonacci representation has some
local structure onto which the tensor pruduct structure of quantum
circuits can be efficiently mapped. (More detail is given in chapter
\ref{Jones}.)

The upshot of this correspondence between braids and quantum
circuits is that one in principle can solve any problem in BQP by
dragging Fibonacci anyons around each other. Conversely, it has also
been shown that quantum circuits can simulate topological quantum
field theories\cite{Freedman2}. Thus topological quantum computing
with Fibonacci anyons is equivalent to BQP. If non-Abelian anyons can
be detected and manipulated, they may provide a useful medium for
quantum computation. Topological quantum computations are believed to
have a high degree of inherent fault tolerance. As long as the anyons
are kept well separated, small deviations in their trajectories will
not change the topology of their braiding, and hence will not change
the encoded quantum circuit. Furthermore, the degenerate eigenspace in
which the anyons live is protected by an energy gap. Thus thermal
transitions out of the space are unlikely. Thermal transisions between
states within the degenerate space, while not protected against by an
energy gap, are also unlikely because they can only be induced by
nonlocal, topologocal operations.

The topological model of quantum computation has also been useful in
the development of new quantum algorithms. In 1989, Witten showed that
the Jones polynomial (a powerful and important knot invariant) arises
as a Wilson loop in a particular quantum field theory called
Chern-Simons theory\cite{Witten}. The subsequent discovery by Freedman \emph{et
  al.}\cite{Freedman2} that quantum computers can simulate topological quantum field
theories thus implicitly showed that quantum computers can efficiently
approximate Jones polynomials. Furthermore, the discovery by Freedman
\emph{et al.} that topological quantum field theories can simulate
quantum circuits implicitly showed that a certain problem of
estimating Jones polynomials at the fifth root of unity is
BQP-hard. As discussed in chapter \ref{Jones}, this has since led to a
whole new class of exponential speedups by quantum computation for the
approximation of various knot invariants and other
polynomials. Furthermore, these speedups are very different from
previously known exponential quantum speedups, most of which are in
some way based on the hidden subgroups.

\subsection{Quantum Walks}
\label{walks}

In a continuous time quantum walk, one chooses a graph with nodes that
correspond to orthogonal states in a Hilbert space. The Hamiltonian is
then chosen to be either the adjacency matrix or Laplacian of this
graph. (For regular graphs these are equivalent up to an overall
energy shift). The quantum walk is the unitary time evolution induced
by this Hamiltonian.

Continuous time quantum walks were introduced
in\cite{Farhi_decision}. They have been found to provide an 
exponential speedup over classical computation for at least one
oracular problem\cite{Childs_weld}. Discrete time quantum walks have
also been formulated, and appear to be comparable in power to
continuous time quantum walks. Quantum walks have now been used
to find polynomial speedups for several natural oracular problems, as
discussed in section \ref{algorithms}. No result exists in the
literature answering the the question as to whether quantum
walks are BQP-complete (\emph{i.e.} universal). However, recent
progress suggests that, to my surprise, quantum walks may in fact be
BQP-complete \cite{Childs_personal}. 

Quantum walks are probably not useful in devising physical models of
quantum computation. The most obvious approach is to lay out the nodes
in space and couple them together along the edges. However, in a
quantum walk, the nodes form the basis of the Hilbert space, which
usually has exponentially high dimension. In quantum walk algorithms,
the Hamiltonians is usually not $k$-local for any fixed
$k$. Nevertheless, as discussed in \cite{Childs_thesis,
  Aharonov_Tashma, Childs_Jordan}, these Hamiltonians are efficiently
simulable by quantum circuits since they are sparse and efficiently
row-computable. 

\subsection{One Clean Qubit}

In the one clean qubit model of quantum computation, one is given a
single qubit in a pure state, and $n$ qubits in the maximally mixed
state. One then applies a polynomial size quantum circuit to this
intial state and afterwards performs a single-qubit measurement. The
one clean qubit model was originally proposed as an idealization of
quantum computation on highly mixed states, such as appear in NMR
implementations \cite{Knill_DQC1, Ambainis_DQC1, Chuang_DQC1}.

It is not surprising that one clean qubit computers appear to be weaker than
standard quantum computers. The amazing fact is that they can
nevertheless solve certain problems for which no efficient classical
algorithm is known. These problems include estimating the Pauli
decomposition of the unitary matrix corresponding to a polynomial-size
quantum circuit\footnote{This includes estimating the trace of the
  unitary as a special case.}, \cite{Knill_DQC1, Shepherd}, estimating
quadratically signed weight enumerators\cite{Knill_QWGT}, and estimating
average fidelity decay of quantum maps\cite{decay1, decay2}, and as
shown in chapter \ref{Jones}, approximating certain Jones polynomials.

The one clean qubit complexity class consists of the decision problems
which can be solved in polynomial time by a one clean qubit machine
with correctness probability of at least $2/3$ by running a one clean
qubit computer polynomially many times. In the original
definition\cite{Knill_DQC1} of DQC1 it is assumed that a classical
computer generates the quantum circuits to be applied to the initial
state $\rho$. By this definition DQC1 automatically contains
P. However, it is also interesting to consider a slightly weaker one
clean qubit model, in which the classical computer controlling the
quantum circuits has only the power of NC1. The resulting complexity
class appears to have the interesting property that it it is not
contained in P nor does P contain it. One clean qubit computers and
DQC1 are discussed in more detail in chapter \ref{Jones}.

\subsection{Measurement-based}
\label{measurement_based}

Amazingly, algorithm dependent unitary operations are not necessary
for universal quantum computation. Building on previous
work\cite{Gottesman_Chuang, Nielsen_measurement, Leung_measurement},
Raussendorf and Briegel showed in \cite{Raussendorf_Briegel} that one can perform
universal quantum computation by performing a series of single-qubit
projective measurements on a special entangled initial state. The
initial state need not depend on the computation to be performed,
other than its total size. The basis of a given single-qubit
measurements depends on the quantum circuit to be simulated, and on
the outcomes of the preceeding measurements. This dependence is
efficiently computable classically.

The measurement-based model is a promising candidate for the physical
implementation of quantum computers. The measurement-based model has a
fault tolerance threshold, which can be shown in a simple way by
adapting the existing threshold theorem for the circuit
model\cite{Leung_threshold}. It seems unlikely that the
measurement-based model will be useful for the design of algorithms,
because of its very direct relationship to the circuit model. However,
the class of initial states used in the measurement model, called
graph states, have many interesting properties both physical and
information theoretic. For example, they form the basis (literally as
well as figuratively) of a broad class ``nonadditive'' quantum codes,
which go beyond the stabilizer formalism\cite{Yu_nonadditive,
  Cross_nonadditive}. (Graph states are stabilizer states. However,
quantum error correcting codes can be obtained as the span of a
set of graph states. In general such a span is not equal to the
subspace stabilized by any subgroup of the Pauli group.)

\subsection{Quantum Turing Machines}

Quantum Turing machines were first formulated in \cite{Deutsch} and
further studied in \cite{Bernstein_Vazirani}. Quantum Turing machines
are defined analogously to classical Turing machines except with a
tape of qubits instead of a tape of bits, and with transition
amplitudes instead of deterministic transition rules. Quantum Turing
machines are usually somewhat cumbersome to work with, and have been
replaced by quantum circuits for most applications. However, quantum
Turing machines remain the only known model by which to define quantum
Kolmogorov complexity.

\section{Outline of New Results}

In this thesis I present three main results relating to different
models of quantum computation.

Recently, there has been growing interest in using adiabatic quantum
computation as an architecture for experimentally realizable quantum
computers. One of the reasons for this is the idea that the energy gap
should provide some inherent resistance to noise. It is now known that
universal quantum computation can be achieved adiabatically using
2-local Hamiltonians. The energy gap in these Hamiltonians scales as
an inverse polynomial in the problem size. In chapter \ref{FT_AQC} I
present stabilizer codes that can be used to produce a constant
energy gap against 1-local and 2-local noise. The corresponding
fault-tolerant universal Hamiltonians are 4-local and 6-local
respectively, which is the optimal result achievable within this
framework. I did this work in collaboration with Edward Farhi and
Peter Shor.

It is known that evaluating a certain approximation to the Jones
polynomial for the plat closure of a braid is a BQP-complete
problem. In chapter \ref{Jones} I show that evaluating a certain
additive approximation to the Jones polynomial at a fifth root of
unity for the trace closure of a braid is a complete problem for the
one clean qubit complexity class DQC1. That is, a one clean qubit 
computer can approximate these Jones polynomials in time polynomial in
both the number of strands and number of crossings, and the problem of
simulating a one clean qubit computer is reducible to approximating
the Jones polynomial of the trace closure of a braid. I did this work
in collaboration with Peter Shor.

Adiabatic quantum algorithms are often most easily formulated using
many-body interactions. However, experimentally available interactions
are generally two-body. In 2004, Kempe, Kitaev, and Regev introduced
perturbative gadgets, by which arbitrary three-body effective
interactions can be obtained using Hamiltonians consisting only of
two-body interactions\cite{Kempe}. These three-body effective
interactions arise from the third order in perturbation theory. Since
their introduction, perturbative gadgets have become a standard tool
in the theory of quantum computation. In chapter \ref{gadgets} I
construct generalized gadgets so that one can directly obtain
arbitrary $k$-body effective interactions from two-body Hamiltonians
using $k\th$ order in perturbation theory. I did this work in
collaboration with Edward Farhi.

\chapter{Fault Tolerance of Adiabatic Quantum Computers}
\label{FT_AQC}

\section{Introduction}

Recently, there has been growing interest in using adiabatic quantum
computation as an architecture for experimentally realizable quantum
computers. Aharonov \emph{et al.}\cite{Aharonov_adiabatic}, building on ideas by
Feynman\cite{Feynman2} and Kitaev\cite{Kitaev_book}, showed that any quantum
circuit can be simulated by an adiabatic quantum algorithm. The energy
gap for this algorithm scales as an inverse polynomial in $G$,  the
number of gates in the original quantum circuit. $G$ is identified as
the running time of the original circuit. By the adiabatic theorem,
the running time of the adiabatic simulation is polynomial in
$G$. Because the slowdown is only polynomial, adiabatic quantum
computation is a form of universal quantum computation.

Most experimentally realizable Hamiltonians involve only few-body
interactions. Thus theoretical models of quantum computation are
usually restricted to involve interactions between at most some
constant number of qubits $k$. Any Hamiltonian on $n$ qubits can be
expressed as a linear combination of terms, each of which is a tensor
product of $n$ Pauli matrices, where we include the $2 \times 2$
identity as a fourth Pauli matrix. If each of these tensor products
contains at most $k$ Pauli matrices not equal to the identity then the
Hamiltonian is said to be $k$-local. The Hamiltonian used in the
universality construction of \cite{Aharonov_adiabatic} is 3-local throughout the
time evolution. Kempe \emph{et al.} subsequently improved this to
2-local in \cite{Kempe}.

Schr\"odinger's equation shows that, for any constant $g$, $g H(g t)$
yields the same time evolution from time $0$ to $T/g$ that $H(t)$
yields from $0$ to $T$. Thus, the running time of an adiabatic
algorithm would not appear to be well defined. However, in any
experimental realization there will be a limit to the magnitude of the
fields and couplings. Thus it is reasonable to limit the norm of each
term in $H(t)$. Such a restriction enables one to make statements
about how the running time of an adiabatic algorithm scales with some
measure of the problem size, such as $G$. 

One of the reasons for interest in adiabatic quantum computation
as an architecture is the idea that adiabatic quantum computers may
have some inherent fault tolerance \cite{Childs_robustness, Sarandy, Aberg,
  Roland, Kaminsky} . Because the final state depends only on the final
Hamiltonian, adiabatic quantum computation may be resistant to slowly
varying control errors, which cause $H(t)$ to vary from its intended
path, as long as the final Hamiltonian is correct. An exception to this
would occur if the modified path has an energy gap small enough to
violate the adiabatic condition. Unfortunately, it is generally quite
difficult to evaluate the energy gap of arbitrary local Hamiltonians.

Another reason to expect that adiabatic quantum computations may be
inherently fault tolerant is that the energy gap should provide some
inherent resistance to noise caused by stray couplings to the
environment. Intuitively, the system will be unlikely to get excited
out of its ground state if $k_b T$ is less than the energy
gap. Unfortunately, in most proposed applications of adiabatic quantum
computation, the energy gap scales as an inverse polynomial in the
problem size. Such a gap only affords protection if the temperature
scales the same way. However, a temperature which shrinks polynomially
with the problem size may be hard to achieve experimentally.

To address this problem, we propose taking advantage of the
possibility that the decoherence will act independently on the
qubits. The rate of decoherence should thus depend on the energy gap
against local noise. We construct a class of stabilizer codes such
that encoded Hamiltonians are guaranteed to have a constant energy gap
against single-qubit excitations. These stabilizer codes are designed
so that adiabatic quantum computation with 4-local Hamiltonians is
universal for the encoded states. We illustrate the usefulness of
these codes for reducing decoherence using a noise model, proposed in
\cite{Childs_robustness}, in which each qubit independently couples to a photon
bath.

\section{Error Detecting Code}

To protect against decoherence we wish to create an energy gap against
single-qubit disturbances. To do this we use a quantum error
correcting code such that applying a single Pauli operator to any
qubit in a codeword will send this state outside of the codespace. Then we
add an extra term to the Hamiltonian which gives an energy penalty to
all states outside the codespace. Since we are only interested in
creating an energy penalty for states outside the codespace, only the
fact that an error has occurred needs to be detectable. Since we are
not actively correcting errors, it is not necessary for distinct
errors to be distinguishable. In this sense, our code is not truly an
error correcting code but rather an error \emph{detecting} code. Such
passive error correction is similar in spirit to ideas suggested for
the circuit model in \cite{Bacon}.

It is straightforward to verify that the 4-qubit code
\begin{eqnarray}
\label{zero_logical}
\ket{0_L} & = & \frac{1}{2} \left( \ket{0000} + i\ket{0011}
+i\ket{1100} + \ket{1111} \right)
\\
\label{one_logical}
\ket{1_L} & = & \frac{1}{2} \left( -\ket{0101} + i\ket{0110} +
i\ket{1001} - \ket{1010} \right)
\end{eqnarray}
satisfies the error-detection requirements, namely
\begin{equation}
\label{detection}
\bra{0_L} \sigma \ket{0_L} = \bra{1_L} \sigma \ket{1_L} = \bra{0_L}
\sigma \ket{1_L} = 0
\end{equation}
where $\sigma$ is any of the three Pauli operators acting on one
qubit. Furthermore, the following 2-local operations act as encoded
Pauli $X$, $Y$, and $Z$ 
operators.
\begin{equation}
\label{logical_operators}
\begin{array}{lll}
X_L & = & Y \otimes I \otimes Y \otimes I
\\
Y_L & = & -I \otimes X \otimes X \otimes I
\\
Z_L & = & Z \otimes Z \otimes I \otimes I
\end{array}
\end{equation}
That is,
\[
\begin{array}{llllll}
X_L \ket{0_L} & = & \ket{1_L}, & X_L \ket{1_L}   & = & \ket{0_L}, \\
Y_L \ket{0_L} & = & i \ket{1_L}, & Y_L \ket{1_L} & = & - i \ket{0_L}, \\
Z_L \ket{0_L} & = & \ket{0_L}, & Z_L \ket{1_L}   & = & - \ket{1_L}.
\end{array}
\]
An arbitrary state of a single qubit $\alpha \ket{0} + \beta \ket{1}$
is encoded as $\alpha \ket{0_L} + \beta \ket{1_L}$. 

Starting with an arbitrary 2-local Hamiltonian $H$ on $N$ bits, we
obtain a new fault tolerant Hamiltonian on $4N$ bits by the following
procedure. An arbitrary 2-local Hamiltonian can be written as a sum of
tensor products of pairs of Pauli matrices acting on different
qubits. After writing out $H$ in this way, make the following
replacements
\[
\begin{array}{llll}
I \to I^{\otimes 4}, & X \to X_L, & Y \to Y_L, & Z \to Z_L
\end{array}
\] 
to obtain a new 4-local Hamiltonian $H_{SL}$ acting on $4N$
qubits. The total fault tolerant Hamiltonian $H_S$ is
\begin{equation}
\label{total}
H_S = H_{SL} + H_{SP}
\end{equation}
where $H_{SP}$ is a sum of penalty terms, one acting on each encoded
qubit, providing an energy penalty of at least $E_p$ for going outside
the code space. We use the subscript $S$ to indicate that the
Hamiltonian acts on the system, as opposed to the environment, which
we introduce later. Note that $H_{SL}$ and $H_{SP}$ commute, and thus
they share a set of simultaneous eigenstates.

If the ground space of $H$ is spanned by  $\ket{\psi^{(1)}} \ldots
\ket{\psi^{(m)}}$ then the ground space of $H_S$ is spanned by the
encoded states $\ket{\psi_L^{(1)}} \ldots
\ket{\psi_L^{(m)}}$. Furthermore, the penalty terms provide an energy
gap against 1-local noise which does not shrink as the size of the
computation grows.

The code described by equations \ref{zero_logical} and
\ref{one_logical} can be obtained using the stabilizer formalism
\cite{Gottesman_thesis, Nielsen_Chuang}. In this formalism, a quantum code is not
described by explicitly specifying a set of basis states for the code
space. Rather, one specifies the generators of the stabilizer group for
the codespace. Let $G_n$ be the Pauli group on $n$ qubits (\emph{i.e.}
the set of all tensor products of $n$ Pauli operators with 
coefficients of $\pm 1$ or $\pm i$). The stabilizer group of a
codespace $C$ is the subgroup $S$ of $G_n$ such that $x \ket{\psi} =
\ket{\psi}$ for any $x \in S$ and any $\ket{\psi} \in C$.

A $2^k$ dimensional codespace over $n$ bits can be specified by
choosing $n-k$ independent commuting generators for the stabilizer
group $S$. By independent we mean that no generator
can be expressed as a product of others. In our case we are encoding a
single qubit using 4 qubits, thus $k=1$ and $n=4$, and we need 3
independent commuting generators for $S$. 

To satisfy the orthogonality conditions, listed in equation
\ref{detection}, which are necessary for error detection, it suffices
for each Pauli operator on a given qubit to anticommute with at least
one of the generators of the stabilizer group. The generators
\begin{eqnarray}
\label{generators}
g_1 & = & X \otimes X \otimes X \otimes X \nonumber \\
g_2 & = & Z \otimes Z \otimes Z \otimes Z \nonumber \\
g_3 & = & X \otimes Y \otimes Z \otimes I
\end{eqnarray}
satisfy these conditions, and generate the stabilizer group for
the code given in equations \ref{zero_logical} and \ref{one_logical}.

Adding one term of the form
\begin{equation}
\label{penalty}
H_p = -E_p (g_1+g_2+g_3)
\end{equation}
to the encoded Hamiltonian for each encoded qubit yields an energy
penalty of at least $E_p$ for any state outside the codespace.

2-local encoded operations are optimal. None of the encoded
operations can be made 1-local, because they would then have the same
form as the errors we are trying to detect and penalize. Such an
operation would not commute with all of the generators. 

\section{Noise Model}

Intuitively, one expects that providing an energy gap against a Pauli
operator applied to any qubit protects against 1-local noise. We
illustrate this using a model of decoherence proposed in
\cite{Childs_robustness}. In this model, the quantum computer is a set of
spin-$1/2$ particles weakly coupled to a large photon bath. The
Hamiltonian for the combined system is
\[
H = H_S + H_E + \lambda V,
\]
where $H_S(t)$ is the adiabatic Hamiltonian that implements the
algorithm by acting only on the spins, $H_E$ is the Hamiltonian which
acts only on the photon bath, and $\lambda V$ is a weak coupling
between the spins and the photon bath. Specifically, $V$ is assumed to
take the form
\[
V = \sum_i \int_0^\infty \ud \omega \left[ g(\omega) a_\omega
  \sigma_+^{(i)} + g^*(\omega)a_\omega^\dag \sigma_-^{(i)} \right],
\]
where $\sigma_{\pm}^{(i)}$ are raising and lowering operators for the
$i$th spin, $a_\omega$ is the annihilation operator for the photon
mode with frequency $\omega$, and $g(\omega)$ is the spectral density. 

From this premise Childs \emph{et al.} obtain the following master
equation 
\begin{equation}
\label{master_equation}
\frac{\ud \rho}{\ud t} = -i[H_S,\rho]-\sum_{a,b} M_{ab} \ 
\mathcal{E}_{ab}(\rho)
\end{equation}
where
\begin{eqnarray*}
M_{ab} & = &  \sum_i \left[ N_{ba}
  |g_{ba}|^2 \bra{a} \sigma_-^{(i)} \ket{b} \bra{b} \sigma_+^{(i)}
  \ket{a} \right. \\
  & & \left. + (N_{ab}+1) |g_{ab}|^2 \bra{b} \sigma_-^{(i)} \ket{a}
      \bra{a} \sigma_+^{(i)} \ket{b} \right]
\end{eqnarray*}
is a scalar,
\[
\mathcal{E}_{ab}(\rho) = \ket{a} \bra{a} \rho
+ \rho \ket{a} \bra{a} - 2 \ket{b} \bra{a} \rho \ket{a} \bra{b}
\]
is an operator, $\ket{a}$ is the instantaneous eigenstate of $H_S$
with energy $\omega_a$,
\[
N_{ba} = \frac{1}{\exp \left[ \beta (\omega_b-\omega_a)\right]-1}
\]
is the Bose-Einstein distribution at temperature $1/\beta$, and
\begin{equation}
\label{g}
g_{ba} = \left\{ \begin{array}{ll}
\lambda g(\omega_b - \omega_a), & \omega_b > \omega_a, \\
0, & \omega_b \leq \omega_a. \end{array} \right.
\end{equation}

Suppose that we encode the original $N$-qubit Hamiltonian as a
$4N$-qubit Hamiltonian as described above. As stated in equation
\ref{total}, the total spin Hamiltonian $H_S$ on $4N$ spins
consists of the encoded version $H_{SL}$ of the original Hamiltonian
$H_S$ plus the penalty terms $H_{SP}$.  

Most adiabatic quantum computations use an initial Hamiltonian with an
eigenvalue gap of order unity, independent of problem size. In such
cases, a nearly pure initial state can be achieved at constant
temperature. Therefore, we'll make the approximation that the spins
start in the pure ground state of the initial Hamiltonian, which we'll
denote $\ket{0}$. Then we can use equation \ref{master_equation} to
examine $\ud \rho/ \ud t$ at $t=0$. Since the initial state is $\rho =
\ket{0} \bra{0}$, $\mathcal{E}_{ab}(\rho)$ is zero unless $\ket{a} =
\ket{0}$. The master equation at $t=0$ is therefore 
\begin{equation}
\label{intermediate}
\left. \frac{\ud \rho}{\ud t}\right|_{t=0} = -i[H_S,\rho]-\sum_b
M_{0b} \ \mathcal{E}_{0b}(\rho).
\end{equation}

$H_{SP}$ is given by a sum of terms of the form \ref{penalty}, and it
commutes with $H_{SL}$. Thus, $H_S$ and $H_{SP}$ share a complete set
of simultaneous eigenstates. The eigenstates of $H_S$ can thus be
separated into those which are in the codespace $C$ (\emph{i.e.} the
ground space of $H_{SP}$) and those which are in the orthogonal space
$C^\perp$. The ground state $\ket{0}$ is in the codespace. $M_{0b}$
will be zero unless $\ket{b} \in C^\perp$, because $\sigma_{\pm} = (X
\pm i Y)/2$, and any Pauli operator applied to a single bit takes us
from $C$ to $C^\perp$. Equation \ref{intermediate} therefore becomes  
\begin{equation}
\label{master_equation2}
\left. \frac{\ud \rho}{\ud t} \right|_{t=0} = -i
     [H_S,\rho]+\sum_{b \in C^\perp} M_{0b} \ \mathcal{E}_{0b}(\rho)
\end{equation}

Since $\ket{0}$ is the ground state, $\omega_b \geq \omega_0$, thus
equation \ref{g} shows that the terms in $M_{0b}$ proportional to
$|g_{0b}|^2$ will vanish, leaving only
\[
M_{0b} = \sum_i N_{b0} |g_{b0}|^2 \bra{0}\sigma_-^{(i)}\ket{b} \bra{b}
\sigma_+^{(i)} \ket{0}.
\]

Now let's examine $N_{b0}$.
\[
\omega_b - \omega_0 = \bra{b} (H_{SL} + H_{SP}) \ket{b} - \bra{0}
(H_{SL} + H_{SP}) \ket{0}.
\] 
$\ket{0}$ is in the ground space of $H_{SL}$, thus 
\[
\bra{b} H_{SL} \ket{b} - \bra{0} H_{SL} \ket{0} \geq 0,
\]
and so
\[
\omega_b - \omega_0 \geq \bra{b} 
H_{SP} \ket{b} - \bra{0} H_{SP} \ket{0}.
\]
Since $\ket{b} \in C^\perp$ and $\ket{0} \in C$, 
\[
\bra{b} H_{SP} \ket{b} - \bra{0} H_{SP} \ket{0}
= E_p,
\]
thus $\omega_b - \omega_0 \geq E_p$.

A sufficiently large $\beta E_p$ will make $N_{ba}$ small enough that
the term $\sum_{b \in C^\perp} M_{0b} \mathcal{E}(\rho)$ can be
neglected from the master equation, leaving
\[
\left. \frac{\ud \rho}{\ud t} \right|_{t=0} \approx - i [H_S,\rho]
\]
which is just Schr\"odinger's equation with a Hamiltonian equal to
$H_S$ and no decoherence. Note that the preceding derivation did
not depend on the fact that $\sigma_\pm^{(i)}$ are raising and
lowering operators, but only on the fact that they act on a single
qubit and can therefore be expressed as a linear combination of Pauli
operators.

$N_{b0}$ is small but nonzero. Thus, after a sufficiently long time,
the matrix elements of $\rho$ involving states other than $\ket{0}$
will become non-negligible and the preceding picture will break down. How
long the computation can be run before this happens depends on the
magnitude of $\sum_{b \in C^\perp} M_{ob} \mathcal{E}(\rho)$, which
shrinks exponentially with $E_p/T$ and grows only polynomially
with the number of qubits $N$. Thus it should be sufficient for
$1/T$ to grow logarithmically with the problem size for the noise due
to the terms present in equation \ref{master_equation} to be supressed. In
contrast, one expects that if the Hamiltonian had only an inverse
polynomial gap against 1-local noise, the temperature would need to
shrink polynomially rather than logarithmically. 

One should note that equation \ref{master_equation} is derived by truncating
at second order in the coupling between the system and bath. Thus, at
sufficiently long timescales, higher order couplings may become
relevant. Physical 1-local terms can give rise to $k$-local virtual
terms at $k\th$ order in perturbation theory, thus protecting against
these may require an extension of the technique present
here. Nevertheless, the technique presented here protects against the
lowest order noise terms, which should be the largest ones provided
that the coupling to the environment is weak.

\section{Higher Weight Errors}

Now that we know how to obtain a constant gap against 1-local noise,
we may ask whether the same is possible for 2-local noise. To 
accomplish this we need to find a stabilizer group such that any pair
of Pauli operators on two bits anticommutes with at least one of the
generators. This is exactly the property satisfied by the
standard\cite{Nielsen_Chuang} 5-qubit stabilizer code, whose stabilizer group
is generated by
\begin{eqnarray}
\label{fivegenerators}
g_1 & = & X \otimes Z \otimes Z \otimes X \otimes I \nonumber \\
g_2 & = & I \otimes X \otimes Z \otimes Z \otimes X \nonumber \\
g_3 & = & X \otimes I \otimes X \otimes Z \otimes Z \nonumber \\
g_4 & = & Z \otimes X \otimes I \otimes X \otimes Z. 
\end{eqnarray}
The codewords for this code are
\begin{eqnarray*}
\ket{0_L} & = & \frac{1}{4} \left[ \ \ket{00000} + \ket{10010} +
  \ket{01001} +\ket{10100} \right. \\
& & + \ket{01010} - \ket{11011} - \ket{00110} - \ket{11000} \\
& & - \ket{11101} -\ket{00011} - \ket{11110} - \ket{01111} \\
& & \left. - \ket{10001} - \ket{01100} - \ket{10111} + \ket{00101}
\  \right]
\end{eqnarray*}
\begin{eqnarray*}
\ket{1_L} & = & \frac{1}{4} \left[ \ \ket{11111} + \ket{01101} +
  \ket{10110} + \ket{01011} \right. \\
& & + \ket{10101} - \ket{00100} - \ket{11001} - \ket{00111} \\
& & - \ket{00010} - \ket{11100} - \ket{00001} - \ket{10000} \\
& & \left. - \ket{01110} - \ket{10011} - \ket{01000} + \ket{11010}
\ \right]. 
\end{eqnarray*}
The encoded Pauli operations for this code are conventionally
expressed as
\begin{eqnarray*}
X_L & = & X \otimes X \otimes X \otimes X \otimes X \\
Y_L & = & Y \otimes Y \otimes Y \otimes Y \otimes Y \\
Z_L & = & Z \otimes Z \otimes Z \otimes Z \otimes Z.
\end{eqnarray*}
However, multiplying these encoded operations by members of the
stabilizer group doesn't affect their action on the codespace. Thus we
obtain the following equivalent set of encoded operations.
\begin{eqnarray}
\label{five_logical}
X_L & = & -X \otimes I \otimes Y \otimes Y \otimes I  \nonumber \\
Y_L & = & -Z \otimes Z \otimes I \otimes Y \otimes I \nonumber \\
Z_L & = & -Y \otimes Z \otimes Y \otimes I \otimes I
\end{eqnarray}
These operators are all 3-local. This is the best that can be hoped
for, because the code protects against 2-local operations and
therefore any 2-local operation must anticommute with at least one of
the generators.

Besides increasing the locality of the encoded operations, one can
seek to decrease the number of qubits used to construct the
codewords. The quantum singleton bound\cite{Nielsen_Chuang} shows that the
five qubit code is already optimal and cannot be improved in this
respect.

The distance $d$ of a quantum code is the minimum number of qubits of
a codeword which need to be modified before obtaining a nonzero inner
product with a different codeword. For example, applying $X_L$, which
is 3-local, to $\ket{0_L}$ of the 5-qubit code converts it into
$\ket{1_L}$, but applying any 2-local operator to any of the codewords
yields something outside the codespace. Thus the distance of the
5-qubit code is 3. Similarly the distance of our 4-qubit code is 2. To
detect $t$ errors a code needs a distance of $t+1$, and to correct $t$
errors, it needs a distance of $2t+1$.

The quantum singleton bound states that the distance of any quantum
code which uses $n$ qubits to encode $k$ qubits will satisfy
\begin{equation}
\label{singleton}
n-k \geq 2(d-1).
\end{equation}
To detect 2 errors, a code must have distance 3. A code which
encodes a single qubit with distance 3 must use at least 5 qubits, by
equation \ref{singleton}. Thus the 5-qubit code is optimal. To detect
1 error, a code must have distance 2. A code which encodes a single
qubit with distance 2 must have at least 3 qubits, by equation
\ref{singleton}. Thus it appears possible that our 4-qubit code is not
optimal. However, no 3-qubit stabilizer code can detect all
single-qubit errors, which we show as follows.

The stabilizer group for a 3-qubit code would have two independent
generators, each being a tensor product of 3 Pauli operators. 
\begin{eqnarray*}
g_1 & = & \sigma_{11} \otimes \sigma_{12} \otimes \sigma_{13} \\
g_2 & = & \sigma_{21} \otimes \sigma_{22} \otimes \sigma_{23}
\end{eqnarray*} \\
These must satisfy the following two conditions: (1) they commute, and
(2) an $X, Y$, or $Z$ on any of the three qubits anticommutes with at
least one of the generators. This is impossible, because condition (2)
requires $\sigma_{1i} \neq \sigma_{2i} \neq I$ for each
$i=1,2,3$. In this case $g_1$ and $g_2$ anticommute.

The stabilizer formalism describes most but not all currently known
quantum error correcting codes. We do not know whether a 3-qubit code
which detects all single-qubit errors while still maintaining 2-local
encoded operations can be found by going outside the stabilizer
formalism. It may also be interesting to investigate whether there
exist computationally universal 3-local or 2-local adiabatic
Hamiltonians with a constant energy gap against local noise.

\chapter{DQC1-completeness of Jones Polynomials}
\label{Jones}

\section{Introduction}

It is known that evaluating a certain approximation to the Jones
polynomial for the plat closure of a braid is a BQP-complete
problem. That is, this problem exactly captures the power of the
quantum circuit model\cite{Freedman, Aharonov1, Aharonov2}. The one
clean qubit model is a model of quantum computation in which all but
one qubit starts in the maximally mixed state. One clean qubit
computers are believed to be strictly weaker than standard quantum
computers, but still capable of solving some classically intractable
problems \cite{Knill_DQC1}. Here we show that evaluating a certain
approximation to the Jones polynomial at a fifth root of unity for the
trace closure of a braid is a complete problem for the one clean qubit
complexity class. That is, a one clean qubit computer can approximate
these Jones polynomials in time polynomial in both the number of
strands and number of crossings, and the problem of simulating a one
clean qubit computer is reducible to approximating the Jones
polynomial of the trace closure of a braid.

\section{One Clean Qubit}
\label{DQC1}

The one clean qubit model of quantum computation originated as an
idealized model of quantum computation on highly mixed initial states,
such as appear in NMR implementations\cite{Knill_DQC1, Ambainis_DQC1}. In this
model, one is given an initial quantum state consisting of a
single qubit in the pure state $\ket{0}$, and $n$ qubits in the
maximally mixed state. This is described by the density matrix
\[
\rho = \ket{0}\bra{0} \otimes \frac{I}{2^n}.
\]

One can apply any polynomial-size quantum circuit to $\rho$, and
then measure the first qubit in the computational basis. Thus, if the
quantum circuit implements the unitary transformation $U$, the
probability of measuring $\ket{0}$ will be 
\begin{equation}
\label{experiment}
p_0 = \tr[(\ket{0}\bra{0}\otimes I) U \rho U^\dag] =
2^{-n} \tr[(\ket{0}\bra{0} \otimes I) U (\ket{0}\bra{0} \otimes I)  U^\dag].
\end{equation}

Computational complexity classes are typically described using
decision problems, that is, problems which admit yes/no
answers. This is mathematically convenient, and the implications for
the complexity of non-decision problems are usually straightforward to
obtain (\emph{cf.} \cite{Papadimitriou}). The one clean qubit
complexity class consists of the decision problems which can be solved
in polynomial time by a one clean qubit machine with correctness 
probability of at least $2/3$. The experiment described in 
equation  \ref{experiment} can be repeated polynomially many
times. Thus, if $p_1 \geq 1/2 + \epsilon$ for instances to which the
answer is yes, and $p_1 \leq 1/2 - \epsilon$ otherwise, then by
repeating the experiment $\mathrm{poly}(1/\epsilon)$ times and taking
the majority vote one can achieve $2/3$ probability of
correctness. Thus, as long as $\epsilon$ is at least an inverse
polynomial in the problem size, the problem is contained in the one
clean qubit complexity class. Following \cite{Knill_DQC1}, we will refer to
this complexity class as DQC1.

A number of equivalent definitions of the one clean qubit complexity
class can be made. For example, changing the pure part of the
initial state and the basis in which the final measurement is
performed does not change the resulting complexity class. Less
trivially, allowing logarithmically many clean qubits results in
the same class, as discussed below. It is essential that on a given
copy of $\rho$, measurements are performed only at the end of the
computation. Otherwise, one could obtain a pure state
by measuring $\rho$ thus making all the qubits ``clean'' and
re-obtaining BQP. Remarkably, it is not necessary to have even one
fully polarized qubit to obtain the class DQC1. As shown in
\cite{Knill_DQC1}, a single partially polarized qubit suffices.

In the original definition\cite{Knill_DQC1} of DQC1 it is assumed that a
classical computer generates the quantum circuits to be applied to the
initial state $\rho$. By this definition DQC1 automatically contains
the complexity class P. However, it is also interesting to consider a
slightly weaker one clean qubit model, in which the classical computer
controlling the quantum circuits has only the power of NC1. The
resulting complexity class appears to have the interesting property
that it is incomparable to P. That is, it is not contained in P nor
does P contain it. We suspect that our algorithm and hardness proof
for the Jones polynomial carry over straightforwardly to this
NC1-controlled one clean qubit model. However, we have not pursued
this point.

Any $2^n \times 2^n$ unitary matrix can be decomposed as a linear
combination of $n$-fold tensor products of Pauli matrices. As
discussed in \cite{Knill_DQC1}, the problem of estimating a coefficient in
the Pauli decomposition of a quantum circuit to polynomial accuracy is
a DQC1-complete problem. Estimating the normalized trace of a quantum
circuit is a special case  of this, and it is also DQC1-complete. This
point is discussed in \cite{Shepherd}. To make our presentation
self-contained, we will sketch here a proof that trace estimation is
DQC1-complete. Technically, we should consider the decision problem of
determining whether the trace is greater than a given
threshold. However, the trace estimation problem is easily reduced to
its decision version by the method of binary search, so we will
henceforth ignore this point.

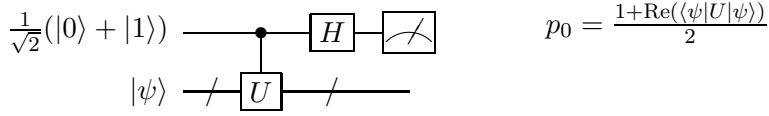
\begin{figure}
\begin{center}
\[
\begin{array}{lll}
\mbox{\Qcircuit @C=1em @R=.7em {
  \lstick{\frac{1}{\sqrt{2}}(\ket{0}+\ket{1})} & \qw & \ctrl{1} & \gate{H} & \meter \\
  \lstick{\ket{\psi}} & {/} \qw  & \gate{U} & {/} \qw & \qw & }} &
  \quad & p_0 = \frac{1+\mathrm{Re}(\bra{\psi}U\ket{\psi})}{2}
\end{array} \]
\caption{\label{Hadamard} This circuit implements the Hadamard
  test. A horizontal line represents a qubit. A horizontal line with a
  slash through it represents a register of multiple qubits. The
  probability $p_0$ of measuring $\ket{0}$ is as shown  above. Thus,
  one can obtain the real part of $\bra{\psi}U\ket{\psi}$ 
  to precision $\epsilon$ by making $O(1/\epsilon^2)$ measurements and
  counting what fraction of the measurement outcomes are
  $\ket{0}$. Similarly, if the control bit is instead initialized to 
  $\frac{1}{\sqrt{2}} (\ket{0} - i \ket{1})$, one can estimate the
  imaginary part of $\bra{\psi}U\ket{\psi}$.}
\end{center}
\end{figure}

First we'll show that trace estimation is contained in DQC1. Suppose we
are given a quantum circuit on $n$ qubits which consists of
polynomially many gates from some finite universal gate set. Given a
state $\ket{\psi}$ of $n$ qubits, there is a standard technique for
estimating $\bra{\psi} U \ket{\psi}$, called the Hadamard
test\cite{Aharonov1}, as shown in figure \ref{Hadamard}. Now suppose
that we use the circuit from figure \ref{Hadamard}, but choose
$\ket{\psi}$ uniformly at random from the $2^n$ computational basis
states. Then the probability of getting outcome $\ket{0}$ for a given
measurement will be
\[
p_0 = \frac{1}{2^n} \sum_{x \in \{0,1\}^n} \frac{1 +
  \mathrm{Re}(\bra{x} U \ket{x})}{2} = \frac{1}{2} +
\frac{\mathrm{Re}(\tr \ U)}{2^{n+1}}.
\]
Choosing $\ket{\psi}$ uniformly at random from the $2^n$ computational
basis states is exactly the same as inputting the density matrix $I/2^n$
to this register. Thus, the only clean qubit is the control
qubit. Trace estimation is therefore achieved in the one clean qubit
model by converting the given circuit for $U$ into a circuit for
controlled-$U$ and adding Hadamard gates on the control bit. One can
convert a circuit for $U$ into a circuit for controlled-$U$ by
replacing each gate $G$ with a circuit for controlled-$G$. The
overhead incurred is thus bounded by a constant factor
\cite{Nielsen_Chuang}. 

Next we'll show that trace estimation is hard for DQC1. Suppose we are
given a classical description of a quantum circuit implementing some
unitary transformation $U$ on $n$ qubits. As shown in equation
\ref{experiment}, the probability of obtaining outcome $\ket{0}$ from
the one clean qubit computation of this circuit is proportional to the
trace of the non-unitary operator $(\ket{0}\bra{0} \otimes I) U
(\ket{0}\bra{0} \otimes I)  U^\dag$, which acts on $n$
qubits. Estimating this can be achieved by estimating the trace of
\[
U' = \begin{array}{l} \Qcircuit @C=1em @R=.7em {
 &  \qw   & \multigate{1}{U^\dag} & \ctrl{2} & \multigate{1}{U} & \ctrl{3} & \qw    & \qw \\
 & {/}\qw & \ghost{U^\dag}        & \qw      & \ghost{U}        & \qw      & {/}\qw & \qw \\
 &  \qw   & \qw                   & \targ    & \qw              & \qw      & \qw    & \qw \\
 &  \qw   & \qw                   & \qw      & \qw              & \targ    & \qw    & \qw
} \end{array}
\]
which is a unitary operator on $n+2$ qubits. This suffices because
\begin{equation}
\label{tracform}
\tr[(\ket{0}\bra{0} \otimes I) U (\ket{0}\bra{0} \otimes I)
  U^\dag] = \frac{1}{4} \tr[ U' ].
\end{equation}
To see this, we can think in terms of the computational basis:
\[
\tr[U'] = \sum_{x \in \{0,1\}^n} \bra{x} U' \ket{x}.
\]
If the first qubit of $\ket{x}$ is $\ket{1}$, then
the rightmost CNOT in $U'$ will flip the lowermost qubit. The resulting state
will be orthogonal to $\ket{x}$ and the corresponding  matrix element
will not contribute to the trace. Thus this CNOT gate simulates the
initial projector $\ket{0}\bra{0} \otimes I$ in equation
\ref{tracform}. Similarly, the other CNOT in $U'$ simulates the other
projector in equation \ref{tracform}.

The preceding analysis shows that, given a description of a quantum
circuit implementing a unitary transformation $U$ on $n$-qubits, the
problem of  approximating $\frac{1}{2^n} \tr \ U$ to within $\pm
\frac{1}{\mathrm{poly}(n)}$ precision is DQC1-complete.

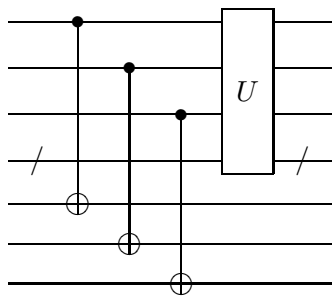
\begin{figure}
\begin{center}
\mbox{ 
\Qcircuit @C=1em @R=.7em {
  & \qw    & \ctrl{4} & \qw      & \qw      & \multigate{3}{U} & \qw    & \qw \\
  & \qw    & \qw      & \ctrl{4} & \qw      & \ghost{U}        & \qw    & \qw \\
  & \qw    & \qw      & \qw      & \ctrl{4} & \ghost{U}        & \qw    & \qw \\
  & {/}\qw & \qw      & \qw      & \qw      & \ghost{U}        & {/}\qw & \qw \\ 
  & \qw    & \targ    & \qw      & \qw      & \qw              & \qw    & \qw \\
  & \qw    & \qw      & \targ    & \qw      & \qw              & \qw    & \qw \\
  & \qw    & \qw      & \qw      & \targ    & \qw              & \qw    & \qw
} }
\caption{\label{ancillas} Here CNOT gates are used to simulate 3 clean
  ancilla qubits.}
\end{center}
\end{figure}

Some unitaries may only be efficiently implementable using ancilla
bits. That is, to implement $U$ on $n$-qubits using a quantum circuit,
it may be most efficient to construct a circuit on $n+m$ qubits which
acts as $U \otimes I$, provided that the $m$ ancilla qubits are all
initialized to $\ket{0}$. These ancilla qubits are used as work bits
in intermediate steps of the computation. To estimate the trace of
$U$, one can construct a circuit $U_a$ on $n+2m$ qubits by adding CNOT
gates controlled by the $m$ ancilla qubits and acting on $m$ extra
qubits, as shown in figure \ref{ancillas}. This simulates the presence
of $m$ clean ancilla qubits, because if any of the ancilla qubits is
in the $\ket{1}$ state then the CNOT gate will flip the corresponding
extra qubit, resulting in an orthogonal state which will not
contribute to the trace.

With one clean qubit, one can estimate the trace of $U_a$ to a
precision of $\frac{2^{n+2m}}{\mathrm{poly}(n,m)}$. By construction, 
$\tr[U_a] = 2^m \tr[U]$. Thus, if $m$ is logarithmic in $n$, then
one can obtain $\tr[U]$ to precision $\frac{2^n}{\mathrm{poly}(n)}$,
just as can be obtained for circuits not requiring ancilla
qubits. This line of reasoning also shows that the $k$-clean qubit
model gives rise to the same complexity class as the one clean qubit
model, for any constant $k$, and even for $k$ growing logarithmically
with $n$.

It seems unlikely that the trace of these exponentially large unitary
matrices can be estimated to this precision on a classical computer in
polynomial time. Thus it seems unlikely that DQC1 is contained in
P. (For more detailed analysis of this point see \cite{Datta}.)
However, it also seems unlikely that DQC1 contains all of BQP. In
other words, one clean qubit computers seem to provide exponential
speedup over classical computation for some problems despite being
strictly weaker than standard quantum computers.

\section{Jones Polynomials}

A knot is defined to be an embedding of the circle in $\mathbb{R}^3$
considered up to continuous transformation (isotopy). More
generally, a link is an embedding of one or more circles in
$\mathbb{R}^3$ up to isotopy. In an oriented knot or link, one of
the two possible traversal directions is chosen for each circle.
Some examples of knots and links are shown in figure \ref{knots}. One
of the fundamental tasks in knot theory is, given two representations
of knots, which may appear superficially different, determine whether
these both represent the same knot. In other words, determine whether
one knot can be deformed into the other without ever cutting the
strand.

\begin{figure}
\begin{center}
\includegraphics[width=0.43\textwidth]{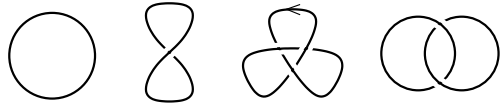}
\caption{\label{knots} Shown from left to right are the unknot,
  another representation of the unknot, an oriented trefoil knot, and
  the Hopf link. Broken lines indicate undercrossings.}
\end{center}
\end{figure} 

Reidemeister showed in 1927 that two knots are the same if and only if
one can be deformed into the other by some sequence constructed from
three elementary moves, known as the Reidemeister moves, shown in
figure \ref{Reidemeister}. This reduces the problem of distinguishing
knots to a combinatorial problem, although one for which no efficient
solution is known. In some cases, the sequence of Reidemeister moves
needed to show equivalence of two knots involves intermediate steps
that increase the number of crossings. Thus, it is very difficult to
show upper bounds on the number of moves necessary. The most
thoroughly studied knot equivalence problem is the problem of deciding
whether a given knot is equivalent to the unknot. Even showing the
decidability of this problem is highly nontrivial. This was achieved
by Haken in 1961\cite{Haken}. In 1998 it was shown by Hass, Lagarias,
and Pippenger that the problem of recognizing the unknot is
contained in NP\cite{Hass}.

\begin{figure}
\begin{center}
\includegraphics[width=0.85\textwidth]{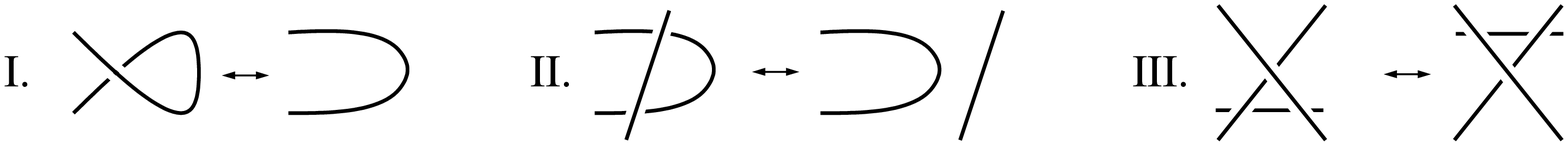}
\caption{\label{Reidemeister} Two knots are the same
  if and only if one can be deformed into the other by some sequence
  of the three Reidemeister moves shown above.}
\end{center}
\end{figure} 

A knot invariant is any function on knots which is invariant under the
Reidemeister moves. Thus, a knot invariant always takes the same
value for different representations of the same knot, such as the two
representations of the unknot shown in figure \ref{knots}. In general,
there can be distinct knots which a knot invariant fails to
distinguish. 

One of the best known knot invariants is the Jones polynomial,
discovered in 1985 by Vaughan Jones\cite{Jones}. To any oriented knot
or link, it associates a Laurent polynomial in the variable
$t^{1/2}$. The Jones polynomial has a degree in $t$ which grows at
most linearly with the number of crossings in the link. The
coefficients are all integers, but they may be exponentially
large. Exact evaluation of Jones polynomials at all but a few special
values of $t$ is \#P-hard\cite{Jaeger}. The Jones polynomial can be
defined recursively by a simple ``skein'' relation. However, for our
purposes it will be more convenient to use a definition in terms of a
representation of the braid group, as discussed below.

\begin{figure}[!htbp]
\begin{center}
\includegraphics[width=0.6\textwidth]{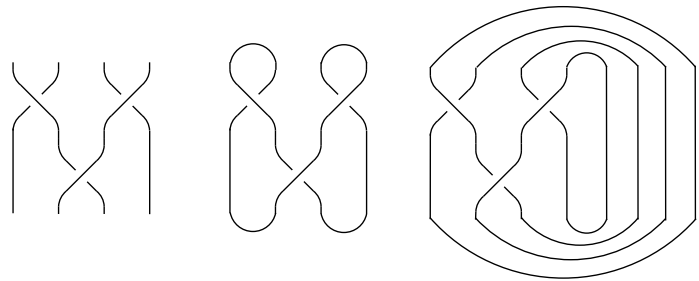}
\caption{\label{trace_plat} Shown from left to right are a braid, its
  plat closure, and its trace closure.}
\end{center}
\end{figure} 

To describe in more detail the computation of Jones polynomials we
must specify how the knot will be represented on the
computer. Although an embedding of a circle in $\mathbb{R}^3$ is a
continuous object, all the topologically relevant information about a
knot can be described in the discrete language of the braid
group. Links can be constructed from braids by joining the free
ends. Two ways of doing this are taking the plat closure and
the trace closure, as shown in figure \ref{trace_plat}. Alexander's
theorem states that any link can be constructed as the trace closure
of some braid. Any link can also be constructed as the plat closure of
some braid. This can be easily proven as a corollary to Alexander's
theorem, as shown in figure \ref{trace_to_plat}.

\begin{figure}[!htbp]
\begin{center}
\includegraphics[width=0.45\textwidth]{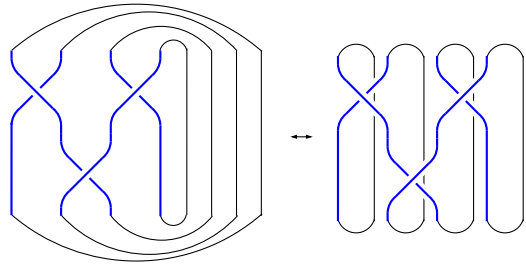}
\caption{\label{trace_to_plat} A trace closure of a braid on $n$
  strands can be converted to a plat closure of a braid on $2n$
  strands by moving the ``return'' strands into the braid.}
\end{center}
\end{figure} 

\begin{figure}[!htbp]
\begin{center}
\includegraphics[width=0.8\textwidth]{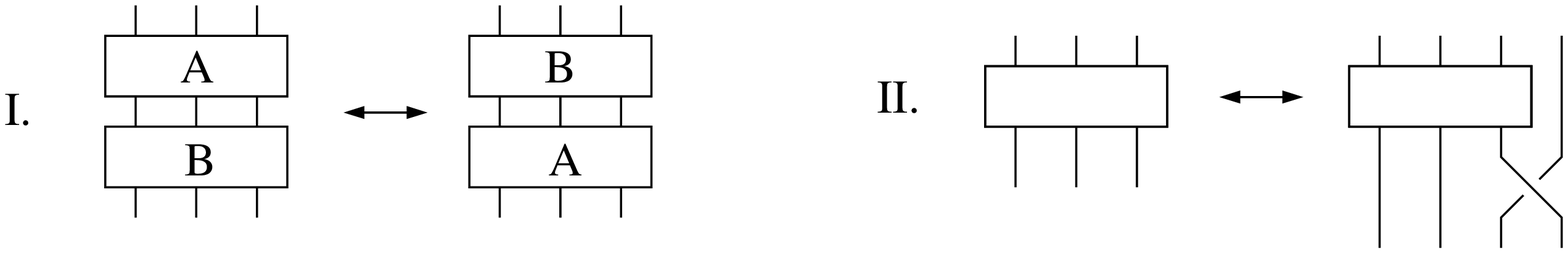}
\caption{\label{Markov} Shown are the two Markov moves. Here the
  boxes represent arbitrary braids. If a function on braids is
  invariant under these two moves, then the corresponding function on
  links induced by the trace closure is a link invariant.}
\end{center}
\end{figure} 

Given that the trace closure provides a correspondence between links
and braids, one may attempt to find functions on braids which yield
link invariants via this correspondence. Markov's theorem shows that a
function on braids will yield a knot invariant provided it is invariant
under the two Markov moves, shown in figure \ref{Markov}. Thus the
Markov moves provide an analogue for braids of the Reidemeister moves
on links. The constraints imposed by invariance under the Reidemeister
moves are enforced in the braid picture jointly by invariance under
Markov moves and by the defining relations of the braid group.

A linear function $f$ satisfying $f(AB) = f(BA)$ is called a
trace. The ordinary trace on matrices is one such function. Taking a
trace of a representation of the braid group yields a function on
braids which is invariant under Markov move I. If the trace and
representation are such that the resulting function is also invariant
under Markov move II, then a link invariant will result. The Jones
polynomial can be obtained in this way.

In \cite{Aharonov1}, Aharonov, \emph{et al.} show that an
additive approximation to the Jones polynomial of the plat or trace
closure of a braid at $t=e^{i 2 \pi/k}$ can be computed on a quantum
computer in time which scales polynomially in the number of strands
and crossings in the braid and in $k$. In \cite{Aharonov2, Wocjan}, it
is shown that for plat closures, this problem is BQP-complete. The
complexity of approximating the Jones polynomial for trace closures
was left open, other than showing that it is contained in BQP.

The results of \cite{Aharonov1, Aharonov2, Wocjan} reformulate and
generalize the previous results of Freedman \emph{et al.}
\cite{Freedman, Freedman2}, which show that certain
approximations of Jones polynomials are BQP-complete. The work of
Freedman \emph{et al.} in turn builds upon Witten's discovery of a
connection between Jones polynomials and topological quantum field
theory \cite{Witten}. Recently, Aharonov \emph{et al.} have
generalized further, obtaining an efficient quantum algorithm for
approximating the Tutte polynomial for any planar graph, at any point
in the complex plane, and also showing BQP-hardness at some points 
\cite{Aharonov3}. As special cases, the Tutte polynomial includes the
Jones polynomial, other knot invariants such as the HOMFLY polynomial,
and partition functions for some physical models such as the Potts
model.

The algorithm of Aharonov \emph{et al.} works by obtaining the Jones
polynomial as a trace of the path model representation of the braid
group. The path model representation is unitary for $t=e^{i 2\pi/k}$
and, as shown in \cite{Aharonov1}, can be efficiently implemented by
quantum circuits. For computing the trace closure of a braid the
necessary trace is similar to the ordinary matrix trace except that
only a subset of the diagonal elements of the unitary implemented by
the quantum circuit are summed, and there is an additional weighting
factor. For the plat closure of a braid the computation instead
reduces to evaluating a particular matrix element of the quantum
circuit. Aharonov \emph{et al.} also use the path model representation
in their proof of BQP-completeness.

Given a braid $b$, we know that the problem of approximating the Jones
polynomial of its plat closure is BQP-hard. By Alexander's
theorem, one can obtain a braid $b'$ whose trace closure is the same
link as the plat closure of $b$. The Jones polynomial depends only on
the link, and not on the braid it was derived from. Thus, one may ask
why this doesn't immediately imply that estimating the Jones
polynomial of the trace closure is a BQP-hard problem. The answer lies
in the degree of approximation. As discussed in section
\ref{conclusion}, the BQP-complete problem for plat closures is to
approximate the Jones polynomial to a certain precision which depends
exponentially on the number of strands in the braid. The number of
strands in $b'$ can be larger than the number of strands in $b$, hence
the degree of approximation obtained after applying Alexander's
theorem may be too poor to solve the original BQP-hard problem.

The fact that computing the Jones polynomial of the trace closure of a
braid can be reduced to estimating a generalized trace of a unitary
operator and the fact that trace estimation is DQC1-complete suggest a
connection between Jones polynomials and the one clean qubit
model. Here we find such a connection by showing that evaluating a
certain approximation to the Jones polynomial of the
trace closure of a braid at a fifth root of unity is
DQC1-complete. The main technical difficulty is obtaining the Jones 
polynomial as a trace over the entire Hilbert space rather than as a
summation of some subset of the diagonal matrix elements. To do this
we will not use the path model representation of the braid group, but
rather the Fibonacci representation, as described in the next section.

\section{Fibonacci Representation}
\label{Fibonacci}

The Fibonacci representation $\rho_F^{(n)}$ of the braid group $B_n$
is described in \cite{Kauffman} in the context of Temperley-Lieb
recoupling theory. Temperley-Lieb recoupling theory describes two
species of idealized ``particles'' denoted by $p$ 
and $*$. We will not delve into the conceptual and mathematical
underpinnings of Temperley-Lieb recoupling theory. For present
purposes, it will be sufficient to regard it as a formal procedure for
obtaining a particular unitary representation of the braid group whose
trace yields the Jones polynomial at $t = e^{i 2 \pi/5}$. Throughout
most of this paper it will be more convenient to express the Jones
polynomial in terms of $A = e^{-i 3 \pi/5}$, with $t$ defined by $t =
A^{-4}$. 

It is worth noting that the Fibonacci representation is a special case
of the path model representation used in \cite{Aharonov1}. The path
model representation applies when $t = e^{i 2 \pi/k}$ for any integer
$k$, whereas the Fibonacci representation is for $k=5$. The
relationship between these two representations is briefly discussed in
section \ref{rep_relations}. However, for the sake of making our
discussion self contained, we will derive all of our results directly
within the Fibonacci representation.

\begin{figure}
\begin{center}
\includegraphics[width=0.25\textwidth]{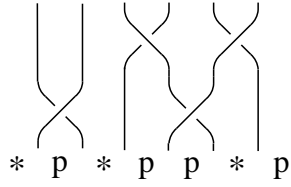}
\caption{\label{pstar} For an $n$-strand braid we can write a length
  $n+1$ string of $p$ and $*$ symbols across the base. The string may
  have no two $*$ symbols in a row, but can be otherwise arbitrary.}
\end{center}
\end{figure} 

\begin{figure}
\begin{center}
\includegraphics[width=0.5\textwidth]{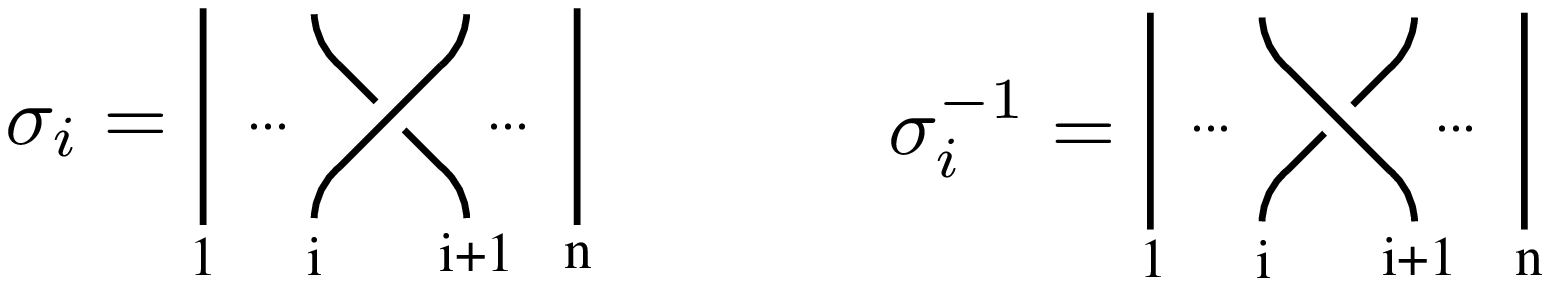}
\caption{\label{braids} $\sigma_i$ denotes the elementary crossing of
  strands $i$ and $i+1$. The braid group on $n$ strands $B_n$ is
  generated by $\sigma_1 \ldots \sigma_{n-1}$, which satisfy the
  relations $\sigma_i \sigma_j = \sigma_j \sigma_i$ for $|i-j|>1$ and
  $\sigma_{i+1} \sigma_i \sigma_{i+1} = \sigma_i \sigma_{i+1}
  \sigma_i$ for all $i$. The group operation corresponds to
  concatenation of braids.}
\end{center}
\end{figure}

Given an $n$-strand braid $b \in B_n$, we can write a length $n+1$
string of $p$ and $*$ symbols across the base as shown in figure
\ref{pstar}. These strings have the restriction that no two $*$
symbols can be adjacent. The number of such strings is $f_{n+3}$,
where $f_n$ is the $n\th$ Fibonacci number, defined so that $f_1 = 1$,
$f_2 = 1$, $f_3 = 2,\ldots$ Thus the formal linear combinations of
such strings form an $f_{n+3}$-dimensional vector space. For each $n$,
the Fibonacci representation $\rho_F^{(n)}$ is a homomorphism from
$B_n$ to the group of unitary linear transformations on this
space. We will describe the Fibonacci representation in terms of its
action on the elementary crossings which generate the braid group, as
shown in figure \ref{braids}.

The elementary crossings correspond to linear operations which mix
only those strings which differ by the symbol beneath the
crossing. The linear transformations have a local structure, so that
the coefficients for the symbol beneath the crossing to be changed or
unchanged depend only on that symbol and its two neighbors. For
example, using the notation of \cite{Kauffman}, 
\begin{equation}
\label{picto}
\includegraphics[width=0.35\textwidth]{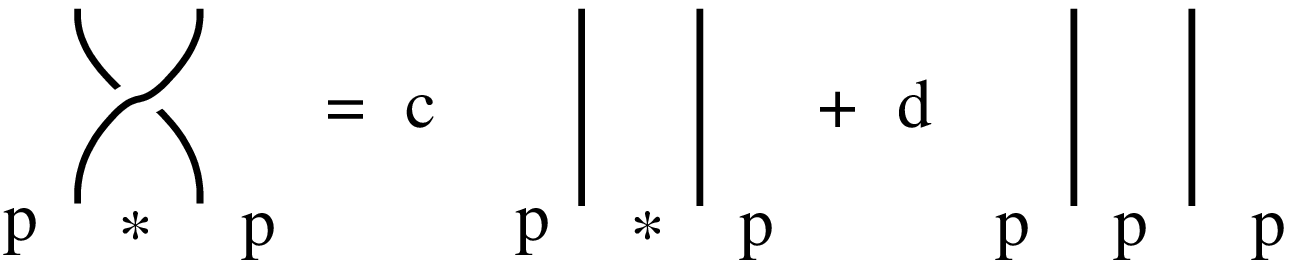}
\end{equation}
which means that the elementary crossing $\sigma_i$ corresponds to a
linear transformation which takes any string whose $i\th$ through 
$(i+2)\th$ symbols are $p*p$ to the coefficient $c$ times the same
string plus the coefficient $d$ times the same string with the $*$ at
the $(i+1)\th$ position replaced by $p$. (As shown in figure 9, the
$i\th$ crossing is over the $(i+1)\th$ symbol.) To compute the linear
transformation that the representation of a given braid applies to a
given string of symbols, one can write the symbols across the base of
the braid, and then apply rules of the form \ref{picto} until all the
crossings are removed, and all that remains are various coefficients
for different strings to be written across the base of a set of
straight strands.

For compactness, we will use $(p \widehat{*} p) = c(p * p) + d(p p p)$
as a shorthand for equation \ref{picto}. In this notation, the
complete set of rules is as follows.

\begin{eqnarray}
\label{rules}
(* \widehat{p} p) & = & a(*pp) \nonumber \\
(* \widehat{p} *) & = & b(*p*) \nonumber \\
(p \widehat{*} p) & = & c(p*p) + d(ppp) \nonumber \\
(p \widehat{p} *) & = & a(pp*) \nonumber \\
(p \widehat{p} p) & = & d(p*p) + e(ppp),
\end{eqnarray}
where
\begin{eqnarray}
\label{constants}
a & = & -A^4 \nonumber \\
b & = & A^8 \nonumber \\
c & = & A^8 \tau^2 - A^4 \tau \nonumber \\
d & = & A^8 \tau^{3/2} + A^4 \tau^{3/2} \nonumber \\
e & = & A^8 \tau - A^4 \tau^2 \nonumber \\
A & = & e^{-i3 \pi/5} \nonumber \\
\tau & = & 2/(1+\sqrt{5}).
\end{eqnarray}

Using these rules we can calculate any matrix from the Fibonacci
representation of the braid group. Notice that this is a reducible
representation. These rules do not allow the rightmost symbol or
leftmost symbol of the string to change. Thus the vector space
decomposes into four invariant subspaces, namely the subspace spanned
by strings which begin and end with $p$, and the $* \ldots *$,
$p \ldots *$, and $*\ldots p$ subspaces. As an example, we can use the
above rules to compute the action of $B_3$ on the $* \ldots p$
subspace. 

\begin{equation}
\label{examp}
\rho_{*p}^{(3)}(\sigma_1) = \left[ \begin{array}{cc}
b & 0 \\
0 & a
\end{array} \right]
\begin{array}{c}
\textrm{$*$p$*$p} \\
\textrm{$*$ppp}
\end{array}
\quad \quad \quad
\rho_{*p}^{(3)}(\sigma_2) = \left[ \begin{array}{cc}
c & d \\
d & e
\end{array} \right]
\begin{array}{c}
\textrm{$*$p$*$p} \\
\textrm{$*$ppp}
\end{array}
\end{equation}

In section \ref{proof} we prove that the Jones polynomial evaluated
at $t = e^{i 2 \pi/5}$ can be obtained as a weighted trace of the Fibonacci
representation over the $* \ldots *$ and $*\ldots p$ subspaces.

\section{Computing the Jones Polynomial in DQC1}
\label{containment}

As mentioned previously, the Fibonacci representation acts on the
vector space of formal linear combinations of strings of $p$ and $*$
symbols in which no two $*$ symbols are adjacent. The set of length
$n$ strings of this type, $P_n$, has $f_{n+2}$ elements, where $f_n$ is
the $n\th$ Fibonacci number: $f_1 = 1$, $f_2 = 1$, $f_3 = 2$, and so
on. As shown in section \ref{bijective}, one can construct a
bijective correspondence between these strings and the integers from
$0$ to $f_{n+2}-1$ as follows. If we think of $*$ as $1$ and $p$ as
$0$, then with a string $s_n s_{n-1} \ldots s_1$ we associate the
integer
\begin{equation}
\label{correspondence}
z(s) = \sum_{i=1}^n s_i f_{i+1}. 
\end{equation}
This is known as the Zeckendorf representation.

Representing integers as bitstrings by the usual method of place
value, we thus have a correspondence between the elements of $P_n$ and
the bitstrings of length $b = \lceil \log_2(f_{n+2}) \rceil$. This
correspondence will be a key element in computing the Jones polynomial
with a one clean qubit machine. Using a one clean qubit machine, one
can compute the trace of a unitary over the entire Hilbert space of
$2^n$ bitstrings. Using CNOT gates as above, one can also compute with
polynomial overhead the trace over a subspace whose dimension is a
polynomially large fraction of the dimension of the entire Hilbert
space. However, it is probably not possible in general for a one clean
qubit computer to compute the trace over subspaces whose dimension is
an exponentially small fraction of the dimension of the total Hilbert
space. For this reason, directly mapping the strings of $p$ and $*$
symbols to strings of $1$ and $0$ will not work. In contrast, the
correspondence described in equation \ref{correspondence} maps $P_n$
to a subspace whose dimension is at least half the dimension of the
full $2^b$-dimensional Hilbert space.

In outline, the DQC1 algorithm for computing the Jones polynomial works
as follows. Using the results described in section \ref{DQC1}, we will
think of the quantum circuit as acting on $b$ maximally mixed qubits
plus $O(1)$ clean qubits.  Thinking in terms of the computational
basis, we can say that the first $b$ qubits are in a uniform
probabilistic mixture of the $2^b$ classical bitstring states. By
equation \ref{correspondence}, most of these bitstrings correspond to
elements of $P_n$. In the Fibonacci representation, an elementary
crossing on strands $i$ and $i-1$ corresponds to a linear
transformation which can only change the value of the $i\th$ symbol in
the string of $p$'s and $*$'s. The coefficients for changing this
symbol or leaving it fixed depend only on the two neighboring
symbols. Thus, to simulate this linear transformation, we will use a
quantum circuit which extracts the $(i-1)\th$, $i\th$, and $(i+1)\th$
symbols from their bitstring encoding, writes them into an ancilla
register while erasing them from the bitstring encoding, performs the
unitary transformation prescribed by equation \ref{rules} on the
ancillas, and then transfers this symbol back into the bitstring
encoding while erasing it from the ancilla register. Constructing one
such circuit for each crossing, multiplying them together, and
performing DQC1 trace-estimation yields an approximation to the Jones
polynomial.

Performing the linear transformation demanded by equation \ref{rules}
on the ancilla register can be done easily by invoking gate set
universality (\emph{cf.} Solovay-Kitaev theorem \cite{Nielsen_Chuang}) since
it is just a three-qubit unitary operation. The harder steps are transferring
the symbol values from the bitstring encoding to the ancilla register and
back.

It may be difficult to extract an arbitrary symbol from
the bitstring encoding. However, it is relatively easy to extract the
leftmost ``most significant'' symbol, which determines whether the
Fibonacci number $f_n$ is present in the sum shown in equation
\ref{correspondence}. This is because, for a string $s$ of length $n$,
$z(s) \geq f_{n-1}$ if and only if the leftmost symbol is $*$. Thus,
starting with a clean $\ket{0}$ ancilla qubit, one can transfer the
value of the leftmost symbol into the ancilla as follows. First,
check whether $z(s)$ (as represented by a bitstring using place value)
is $\geq f_{n-1}$. If so flip the ancilla qubit. Then, conditioned on
the value of the ancilla qubit, subtract $f_{n-1}$ from the
bitstring. (The subtraction will be done modulo $2^b$ for
reversibility.) 

Any classical circuit can be made reversible with only constant
overhead. It thus corresponds to a unitary matrix which permutes the
computational basis. This is the standard way of implementing classical
algorithms on a quantum computer\cite{Nielsen_Chuang}. However, the resulting
reversible circuit may require clean ancilla qubits as work space in
order to be implemented efficiently. For a reversible circuit to be
implementable on a one clean qubit computer, it must be efficiently
implementable using at most logarithmically many clean
ancillas. Fortunately, the basic operations of arithmetic and
comparison for integers can all be done classically by NC1 circuits
\cite{Wegener}. NC1 is the complexity class for problems solvable by
classical circuits of logarithmic depth. As shown in \cite{Ambainis_DQC1},
any classical NC1 circuit can be converted into a reversible circuit
using only three clean ancillas. This is a consequence of Barrington's
theorem. Thus, the process described above for extracting the leftmost
symbol can be done efficiently in DQC1.

More specifically, Krapchenko's algorithm for adding two $n$-bit
numbers has depth $\lceil \log n \rceil + O(\sqrt{\log n})$
\cite{Wegener}. A lower bound of depth $\log n$ is also known, so this
is essentially optimal \cite{Wegener}. Barrington's construction
\cite{Barrington} yields a sequence of $2^{2 d}$ gates on 3 clean
ancilla qubits \cite{Ambainis_DQC1} to simulate a circuit of depth
$d$. Thus we obtain an addition circuit which has quadratic size (up
to a subpolynomial factor). Subtraction can be obtained analogously,
and one can determine whether $a \geq b$ can be done by subtracting
$a$ from $b$ and looking at whether the result is negative.

Although the construction based on Barrington's theorem has polynomial
overhead and is thus sufficient for our purposes, it seems worth
noting that it is possible to achieve better efficiency. As shown by
Draper \cite{Draper}, there exist ancilla-free  
quantum circuits for performing addition and subtraction, which
succeed with high probability and have nearly linear
size. Specifically, one can add or subtract a hardcoded number $a$
into an $n$-qubit register $\ket{x}$ modulo $2^n$ by performing
quantum Fourier transform, followed by $O(n^2)$ controlled-rotations,
followed by an inverse quantum Fourier transform. Furthermore, using
approximate quantum Fourier transforms\cite{Coppersmith, Barenco},
\cite{Draper} describes an approximate version of the circuit, which,
for any value of parameter $m$, uses a total of only $O(m n \log n)$
gates\footnote{A linear-size quantum circuit for exact ancilla-free
  addition is known, but it does not generalize easily to the case of
  hardcoded summands \cite{Cuccaro}.} to produce an output having an
inner product with $\ket{x+a \mod 2^n}$ of $1-O(2^{-m})$.

Because they operate modulo $2^n$, Draper's quantum circuits for
addition and subtraction do not immediately yield fast ancilla-free quantum
circuits for comparison, unlike the classical case. Instead, start with
an $n$-bit number $x$ and then introduce a single clean ancilla qubit
initialized to $\ket{0}$. Then subtract an $n$-bit hardcoded number
$a$ from this register modulo $2^{n+1}$. If $a > x$ then the result
will wrap around into the range $[2^n,2^{n+1}-1]$, in which case the
leading bit will be 1. If $a \leq x$ then the result will be in the
range $[0,2^n-1]$. After copying the result of this leading qubit and
uncomputing the subtraction, the comparison is
complete. Alternatively, one could use the linear size quantum 
comparison circuit devised by Takahashi and Kunihiro, which uses $n$
uninitialized ancillas but no clean ancillas\cite{Takahashi}.

\begin{figure}
\[
\begin{array}{cccc|cccccc}
1 & 2 & 3 & 5 & 13 & 8 & 5 & 3 & 2 & 1 \\
* & p & p & * & p  & p & * & p & p & p
\end{array}
\quad \leftrightarrow \quad (6,5)
\]
\caption{\label{pairs} Here we make a correspondence between strings
  of $p$ and $*$ symbols and ordered pairs of integers. The string of
  9 symbols is split into substrings of length 4 and 5, and each one
  is used to compute an integer by adding the $(i+1)\th$ Fibonacci
  number if $*$ appears in the $i\th$ place. Note the two strings are
  read in different directions.}
\end{figure}

Unfortunately, most crossings in a given braid will not be acting on
the leftmost strand. However, we can reduce the problem of extracting
a general symbol to the problem of extracting the leftmost
symbol. Rather than using equation \ref{correspondence} to make a
correspondence between a string from $P_n$ and a single integer, we 
can split the string at some chosen point, and use equation
\ref{correspondence} on each piece to make a correspondence between
elements of $P_n$ and ordered pairs of integers, as shown in figure
\ref{pairs}. To extract the $i\th$ symbol, we thus convert encoding
\ref{correspondence} to the encoding where the string is split between
the $i\th$ and $(i-1)\th$ symbols, so that one only needs to extract
the leftmost symbol of the second string. Like equation
\ref{correspondence}, this is also an efficient encoding, in which the
encoded bitstrings form a large fraction of all possible bitstrings.

To convert encoding \ref{correspondence} to a split encoding with the
split at an arbitrary point, we can move the split rightward by one
symbol at a time. To introduce a split between the leftmost and
second-to-leftmost symbols, one must extract the leftmost symbol as
described above. To move the split one symbol to the right, one must
extract the leftmost symbol from the right string, and if it is $*$
then add the corresponding Fibonacci number to the left string. This
is again a procedure of addition, subtraction, and comparison of
integers. Note that the computation of Fibonacci
numbers in NC1 is not necessary, as these can be hardcoded into the
circuits. Moving the split back to the left works analogously. As
crossings of different pairs of strands are being simulated, the split
is moved to the place that it is needed. At the end it is moved all
the way leftward and eliminated, leaving a superposition of bitstrings
in the original encoding, which have the correct coefficients
determined by the Fibonacci representation of the given braid.

Lastly, we must consider the weighting in the trace, as described by
equation \ref{weightdef}. Instead of weight $W_s$, we will use
$W_s/\phi$ so that the possible weights are 1 and $1/\phi$ both of
which are $\leq 1$. We can impose any weight $\leq 1$ by doing a
controlled rotation on an extra qubit. The CNOT trick for simulating
a clean qubit which was described in section \ref{DQC1} can be viewed
as a special case of this. All strings in which that qubit takes the
value $\ket{1}$ have weight zero, as imposed by a $\pi/2$ rotation on
the extra qubit. Because none of the weights are smaller than
$1/\phi$, the weighting will cause only a constant overhead
in the number of measurements needed to get a given precision.

\section{DQC1-hardness of Jones Polynomials}

We will prove DQC1-hardness of the problem of estimating the Jones
polynomial of the trace closure of a braid by a reduction from the
problem of estimating the trace of a quantum circuit. To do this, we
will specify an encoding, that is, a map $\eta:Q_n \to S_m$
from the set $Q_n$ of strings of $p$ and $*$ symbols which start with
$*$ and have no two $*$ symbols in a row, to $S_m$, the set of
bitstrings of length $m$. For a given quantum circuit, we will
construct a braid whose Fibonacci representation implements the
corresponding unitary transformation on the encoded bits. The Jones
polynomial of the trace closure of this braid, which is the trace of
this representation, will equal the trace of the encoded quantum
circuit.

Unlike in section \ref{containment}, we will not use a one to one
encoding between bit strings and strings of $p$ and $*$ symbols. All
we require is that a sum over all strings of $p$ and $*$ symbols
corresponds to a sum over bitstrings in which each bitstring appears
an equal number of times. Equivalently, all bitstrings $b \in S_m$
must have a preimage $\eta^{-1}(b)$ of the same size. This insures an
unbiased trace in which no bitstrings are overweighted. To achieve
this we can divide the symbol string into blocks of three symbols and
use the encoding

\begin{equation}
\label{codon}
\begin{array}{ccc}
\textrm{ppp} & \to & 0 \\
\textrm{p$*$p} & \to & 1 \\
\end{array}
\end{equation}

The strings other than ppp and p$*$p do not correspond to any bit
value. Since both the encoded 1 and the encoded 0 begin and end with
p, they can be preceded and followed by any allowable string. Thus,
changing an encoded 1 to an encoded zero does not change the number of
allowed strings of $p$ and $*$ consistent with that encoded
bitstring. Thus the condition that $|\eta^{-1}(b)|$ be independent of
$b$ is satisfied.

We would also like to know \emph{a priori} where in the string of p
and $*$ symbols a given bit is encoded. This way, when we need to
simulate a gate acting on a given bit, we would know which strands the
corresponding braid should act on. If we were to simply divide our
string of symbols into blocks of three and write down the
corresponding bit string (skipping every block which is not in one of
the two coding states ppp and p$*$p) then this would not be the
case. Thus, to encode $n$ bits, we will instead divide the string of
symbols into $n$ superblocks, each consisting of $c \log n$ blocks of
three for some constant $c$. To decode a superblock, scan it from left
to right until you reach either a ppp block or a p$*$p block. The
first such block encountered determines whether the superblock encodes
a 1 or a 0, according to equation \ref{codon}. Now imagine we choose a
string randomly from $Q_{3cn \log n}$. By choosing the constant
prefactor $c$ in our superblock size we can ensure that in the entire
string of $3cn \log n$ symbols,  the probability of there being any
noncoding superblock which contains neither a ppp block nor a p$*$p
block is polynomially small. If this is the case, then these noncoding
strings will contribute only a polynomially small additive error to
the estimate of the circuit trace, on par with the other sources of
error.

The gate set consisting of the CNOT, Hadamard, and $\pi/8$ gates is
known to be universal for BQP \cite{Nielsen_Chuang}. Thus, it suffices to
consider the simulation of 1-qubit and 2-qubit gates. Furthermore, it
is sufficient to imagine the qubits arranged on a line and to allow
2-qubit gates to act only on neighboring qubits. This is because
qubits can always be brought into neighboring positions by applying a
series of SWAP gates to nearest neighbors. By our encoding a
unitary gate applied to qubits $i$ and $i+1$ will correspond to a
unitary transformation on symbols $i3c \log n$ through $(i+2)3c \log
n-1$. The essence of our reduction is to take each quantum gate
and represent it by a corresponding braid on logarithmically many 
symbols whose Fibonacci representation performs that gate on the
encoded qubits. 

Let's first consider the problem of simulating a gate on the first
pair of qubits, which are encoded in the leftmost two superblocks of
the symbol string. We'll subsequently consider the more difficult case
of operating on an arbitrary pair of neighboring encoded qubits. As
mentioned in section \ref{Fibonacci}, the Fibonacci representation 
$\rho_F^{(n)}$ is reducible. Let $\rho_{**}^{(n)}$ denote the
representation of the braid group $B_n$ defined by the action of
$\rho_F^{(n)}$ on the vector space spanned by strings which begin
and end with $*$. As shown in section \ref{density},
$\rho_{**}^{(n)}(B_n)$ taken modulo phase is a dense subgroup of
$SU(f_{n-1})$, and $\rho_{*p}^{(n)}(B_n)$ modulo phase is a dense
subgroup of $SU(f_n)$. 

In addition to being dense, the $**$ and $*$p blocks of the
Fibonacci representation can be controlled independently. This is a
consequence of the decoupling lemma, as discussed in section
\ref{density}. Thus, given a string of symbols beginning with $*$, and
any desired pair of unitaries on the corresponding $*$p and $**$
vector spaces, a braid can be constructed whose Fibonacci
representation approximates these unitaries to any desired level of
precision. However, the number of crossings necessary may in general
be large. The space spanned by strings of logarithmically many symbols
has only polynomial dimension. Thus, one might guess that the braid
needed to approximate a given pair of unitaries on the $*$p and $**$
vector spaces for logarithmically many symbols will have only
polynomially many crossings. It turns out that this guess is correct,
as we state formally below.
\begin{proposition}
\label{efficiency}
Given any pair of elements $U_{*p} \in SU(f_{k+1})$ and $U_{**} \in
SU(f_k)$, and any real parameter $\epsilon$, one can in polynomial
time find a braid $b \in B_k$ with $\mathrm{poly}(n,\log(1/\epsilon))$
crossings whose Fibonacci representation satisfies
$\| \rho_{*p}(b) - U_{*p} \| \leq \epsilon$ and 
$\| \rho_{**}(b) - U_{**} \| \leq \epsilon$, provided that $k = O(\log
n)$. By symmetry, the same holds when considering $\rho_{p*}$ rather
than $\rho_{*p}$.
\end{proposition}
Note that proposition \ref{efficiency} is a property of the Fibonacci
representation, not a generic consequence of density, since it is in
principle possible for the images of group generators in a dense
representation to lie exponentially close to some subgroup of the
corresponding unitary group. We prove this proposition in section
\ref{app_efficiency}.

With proposition \ref{efficiency} in hand, it is apparent that any
unitary gate on the first two encoded bits can be efficiently
performed. To similarly simulate gates on arbitrary pairs of
neighboring encoded qubits, we will need some way to unitarily bring a
$*$ symbol to a known location within logarithmic distance of the
relevant encoded qubits. This way, we ensure that we are acting in the
$*p$ or $**$ subspaces.

To move $*$ symbols to known locations we'll use an ``inchworm''
structure which brings a pair of $*$ symbols rightward to where they
are needed. Specifically, suppose we have a pair of superblocks which each
have a $*$ in their exact center. The presence of the left $*$ and the
density of $\rho_{*p}$ allow us to use proposition \ref{efficiency} to
unitarily move the right $*$ one superblock to the right by adding
polynomially many crossings to the braid. Then, the presence of the
right $*$ and the density of $\rho_{p*}$ allow us to similarly move
the left $*$ one superblock to the right, thus bringing it into the superblock
adjacent to the one which contains the right $*$. This is illustrated
in figure \ref{inchfig}. To move the inchworm to the left we use the
inverse operation.

\begin{figure}
\begin{center}
\includegraphics[width=0.45\textwidth]{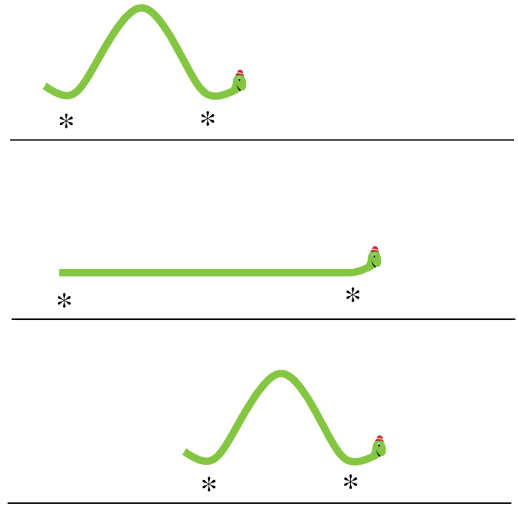}
\caption{\label{inchfig} This sequence of unitary steps is used to
  bring a $*$ symbol where it is needed in the symbol string to ensure
  density of the braid group representation. The presence of the left
  $*$ ensures density to allow the movement of the right $*$ by
  proposition \ref{efficiency}. Similarly, the presence of the right $*$
  allows the left $*$ to be moved.} 
\end{center}
\end{figure}

To simulate a given gate, one first uses the previously described
procedure to make the inchworm crawl to the superblocks just to the
left of the superblocks which encode the qubits on which the gate
acts. Then, by the density of $\rho_{*p}$ and proposition
\ref{efficiency}, the desired gate can be simulated using polynomially
many braid crossings.

To get this process started, the leftmost two superblocks must each contain
a $*$ at their center. This occurs with constant probability. The
strings in which this is not the case can be prevented from
contributing to the trace by a technique analogous to that used in
section \ref{DQC1} to simulate logarithmically many clean
ancillas. Namely, an extra encoded qubit can be conditionally flipped
if the first two superblocks do not both have $*$ symbols at their
center. This can always be done using proposition \ref{efficiency},
since the leftmost symbol in the string is always $*$, and the
$\rho_{*p}$ and $\rho_{**}$ representations are both dense.

It remains to specify the exact unitary operations which move the
inchworm. Suppose we have a current superblock and a target superblock. The
current superblock contains a $*$ in its center, and the target superblock is
the next superblock to the right or left. We wish to move the $*$ to the
center of the target superblock. To do this, we can select the smallest
segment around the center such that in each of these superblocks, the
segment is bordered on its left and right by p symbols. This segment
can then be swapped, as shown in figure \ref{swap}.

\begin{figure}
\begin{center}
\includegraphics[width=0.55\textwidth]{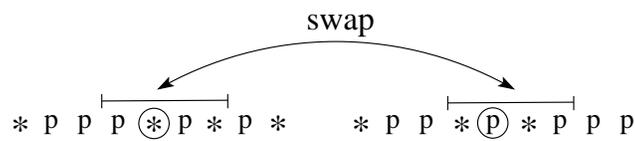}
\caption{\label{swap} This unitary procedure starts with a $*$ in the
  current superblock and brings it to the center of the target superblock.} 
\end{center}
\end{figure}

For some possible strings this procedure will not be well
defined. Specifically there may not be any segment which contains the
center and which is bordered by p symbols in both superblocks. On such
strings we define the operation to act as the identity. For random
strings, the probability of this decreases exponentially with the
superblock size. Thus, by choosing $c$ sufficiently large we can make this
negligible for the entire computation.

As the inchworm moves rightward, it leaves behind a trail. Due to the
swapping, the superblocks are not in their original state after the
inchworm has passed. However, because the operations are unitary, when
the inchworm moves back to the left, the modifications to the superblocks
get undone. Thus the inchworm can shuttle back and forth, moving where
it is needed to simulate each gate, always stopping just to the left
of the superblocks corresponding to the encoded qubits.

The only remaining detail to consider is that the trace appearing
in the Jones polynomial is weighted depending on whether the last
symbol is $p$ or $*$, whereas the DQC1-complete trace estimation
problem is for completely unweighted traces. This problem is easily
solved. Just introduce a single extra superblock at the end of the
string. After bringing the inchworm adjacent to the last superblock,
apply a unitary which performs a conditional rotation on the qubit
encoded by this superblock. The rotation will be by an angle so that
the inner product of the rotated qubit with its original state is
$1/\phi$ where $\phi$ is the golden ratio. This will be done only if
the last symbol is $p$. This exactly cancels out the weighting which
appears in the formula for the Jones polynomial, as described in
section \ref{proof}.

Thus, for appropriate $\epsilon$, approximating the Jones polynomial
of the trace closure of a braid to within $\pm \epsilon$ is
DQC1-hard.

\section{Conclusion}
\label{conclusion}

The preceding sections show that the problem of approximating the
Jones polynomial of the trace closure of a braid with $n$ strands and
$m$ crossings to within $\pm \epsilon$ at $t=e^{i 2 \pi/5}$ is a
DQC1-complete problem for appropriate $\epsilon$. The proofs are based
on the problem of evaluating the Markov trace of the Fibonacci
representation of a braid to $\frac{1}{\mathrm{poly}(n,m)}$
precision. By equation \ref{jones}, we see that this corresponds to
evaluating the Jones polynomial with $\pm
\frac{|D^{n-1}|}{\mathrm{poly}(n,m)}$ precision, where
$D=-A^2-A^{-2}=2\cos (6 \pi/5)$. Whereas approximating the Jones
polynomial of the plat closure of a braid was known\cite{Aharonov2} to
be BQP-complete, it was previously only known that the problem of
approximating the Jones polynomial of the trace closure of a braid was
in BQP. Understanding the complexity of approximating the Jones
polynomial of the trace closure of a braid to precision  
$\pm \frac{|D^{n-1}|}{\mathrm{poly}(n,m)}$ was posed as an open
problem in \cite{Aharonov1}. This paper shows that for $A=e^{-i 3
  \pi/5}$, this problem is DQC1-complete. Such a completeness result
improves our understanding of both the difficulty of the Jones
polynomial problem and the power one clean qubit computers by finding
an equivalence between the two.

It is generally believed that DQC1 is not contained in P and does not
contain all of BQP. The DQC1-completeness result shows that if this
belief is true, it implies that approximating the Jones polynomial of
the trace closure of a braid is not so easy that it can be done
classically in polynomial time, but is not so difficult as to be
BQP-hard.

To our knowledge, the problem of approximating the Jones polynomial of
the trace closure of a braid is one of only four known candidates for
classically intractable problems solvable on a one clean qubit
computer. The others are estimating the Pauli decomposition of the
unitary matrix corresponding to a polynomial-size quantum
circuit\footnote{This includes estimating the trace of the unitary as
  a special case.}, \cite{Knill_DQC1, Shepherd}, estimating quadratically
signed weight enumerators\cite{Knill_QWGT}, and estimating average fidelity
decay of quantum maps\cite{decay1, decay2}.

\section{Jones Polynomials by Fibonacci Representation}
\label{proof}
For any braid $b \in B_n$ we will define $\ttr(b)$ by:
\begin{equation}
\label{tracedef}
\ttr(b) = \frac{1}{\phi f_n + f_{n-1}} \sum_{s \in Q_{n+1}}
W_s \mx
\end{equation}
We will use $\mid$ to denote a strand and $\nmid$ to denote multiple
strands of a braid (in this case $n$). $Q_{n+1}$ is the set of all
strings of $n+1$ $p$ and $*$ symbols which start with $*$ and contain
no two $*$ symbols in a row. The symbol
\[
\mx
\]
denotes the $s,s$ matrix element of the Fibonacci representation of
braid $b$. The weight $W_s$ is
\begin{equation}
\label{weightdef}
W_s = \left\{ \begin{array}{ll}
\phi & \textrm{if $s$ ends with $p$} \\
1    & \textrm{if $s$ ends with $*$}.
\end{array} \right.
\end{equation}
$\phi$ is the golden ratio $(1+\sqrt{5})/\sqrt{2}$. 

As discussed in \cite{Aharonov1}, the Jones polynomial of the trace
closure of a braid $b$ is given by
\begin{equation}
\label{jones}
V_{b^{\mathrm{tr}}}(A^{-4}) = (-A)^{3w(b^{\mathrm{tr}})} D^{n-1} 
\tr(\rho_A(b^{\mathrm{tr}})).
\end{equation}
$b^{\mathrm{tr}}$ is the link obtained by taking the trace
closure of braid $b$. $w(b^{\mathrm{tr}})$ is denotes the writhe of
the link $b^{\mathrm{tr}}$. For an oriented link, one assigns a
value of $+1$ to each crossing of the form
\includegraphics[width=0.2in]{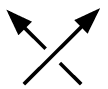}, and the value $-1$
to each crossing of the form
\includegraphics[width=0.2in]{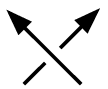}. The writhe of a link is
defined to be the sum of these values over all crossings. $D$ is
defined by $D = -A^2-A^{-2}$. $\rho_A: B_n \to \mathrm{TL}_n(D)$ is a
representation from the braid group to the Temperley-Lieb algebra with
parameter $D$. Specifically, 
\begin{equation}
\label{rhoa}
\rho_A(\sigma_i) = A E_i + A^{-1} \id
\end{equation}
where $E_1 \ldots E_n$ are the generators of $\mathrm{TL}_n(D)$, which
satisfy the following relations.
\begin{eqnarray}
E_i E_j & = & E_j E_i \quad \textrm{ for } |i-j| > 1 \label{rel1} \\
E_i E_{i \pm 1} E_i & = & E_i \label{rel2} \\
E_i^2 & = & D E_i \label{rel3}
\end{eqnarray}
The Markov trace on $\mathrm{TL}_n(D)$ is a linear map $\tr:
\mathrm{TL}_n(D) \to \mathbb{C}$ which satisfies 
\begin{eqnarray}
\tr(\id) & = & 1 \label{mark1} \\
\tr(XY) & = & \tr(YX) \label{mark2} \\
\tr(X E_{n-1}) & = & \frac{1}{D} \tr(X') \label{mark3}
\end{eqnarray}
On the left hand side of equation \ref{mark3}, the trace is on
$\mathrm{TL}_n(D)$, and $X$ is an element of $\mathrm{TL}_n(D)$ not
containing $E_{n-1}$. On the right hand side of equation \ref{mark3},
the trace is on $\mathrm{TL}_{n-1}(D)$, and $X'$ is the element of
$\mathrm{TL}_{n-1}(D)$ which corresponds to $X$ in the obvious way
since $X$ does not contain $E_{n-1}$.

We'll show that the Fibonacci representation satisfies the properties
implied by equations \ref{rhoa}, \ref{rel1}, \ref{rel2}, and
\ref{rel3}. We'll also show that $\ttr$ on the Fibonacci
representation satisfies the properties corresponding to \ref{mark1},
\ref{mark2}, and \ref{mark3}. It was shown in \cite{Aharonov1} that
properties \ref{mark1}, \ref{mark2}, and \ref{mark3}, along with
linearity, uniquely determine the map $\tr$. It will thus follow that
$\ttr(\rho_F^{(n)}(b)) = \tr(\rho_A(b))$, which proves that
the Jones polynomial is obtained from the trace $\ttr$ of
the Fibonacci representation after multiplying by the appropriate
powers of $D$ and $(-A)$ as shown in equation \ref{jones}. Since these
powers are trivial to compute, the problem of approximating the Jones
polynomial at $A = e^{-i 3 \pi/5}$ reduces to the problem of computing
this trace.

\newcommand{\symbox}[1]
{\begin{array}{l} \includegraphics[width=0.5in]{#1.eps} \end{array}}

$\ttr$ is equal to the ordinary matrix trace on the subspace of
strings ending in $*$ plus $\phi$ times the matrix trace on the
subspace of strings ending in p. Thus the fact that the matrix trace
satisfies property \ref{mark2} immediately implies that $\ttr$ does
too. Furthermore, since the dimensions of these subspaces are
$f_{n-1}$ and $f_n$ respectively, we see from equation \ref{tracedef}
that $\ttr(\id) = 1$. To address property \ref{mark3}, we'll
first show that
\begin{equation}
\label{toshow}
\ttr \left( \symbox{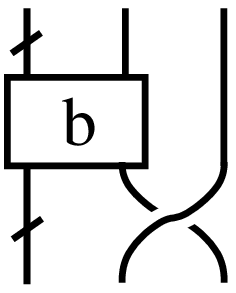} \right) =
  \frac{1}{\delta} \ttr \left( \begin{array}{l}
  \includegraphics[width=0.3in]{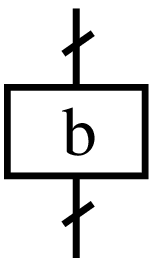} \end{array} \right)
\end{equation}
for some constant $\delta$ which we will calculate. We will then use
equation \ref{rhoa} to relate $\delta$ to $D$.

Using the definition of $\ttr$ we obtain
\begin{eqnarray*}
\ttr \left( \symbox{braidbox3} \right) & = &\frac{1}{f_n \phi + f_{n-1}}
\left[ \phi \sum_{s \in Q_{n-2}} \symbox{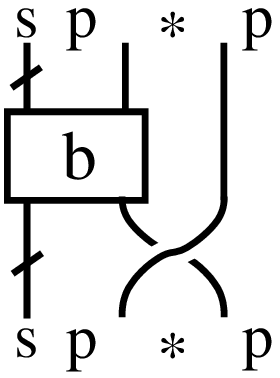} + \phi \sum_{s \in
    Q_{n-2}} \symbox{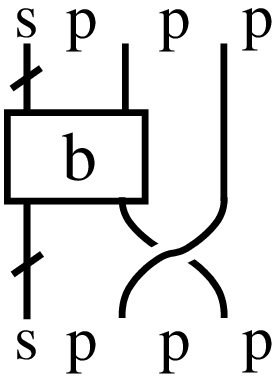} \right.\\
 & & \left. + \sum_{s \in Q_{n-2}} \symbox{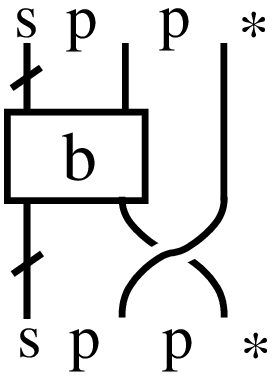} + \sum_{s \in
    Q'_{n-2}} \symbox{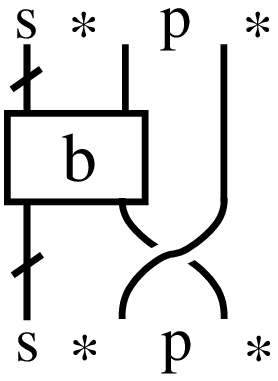} + \phi \sum_{s \in Q'_{n-2}} \symbox{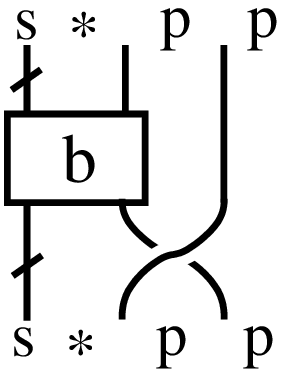} \right]
\end{eqnarray*}
where $Q'_{n-2}$ is the set of length $n-2$ strings of $*$ and $p$
symbols which begin with $*$, end with $p$, and have no two $*$
symbols in a row.

Next we expand according to the braiding rules described in equations
\ref{picto} and \ref{rules}.
\begin{eqnarray*}
& = &\frac{1}{f_n \phi + f_{n-1}}
\left[ \sum_{s \in Q_{n-2}} \left( \phi c \symbox{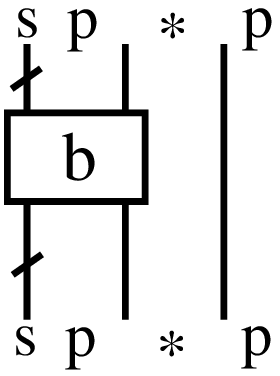} + \phi
  e \symbox{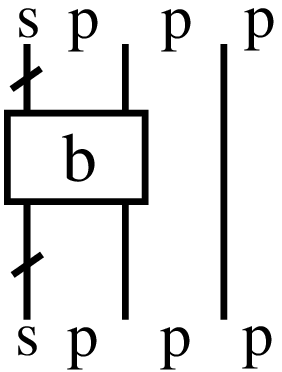} + a \symbox{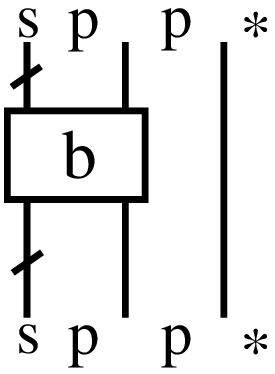} \right) \right.\\
 & & \left. + \sum_{s \in Q'_{n-2}} \left( b \symbox{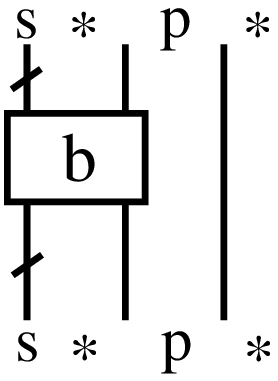} +
  \phi a \symbox{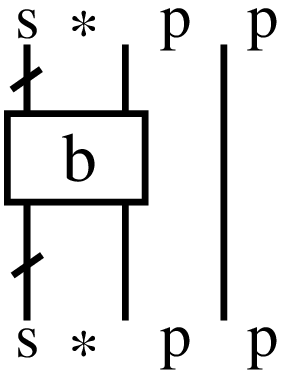} \right) \right]
\end{eqnarray*}
We know that matrix elements in which differing string symbols are
separated by unbraided strands will be zero. To obtain the preceding
expression we have omitted such terms.  Simplifying yields
\newcommand{\smallbox}[1]
{\begin{array}{l} \includegraphics[width=0.35in]{#1.eps} \end{array}}
\[
= \frac{1}{f_n + \phi f_{n-1}} \left[ \sum_{s \in Q_{n-2}} \left( \phi
  c \smallbox{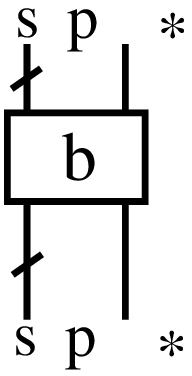} + (\phi e + a) \smallbox{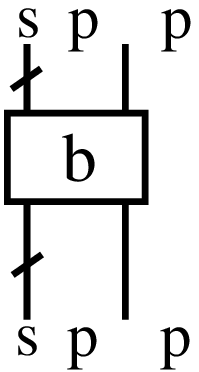} \right) + \sum_{s \in
  Q'_{n-2}} (b+ \phi a) \smallbox{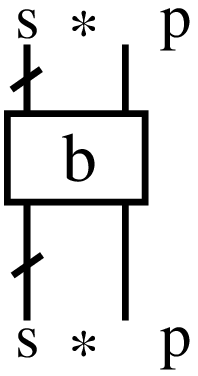} \right].
\]
By the definitions of $A$, $a$, $b$, and $e$, given in equation
\ref{constants}, we see that $\phi e + a = b + \phi a$. Thus the above
expression simplifies to
\newcommand{\minibox}[1]
{\begin{array}{l} \includegraphics[width=0.3in]{#1.eps} \end{array}}
\[
= \frac{1}{f_n \phi + f_{n-1}} \left[ \sum_{s \in Q'_{n-1}} \phi c
  \minibox{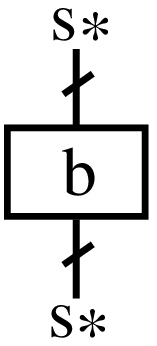} + \sum_{s \in Q_{n-1}} (\phi e + a) \minibox{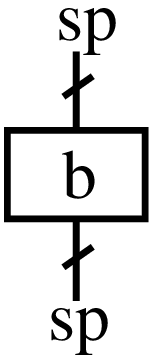} \right]
\]
Now we just need to show that
\begin{equation}
\label{delta1}
\frac{\phi c}{f_n \phi + f_{n-1}} = 
\frac{1}{\delta}\frac{1}{f_{n-1} \phi + f_{n-2}}
\end{equation}
and
\begin{equation}
\label{delta2}
\frac{\phi e + a}{f_n \phi + f_{n-1}} =
\frac{1}{\delta} \frac{\phi}{f_{n-1} \phi + f_{n-2}}.
\end{equation}
The Fibonacci numbers have the property
\[
\frac{f_n \phi + f_{n-1}}{f_{n-1} \phi + f_{n-2}} = \phi
\]
for all $n$. Thus equations \ref{delta1} and \ref{delta2} are
equivalent to
\begin{equation}
\label{equiv1}
\phi c = \frac{1}{\delta} \phi
\end{equation}
and
\begin{equation}
\label{equiv2}
\phi e + a = \frac{1}{\delta} \phi^2
\end{equation}
respectively. For $A=e^{-i 3 \pi/5}$ these both yield $\delta =
A-1$. Hence
\[
\ttr \left( \symbox{braidbox3} \right) = \frac{1}{\delta}
\frac{1}{f_{n-1} \phi + f_{n-2}} \left[ \sum_{s \in Q'_{n-1}}
  \minibox{s} + \sum_{s \in Q_{n-1}} \phi \minibox{p} \right]
\]
\[
= \frac{1}{\delta} \ttr(b),
\]
thus confirming equation \ref{toshow}. 

Now we calculate $D$ from $\delta$. Solving \ref{rhoa} for $E_i$
yields
\begin{equation}
\label{ei}
E_i = A^{-1} \rho_A(\sigma_i) - A^{-2} \id
\end{equation}
Substituting this into \ref{mark3} yields
\[
\tr(X (A^{-1} \rho_A(\sigma_i)-A^{-2} \id)) = \frac{1}{D} \tr(X)
\]
\[
\Rightarrow A^{-1} \tr(X \rho_A(\sigma_i)) - A^{-2} \tr(X) =
\frac{1}{D} \tr(X).
\]
Comparison to our relation $\tr(X \rho_A(\sigma_i)) = \frac{1}{\delta}
\tr(X)$ yields
\[
A^{-1} \frac{1}{\delta} - A^{-2} = \frac{1}{D}.
\]
Solving for $D$ and substituting in $A=e^{-i 3 \pi/5}$ yields
\[
D = \phi.
\]
This is also equal to $-A^2-A^{-2}$ consistent with the usage
elsewhere.

Thus we have shown that $\ttr$ has all the necessary properties. We
will next show that the image of the representation $\rho_F$ of the
braid group $B_n$ also forms a representation of the Temperley-Lieb
algebra $TL_n(D)$. Specifically, $E_i$ is represented by
\begin{equation}
\label{rep}
E_i \to A^{-1} \rho_F^{(n)}(\sigma_i) - A^{-2} \id.
\end{equation}
To show that this is a representation of $TL_n(D)$ we must show that
the matrices described in equation \ref{rep} satisfy the properties
\ref{rel1}, \ref{rel2}, and \ref{rel3}. By the theorem of
\cite{Aharonov1} which shows that a Markov trace on any representation
of the Temperley-Lieb algebra yields the Jones polynomial, it will
follow that the trace of the Fibonacci representation yields the Jones
polynomial.

Since $\rho_F$ is a representation of the braid group and $\sigma_i
\sigma_j = \sigma_j \sigma_i$ for $|i-j| > 1$, it immediately follows
that the matrices described in equation \ref{rep} satisfy condition
\ref{rel1}. Next, we'll consider condition \ref{rel3}. By inspection
of the Fibonacci representation as given by equation 
\ref{rules}, we see that by appropriately ordering the
basis\footnote{We will have to choose different orderings for
  different $\sigma_i$'s.} we can bring $\rho_A(\sigma_i)$ into block
diagonal form, where each block is one of the following $1 \times 1$
or $2 \times 2$ possibilities.
\[
\left[ a \right] \quad \left[ b \right] \quad 
\left[ \begin{array}{cc}
c & d \\
d & e
\end{array} \right]
\]
Thus, by equation \ref{rep}, it suffices to show that
\[
\left( A^{-1} \left[ \begin{array}{cc}
c & d \\
d & e
\end{array} \right]
- A^{-2} \left[ \begin{array}{cc}
1 & 0 \\
0 & 1
\end{array} \right] \right)^2 = D \left( A^{-1}
\left[ \begin{array}{cc}
c & d \\
d & e
\end{array} \right]
- A^{-2} \left[ \begin{array}{cc}
1 & 0 \\
0 & 1
\end{array} \right] \right),
\]
\[
(A^{-1} a - A^{-2})^2 = D (A^{-1} a - A^{-2} ),
\]
and
\[
(A^{-1} b - A^{-2})^2 = D (A^{-1} b - A^{-2}),
\]
\emph{i.e.} each of the blocks square to $D$ times themselves. These
properties are confirmed by direct calculation.

Now all that remains is to check that the correspondence \ref{rep}
satisfies property \ref{rel2}. Using the rules described in equation
\ref{rules} we have
\[
\rho_F^{(3)}(\sigma_1) = \left[ \begin{array}{cccccccc}
b & 0 &   &   &   &   &   &   \\
0 & a &   &   &   &   &   &   \\
  &   & a &   &   &   &   &   \\
  &   &   & e & 0 & d &   &   \\
  &   &   & 0 & a & 0 &   &   \\
  &   &   & d & 0 & c &   &   \\
  &   &   &   &   &   & e & d \\
  &   &   &   &   &   & d & c
\end{array} \right] 
\begin{array}{c}
\textrm{$*$p$*$p} \\
\textrm{$*$ppp} \\
\textrm{$*$pp$*$} \\
\textrm{pppp} \\
\textrm{pp$*$p} \\
\textrm{p$*$pp} \\
\textrm{ppp$*$} \\
\textrm{p$*$p$*$} \\
\end{array}
\quad \quad
\rho_F^{(3)}(\sigma_2) = \left[ \begin{array}{cccccccc}
c & d &   &   &   &   &   &   \\
d & e &   &   &   &   &   &   \\
  &   & a &   &   &   &   &   \\
  &   &   & e & d & 0 &   &   \\
  &   &   & d & c & 0 &   &   \\
  &   &   & 0 & 0 & a &   &   \\
  &   &   &   &   &   & a & 0 \\
  &   &   &   &   &   & 0 & b
\end{array} \right]
\begin{array}{c}
\textrm{$*$p$*$p} \\
\textrm{$*$ppp} \\
\textrm{$*$pp$*$} \\
\textrm{pppp} \\
\textrm{pp$*$p} \\
\textrm{p$*$pp} \\
\textrm{ppp$*$} \\
\textrm{p$*$p$*$} \\
\end{array}
\]
(Here we have considered all four subspaces unlike in equation
\ref{examp}.) Substituting this into equation \ref{rep} yields
matrices which satisfy condition \ref{rel3}. It follows that equation
\ref{rep} yields a representation of the Temperley-Lieb algebra. This
completes the proof that
\[
V_{b^{\mathrm{tr}}}(A^{-4}) = (-A)^{3w(b^{\mathrm{tr}})} D^{n-1} 
\ttr(\rho_F^{(n)}(b^{\mathrm{tr}}))
\]
for $A=e^{-i 3 \pi/5}$.

\section{Density of the Fibonacci representation}
\label{density}

In this section we will show that $\rho_{**}^{(n)}(B_n)$ is a dense
subgroup of $SU(f_{n-1})$ modulo phase, and that
$\rho_{*p}^{(n)}(B_n)$ and $\rho_{p*}^{(n)}(B_n)$ are dense subgroups
of $SU(f_n)$ modulo phase. Similar results regarding the path model
representation of the braid group were proven in \cite{Aharonov2}. Our
proofs will use many of the techniques introduced there.

We'll first show that $\rho_{**}^{(4)}(B_4)$ modulo phase is a dense
subgroup of $SU(2)$. We can then use the bridge lemma from
\cite{Aharonov2} to extend the result to arbitrary $n$.

\begin{proposition} \label{SU2} $\rho_{**}^{(4)}(B_4)$ modulo phase is
  a dense subgroup of $SU(2)$.
\end{proposition}

\begin{proof}
Using equation \ref{rules} we have:
\[
\rho_{**}^{(4)}(\sigma_1) = \rho_{**}^{(4)}(\sigma_3) =
\left[ \begin{array}{cc}
b & 0 \\
0 & a
\end{array} \right] \begin{array}{c}
\textrm{$*$p$*$p$*$} \\
\textrm{$*$ppp$*$}
\end{array}
\quad \quad
\rho_{**}^{(4)}(\sigma_2) =
\left[ \begin{array}{cc}
c & d \\
d & e
\end{array} \right] \begin{array}{c}
\textrm{$*$p$*$p$*$} \\
\textrm{$*$ppp$*$}
\end{array}
\]
We do not care about global phase so we will take
\[
\rho_{**}^{(n)}(\sigma_i) \quad \to \quad \frac{1}{ \left( \det
  \rho_{**}^{(n)}(\sigma_i) \right)^{1/f_{n-1}}}
\ \rho_{**}^{(n)}(\sigma_i) 
\]
to project into $SU(f_{n-1})$. Thus we must show the group $\langle
A,B \rangle$ generated by
\begin{equation}
\label{AB}
A = \frac{1}{\sqrt{ab}} \left[ \begin{array}{cc}
b & 0 \\
0 & a
\end{array} \right] \quad \quad
B = \frac{1}{\sqrt{ce-d^2}} \left[ \begin{array}{cc}
c & d \\
d & e
\end{array} \right]
\end{equation}
is a dense subgroup of $SU(2)$. To do this we will use the well known
surjective homomorphism $\phi:SU(2) \to SO(3)$ whose kernel is $\{ \pm
  \id \}$ (\emph{cf.} \cite{Artin}, pg. 276). A general element of
  $SU(2)$ can be written as
\[
\cos\left( \frac{\theta}{2} \right) \id + i \sin \left( \frac{\theta}{2}
\right) \left[ x \sigma_x + y \sigma_y + z \sigma_z \right]
\]
where $\sigma_x$, $\sigma_y$, $\sigma_z$ are the Pauli matrices,
and $x$, $y$, $z$ are real numbers satisfying $x^2+y^2+z^2=1$. $\phi$
maps this element to the rotation by angle $\theta$ about the axis 
\[
\vec{x} = \left[ \begin{array}{c} 
x \\ 
y \\ 
z \end{array} \right].
\]

Using equations \ref{AB} and \ref{constants}, one finds that $\phi(A)$
and $\phi(B)$ are both rotations by $7 \pi /5$. These rotations are
about different axes which are separated by angle
\[
\theta_{12} = \cos^{-1} (2-\sqrt{5})
\simeq
1.8091137886\ldots
\]
To show that $\rho_{**}^{(4)}(B_4)$ modulo phase is a dense subgroup
of $SU(2)$ it suffices to show that $\phi(A)$ and $\phi(B)$ generate
a dense subgroup of $SO(3)$. To do this we take advantage of the fact
that the finite subgroups of $SO(3)$ are completely known.

\begin{theorem}$\mathrm{(\cite{Artin} \ pg. \ 184)}$
Every finite subgroup of $SO(3)$ is one of the following:
\begin{itemize}
\item[] $C_k$: the cyclic group of order $k$
\item[] $D_k$: the dihedral group of order $k$
\item[] $T$: the tetrahedral group (order 12)
\item[] $O$: the octahedral group (order 24)
\item[] $I$: the icosahedral group (order 60)
\end{itemize}
\end{theorem}
The infinite proper subgroups of $SO(3)$ are all isomorphic to $O(2)$
or $SO(2)$. Thus, since $\phi(A)$ and $\phi(B)$ are rotations about
different axes, $\langle \phi(A), \phi(B) \rangle$ can only be
$SO(3)$ or a finite subgroup of $SO(3)$. If we can show that $\langle
\phi(A), \phi(B) \rangle$ is not contained in any of the finite
subgroups of $SO(3)$ then we are done.

Since $\phi(A)$ and $\phi(B)$ are rotations about different axes we
know that $\langle \phi(A), \phi(B) \rangle$ is not $C_k$ or
$D_k$. Next, we note that $R=\phi(A)^5 \phi(B)^5$ is a rotation by $2
\theta_{12}$. By direct calculation, $2 \theta_{12}$ is not an integer
multiple of $2 \pi / k$ for $k = 1,2,3,4,$ or 5. Thus $R$ has order
greater than 5. As mentioned on pg. 262 of \cite{Jones2}, $T$, $O$,
and $I$ do not have any elements of order greater than 5. Thus,
$\langle \phi(A), \phi(B) \rangle$ is not contained in $C$, $O$, or
$I$, which completes the proof. Alternatively, using more arithmetic
and less group theory, we can see that $2 \theta_{12}$ is not any
integer multiple of $2 \pi / k$ for any $k \leq 30$, thus $R$ cannot
be in $T$, $O$, or $I$ since its order does not divide the order of
any of these groups.
\end{proof}

Next we'll consider $\rho_{**}^{(n)}$ for larger $n$. These will be
matrices acting on the strings of length $n+1$. These can be divided
into those which end in pp$*$ and those which end in $*$p$*$. The
space upon which $\rho_{**}^{(n)}$ acts can correspondingly be divided
into two subspaces which are the span of these two sets of
strings. From equation \ref{rules} we can see that
$\rho_{**}^{(n)}(\sigma_1)\ldots\rho_{**}^{(n)}(\sigma_{n-3})$ will
leave these subspaces invariant. Thus if we order our basis to respect
this grouping of strings,
$\rho_{**}^{(n)}(\sigma_1)\ldots\rho_{**}^{(n)}(\sigma_{n-3})$ will
appear block-diagonal with a block corresponding to each of these
subspaces.

The possible prefixes of $*$p$*$ are all strings of length $n-2$ that
start with $*$ and end with p. Now consider the strings acted upon by
$\rho_{**}^{(n-2)}$. These have length $n-1$ and must end in $*$. The 
possible prefixes of this $*$ are all strings of length $n-2$ that
begin with $*$ and end with p. Thus these are in one to one
correspondence with the strings acted upon by $\rho_{**}^{(n)}$ that
end in $*$p$*$. Furthermore, since the rules \ref{rules} depend only
on the three symbols neighboring a given crossing, the block of 
$\rho_{**}^{(n)}(\sigma_1)\ldots\rho_{**}^{(n)}(\sigma_{n-3})$
corresponding to the $*$p$*$ subspace is exactly the same as 
$\rho_{**}^{(n-2)}(\sigma_1)\ldots\rho_{**}^{(n-2)}(\sigma_{n-3})$. By
a similar argument, the block of 
$\rho_{**}^{(n)}(\sigma_1)\ldots\rho_{**}^{(n)}(\sigma_{n-3})$
corresponding to the pp$*$ is exactly the same as 
$\rho_{**}^{(n-1)}(\sigma_1)\ldots\rho_{**}^{(n-1)}(\sigma_{n-3})$.

For any $n> 3$, $\rho_{**}^{(n)}(\sigma_{n-2})$ will not leave these
subspaces invariant. This is because the crossing $\sigma_{n-2}$ spans
the $(n-1)\th$ symbol. Thus if the $(n-2)\th$ and $n\th$ symbols are
$p$, then by equation \ref{rules}, $\rho_{**}^{(n)}$ can flip the
value of the $(n-1)\th$ symbol. The $n\th$ symbol is guaranteed to be
$p$, since the $(n+1)\th$ symbol is the last one and is therefore $*$
by definition. For any $n > 3$, the space acted upon by
$\rho_{**}^{(n)}(\sigma_{n-1})$ will include some strings in which the 
$(n-2)\th$ symbol is $p$. 

As an example, for five strands:
\[
\rho_{**}^{(5)}(\sigma_1) = \left[ \begin{array}{ccc}
b & 0 & 0 \\
0 & a & 0 \\
0 & 0 & a
\end{array} \right]
\begin{array}{c}
\textrm{$*$p$*$pp$*$} \\
\textrm{$*$pppp$*$} \\
\textrm{$*$pp$*$p$*$}
\end{array} \quad \quad
\rho_{**}^{(5)}(\sigma_2) = \left[ \begin{array}{ccc}
c & d & 0 \\
d & e & 0 \\
0 & 0 & a
\end{array} \right]
\begin{array}{c}
\textrm{$*$p$*$pp$*$} \\
\textrm{$*$pppp$*$} \\
\textrm{$*$pp$*$p$*$}
\end{array}
\]
\[
\rho_{**}^{(5)}(\sigma_3) = \left[ \begin{array}{ccc}
a & 0 & 0 \\
0 & e & d \\
0 & d & c
\end{array} \right]
\begin{array}{c}
\textrm{$*$p$*$pp$*$} \\
\textrm{$*$pppp$*$} \\
\textrm{$*$pp$*$p$*$}
\end{array} \quad \quad
\rho_{**}^{(5)}(\sigma_4) = \left[ \begin{array}{ccc}
a & 0 & 0 \\
0 & a & 0 \\
0 & 0 & b
\end{array} \right]
\begin{array}{c}
\textrm{$*$p$*$pp$*$} \\
\textrm{$*$pppp$*$} \\
\textrm{$*$pp$*$p$*$}
\end{array}
\]
We recognize the upper $2 \times 2$ blocks of
$\rho_{**}^{(5)}(\sigma_1)$, and $\rho_{**}^{(5)}(\sigma_2)$ from
equation \ref{examp}. The lower $1 \times 1$ block matches
$\rho_{**}^{(3)}(\sigma_1)$ and $\rho_{**}^{(3)}(\sigma_2)$, which are
both easily calculated to be $[a]$. $\rho_{**}^{(5)}(\sigma_3)$ mixes
these two subspaces.

We can now use the preceding observations about the recursive
structure of $\{ \rho_{**}^{(n)}| n = 4,5,6,7\ldots \}$ to show
inductively that $\rho_{**}^{(n)}(B_n)$ forms a dense subgroup of
$SU(f_{n-1})$ for all $n$. To perform the induction step we use the
bridge lemma and decoupling lemma from \cite{Aharonov2}.

\begin{lemma}[Bridge Lemma]
Let $C=A \oplus B$ where $A$ and $B$ are vector spaces with
$\dim{B} > \dim{A} \geq 1$. Let $W \in SU(C)$ be a linear
transformation which mixes the subspaces $A$ and $B$. Then the group
generated by $SU(A)$, $SU(B)$, and $W$ is dense in $SU(C)$.
\end{lemma}

\begin{lemma}[Decoupling Lemma]
Let $G$ be an infinite discrete group, and let $A$ and $B$ be two
vector spaces with $\dim(A) \neq \dim(B)$. Let $\rho_a: G \to SU(A)$
and $\rho_b: G \to SU(B)$ be homomorphisms such that $\rho_a(G)$ is
dense in $SU(A)$ and $\rho_b(G)$ is dense in $SU(B)$.  Then for any $U_a
\in SU(A)$ there exist a series of $G$-elements $\alpha_n$ such that
$\lim_{n \to \infty} \rho_a(\alpha_n) = U_a$ and $\lim_{n \to \infty}
\rho_b(\alpha_n) = \id$. Similarly, for any $U_b \in SU(B)$, there exists a
series $\beta_n \in G$ such that $\lim_{n \to \infty} \rho_a(\beta_n) = \id$
and $\lim_{n \to \infty} \rho_a(\beta_n) = U_b$.
\end{lemma}

With these in hand we can prove the main proposition of this
section.
\begin{proposition}
\label{**density}
For any $n \geq 3$, $\rho_{**}^{(n)}(B_n)$ modulo phase is a dense subgroup
of $SU(f_{n-1})$.
\end{proposition}

\begin{proof}
As mentioned previously, the proof will be inductive. The base cases are
$n=3$ and $n=4$. As mentioned previously, $\rho_{**}^{(3)}(\sigma_1) =
\rho_{**}^{(3)}(\sigma_2) = [a]$. Trivially, these generate a dense
subgroup of (indeed, all of) $SU(1) = \{ \id \}$ modulo phase. By
proposition \ref{SU2}, $\rho_{**}^{(4)}(\sigma_1)$, and
$\rho_{**}^{(4)}(\sigma_2)$ generate a dense subgroup of $SU(2)$
modulo phase. Now for induction assume that $\rho_{**}^{(n-1)}(B_{n-1})$ is
a dense subgroup of $SU(f_{n-2})$ and $\rho_{**}^{(n-2)}(B_{n-2})$ is
a dense subgroup of $SU(f_{n-3})$. As noted above, these correspond to
the upper and lower blocks of $\rho_{**}^{(n)}(\sigma_1) \ldots
\rho_{**}^{(n)}(\sigma_{n-2})$. Thus, by the decoupling lemma,
$\rho_{**}^{(n)}(B_n)$ contains an element arbitrarily close to $U
\oplus \id$ for any $U \in SU(f_{n-2})$ and an element arbitrarily
close to $\id \oplus U$ for any $U \in SU(f_{n-3})$. Since, as
observed above, $\rho_{**}^{(n)}(\sigma_{n-1})$ mixes these two
subspaces, the bridge lemma shows that $\rho_{**}^{(n)}(B_n)$ is dense
in $SU(f_{n-1})$.
\end{proof}

From this, the density of $\rho_{*p}^{(n)}$ and $\rho_{p*}^{(n)}$
easily follow.

\begin{corollary}
$\rho_{*p}^{(n)}(B_n)$ and $\rho_{p*}^{(n)}(B_n)$ are dense subgroups
  of $SU(f_n)$ modulo phase.
\end{corollary}

\begin{proof}
It is not hard to see that
\[
\begin{array}{rcl}
\rho_{*p}^{(n)}(\sigma_1) & = & \rho_{**}^{(n+1)}(\sigma_1) \\
 & \vdots & \\
\rho_{*p}^{(n)}(\sigma_{n-1}) & = & \rho_{**}^{(n+1)}(\sigma_{n-1})
\end{array}.
\]
As we saw in the proof of proposition \ref{**density},
$\rho_{**}^{(n+1)}(\sigma_n)$ is not necessary to obtain density in
$SU(f_n)$, that is, $\langle \rho_{**}^{(n+1)}(\sigma_1) , \ldots,
\rho_{**}^{(n+1)}(\sigma_{n-1}) \rangle$ is a dense subgroup of
$SU(f_n)$ modulo phase. Thus, the density of $\rho_{*p}^{(n)}$ in
$SU(f_n)$ follows immediately from the proof of proposition
\ref{**density}. By symmetry, $\rho_{p*}^{(n)}(B_n)$ is isomorphic to
$\rho_{*p}^{(n)}(B_n)$, thus this is a dense subgroup of $SU(f_n)$
modulo phase as well.
\end{proof}

\section{Fibonacci and Path Model Representations}
\label{rep_relations}

For any braid group $B_n$, and any root of unity $e^{i 2 \pi/k}$, the
path model representation is a homomorphism from $B_n$ to to a set of
linear operators. The vector space that these linear operators act is
the space of formal linear combinations of $n$ step paths on the rungs
of a ladder of height $k-1$ that start on the bottom rung. As an
example, all the paths for $n=4$, $k=5$ are shown in below.
\[
\includegraphics[width=5in]{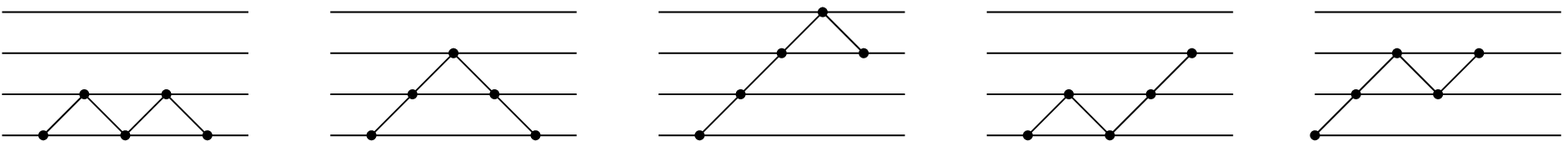}
\]
Thus, the $n=4, k=5$ path model representation is on a five
dimensional vector space. For $k=5$ we can make a bijective
correspondence between the allowed paths of $n$ steps and the set of
strings of $p$ and $*$ symbols of length $n+1$ which start with $*$
and have no to $*$ symbols in a row. To do this, simply label the
rungs from top to bottom as $*$, $p$, $p$, $*$, and directly read off
the symbol string for each path as shown below.
\[
\includegraphics[width=5in]{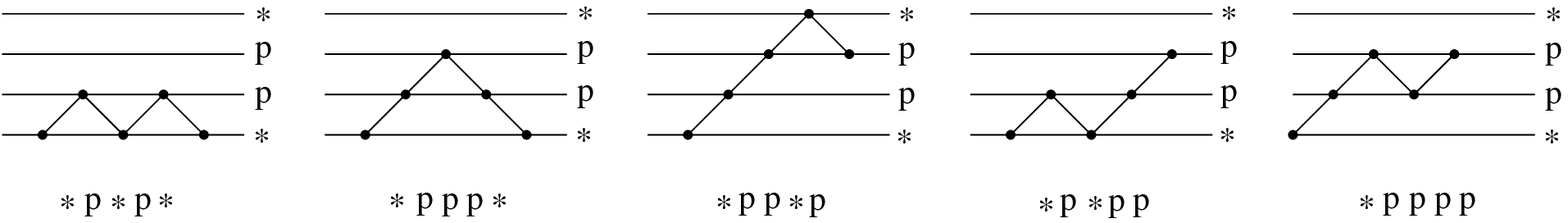}
\]
In \cite{Aharonov1}, it is explained in detail for any given braid how
to calculate the corresponding linear transformation on paths. Using
the correspondence described above, one finds that the path model
representation for $k=5$ is equal to the $-1$ times Fibonacci
representation described in this paper. This sign difference is a
minor detail which arises only because \cite{Aharonov1} chooses a
different fourth root of $t$ for $A$ than we do. This sign difference
is automatically compensated for in the factor of $(-A)^{3
  \cdot \mathrm{writhe}}$, so that both methods yield the correct Jones
polynomial.

\section{Unitaries on Logarithmically Many Strands}
\label{app_efficiency}

In this section we'll prove the following proposition.
\begin{prop1}
Given any pair of elements $U_{*p} \in SU(f_{k+1})$ and $U_{**} \in
SU(f_k)$, and any real parameter $\epsilon$, one can in polynomial
time find a braid $b \in B_k$ with $\mathrm{poly}(n,\log(1/\epsilon))$
crossings whose Fibonacci representation satisfies
$\| \rho_{*p}(b) - U_{*p} \| \leq \epsilon$ and 
$\| \rho_{**}(b) - U_{**} \| \leq \epsilon$, provided that $k = O(\log
n)$. By symmetry, the same holds when considering $\rho_{p*}$ rather
than $\rho_{*p}$.
\end{prop1}
To do so, we'll use a recursive construction. Suppose that we already
know how to achieve proposition \ref{efficiency} on $n$ symbols, and
we wish to extend this to $n+1$ symbols. Using the construction for
$n$ symbols we can efficiently obtain a unitary of the form
\begin{equation}
\label{generic}
M_{n-1}(A,B) = \begin{array}{l}
  \includegraphics[width=2in]{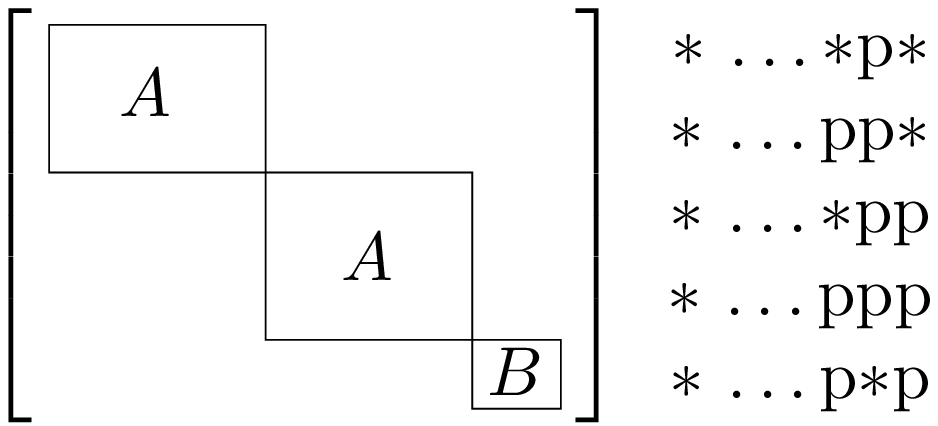} \end{array}
\end{equation}
where $A$ and $B$ are arbitrary unitaries of the appropriate
dimension. The elementary crossing $\sigma_n$ on the last two strands
has the representation\footnote{Here and throughout this section when
  we write a scalar $\alpha$ in a block of the matrix we really mean
  $\alpha I$ where $I$ is the identity operator of appropriate dimension.}
\[
M_n = \left[ \begin{array}{ccccc}
b &   &   &   &   \\
  & a &   &   &   \\
  &   & a &   &   \\
  &   &   & e & d \\
  &   &   & d & c
\end{array} \right]
\begin{array}{c}
\textrm{$*$ \ldots $*$p$*$} \\
\textrm{$*$ \ldots pp$*$} \\
\textrm{$*$ \ldots $*$pp} \\
\textrm{$*$ \ldots ppp} \\
\textrm{$*$ \ldots p$*$p}
\end{array}.
\]
As a special case of equation \ref{generic}, we can obtain
\[
M_{\mathrm{diag}}(\alpha) = \left[ \begin{array}{ccccc}
e^{i \alpha/2} &   &   &   &   \\
  & e^{i \alpha/2} &   &   &   \\
  &   & e^{i \alpha/2} &   &   \\
  &   &   & e^{i \alpha/2} &   \\
  &   &   &   & e^{-i \alpha/2}
\end{array} \right]
\begin{array}{c}
\textrm{$*$ \ldots $*$p$*$} \\
\textrm{$*$ \ldots pp$*$} \\
\textrm{$*$ \ldots $*$pp} \\
\textrm{$*$ \ldots ppp} \\
\textrm{$*$ \ldots p$*$p}
\end{array}.
\]
Where $0 \leq \alpha < 2 \pi$. We'll now show the following.
\begin{lemma}
\label{twodim}
For any element 
\[
\left[ \begin{array}{cc}
V_{11} & V_{12} \\
V_{21} & V_{22}
\end{array} \right] \in SU(2),
\]
one can find some product $P$ of $O(1)$ $M_{\mathrm{diag}}$ matrices and
$M_n$ matrices such that for some phases $\phi_1$ and $\phi_2$,   
\[
P = \left[ \begin{array}{ccccc}
\phi_1 &        &        &        &   \\
       & \phi_2 &        &        &   \\
       &        & \phi_2 &        &   \\
       &        &        & V_{11} & V_{12} \\
       &        &        & V_{21} & V_{22}
\end{array} \right]
\begin{array}{c}
\textrm{$*$ \ldots $*$p$*$} \\
\textrm{$*$ \ldots pp$*$} \\
\textrm{$*$ \ldots $*$pp} \\
\textrm{$*$ \ldots ppp} \\
\textrm{$*$ \ldots p$*$p}
\end{array}.
\]
\end{lemma}
\begin{proof}
Let $B_{\mathrm{diag}}(\alpha)$ and $B_n$ be the following $2 \times
2$ matrices 
\[
B_{\mathrm{diag}}(\alpha) = \left[ \begin{array}{cc}
e^{i \alpha/2} & 0 \\
0 & e^{- i \alpha/2}
\end{array} \right]
\quad \mathrm{and} \quad
B_n = \left[ \begin{array}{cc}
e & d \\
d & c
\end{array}
\right]
\]
We wish to show that we can approximate an arbitrary element of $SU(2)$ as
a product of $O(1)$ $B_{\mathrm{diag}}$ and $B_n$ matrices. To do this, we will
use the well known homomorphism $\phi: SU(2) \to SO(3)$ whose kernel
is $\{ \pm \id \}$ (see section \ref{density}). To obtain an arbitrary
element $V$ of $SU(2)$ modulo phase it suffices to show that
the we can use $\phi(B_n)$ and $\phi(B_{\mathrm{diag}}(\alpha))$ to
obtain an arbitrary $SO(3)$ rotation. In section \ref{density} we
showed that
\[
\left[ \begin{array}{cc}
a & 0 \\
0 & b
\end{array} \right]
\quad \mathrm{and} \quad
\left[ \begin{array}{cc}
e & d \\
d & c
\end{array}
\right]
\]
correspond to two rotations of $7 \pi/5$ about axes which are
separated by an angle of $\theta_{12} \simeq 1.8091137886\ldots$
By the definition of $\phi$, $\phi(B_{\mathrm{diag}}(\alpha))$ is a rotation
by angle $\alpha$ about the same axis that 
$\phi \left( \left[ \begin{array}{cc} a & 0 \\ 0 & b \end{array}
  \right] \right)$ rotates about. $\phi(B_n^5)$ is a $\pi$
rotation. Hence, $R(\alpha) \equiv \phi(B_n^5 
B_{\mathrm{diag}}(\alpha) B_n^5)$ is a rotation by angle $\alpha$
about an axis which is separated by angle\footnote{We subtract $\pi$
  because the angle between axes of rotation is only defined modulo
  $\pi$. Our convention is that these angles are in $[0,\pi)$.} $2
\theta_{12}-\pi$ from the 
axis that $\phi(B_{\mathrm{diag}}(\alpha))$ rotates about. $Q \equiv
R(\pi) \phi(B_{\mathrm{diag}}(\alpha)) R(\pi)$ is a rotation by angle
$\alpha$ about some axis whose angle of separation from the axis that
$\phi(B_{\mathrm{diag}}(\alpha))$ rotates about is $2 (2
\theta_{12}-\pi) \simeq 0.9532$. Similarly, by geometric
visualization, $\phi(B_{\mathrm{diag}}(\alpha')) Q
\phi(B_{\mathrm{diag}}(-\alpha'))$ is a rotation by $\alpha$
about an axis whose angle of separation from the axis that $Q$ rotates
about is anywhere from $0$ to $2 \times 0.9532$ depending on the value
of $\alpha'$. Since $2 \times 0.9532 > \pi/2$, there exists 
some choice of $\alpha'$ such that this angle of separation is
$\pi/2$. Thus, using the rotations we have constructed we can perform
Euler rotations to obtain an arbitrary rotation. 
\end{proof}

As a special case of lemma \ref{twodim}, we can obtain, up to global
phase,
\[
M_{\mathrm{swap}} = \left[ \begin{array}{ccccc}
\phi_1 &        &        &   &   \\
       & \phi_2 &        &   &   \\
       &        & \phi_2 &   &   \\
       &        &        & 0 & 1 \\
       &        &        & 1 & 0
\end{array} \right]
\begin{array}{c}
\textrm{$*$ \ldots $*$p$*$} \\
\textrm{$*$ \ldots pp$*$} \\
\textrm{$*$ \ldots $*$pp} \\
\textrm{$*$ \ldots ppp} \\
\textrm{$*$ \ldots p$*$p}
\end{array}.
\]
Similarly, we can produce $M_{\mathrm{swap}}^{-1}$. Using
$M_{\mathrm{swap}}$, $M_{\mathrm{swap}}^{-1}$, and equation
\ref{generic} we can produce the matrix
\[
M_C = \begin{array}{l} \includegraphics[width=2in]{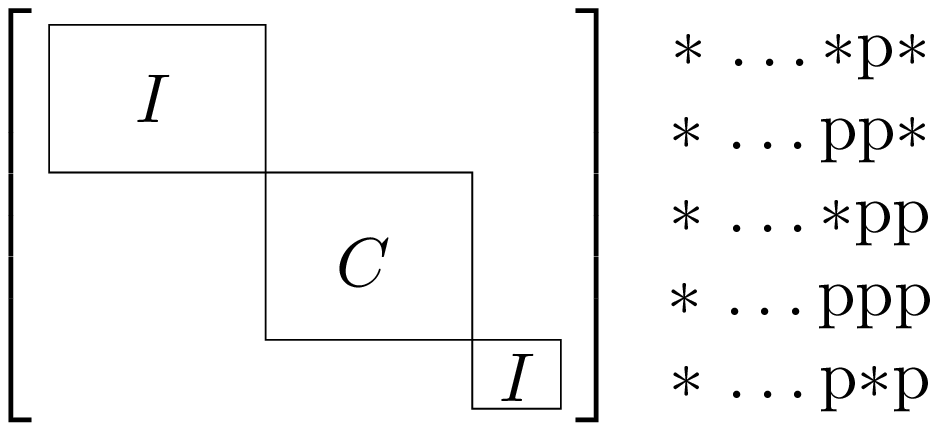} \end{array}
\]
for any unitary $C$. We do it as follows. Since $C$ is a normal
operator, it can be unitarily diagonalized. That is, there exists some
unitary $U$ such that $U C U^{-1} = D$ for some diagonal unitary
$D$. Next, note that in equation \ref{generic} the dimension of $B$ is
more than half that of $A$. Let $d = \mathrm{dim}(A) -
\mathrm{dim}(B)$, and let $I_d$ be the identity operator of dimension
$d$. We can easily construct two diagonal unitaries $D_1$ and $D_2$ of
dimension $\mathrm{dim}(B)$ such that  $(D_1 \oplus I_d)(I_d \oplus
D_2) = D$. As special cases of equation \ref{generic} we can obtain
\[
M_{D_1} = \left[ \begin{array}{ccccc}
1 &   &   &   &   \\
  & 1 &   &   &   \\
  &   & 1 &   &   \\
  &   &   & 1 &  \\
  &   &   &   & D_1
\end{array} \right]
\begin{array}{c}
\textrm{$*$ \ldots $*$p$*$} \\
\textrm{$*$ \ldots pp$*$} \\
\textrm{$*$ \ldots $*$pp} \\
\textrm{$*$ \ldots ppp} \\
\textrm{$*$ \ldots p$*$p}
\end{array}
\]
and
\[
M_{D_2} = \left[ \begin{array}{ccccc}
1 &   &   &   &   \\
  & 1 &   &   &   \\
  &   & 1 &   &   \\
  &   &   & 1 &  \\
  &   &   &   & D_2
\end{array} \right]
\begin{array}{c}
\textrm{$*$ \ldots $*$p$*$} \\
\textrm{$*$ \ldots pp$*$} \\
\textrm{$*$ \ldots $*$pp} \\
\textrm{$*$ \ldots ppp} \\
\textrm{$*$ \ldots p$*$p}
\end{array}
\]
and
\[
M_P = \begin{array}{l} \includegraphics[width=2in]{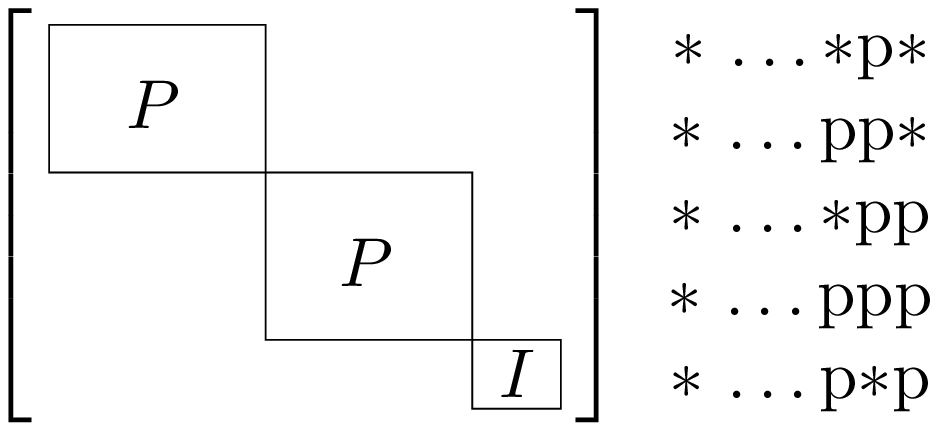} \end{array}
\]
where $P$ is a permutation matrix that shifts the lowest $\dim(B)$
basis states from the bottom of the block to the top of the block.
Thus we obtain
\[
M_2 \equiv M_{\mathrm{swap}} M_{D_2} M_{\mathrm{swap}}^{-1} =
\begin{array}{l}
\includegraphics[width=2in]{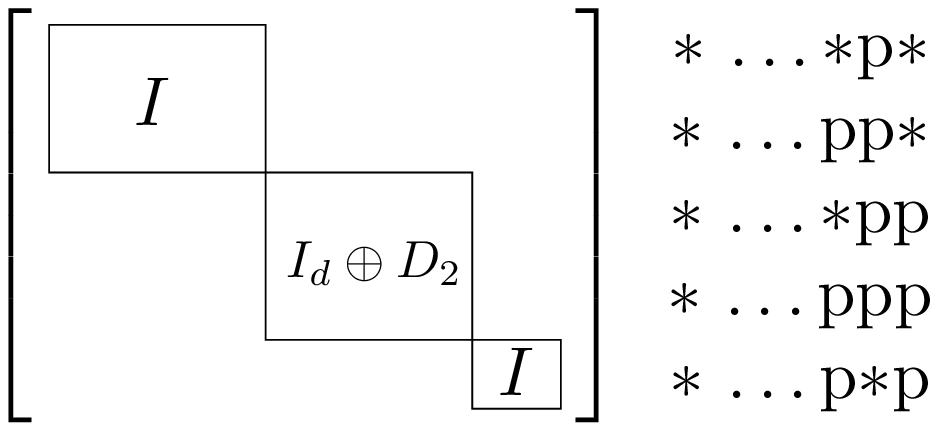}
\end{array}
\]
and
\[
M_1 \equiv M_P M_{\mathrm{swap}} M_{D_1}
M_{\mathrm{swap}}^{-1} M_P^{-1} =
\begin{array}{l}
\includegraphics[width=2in]{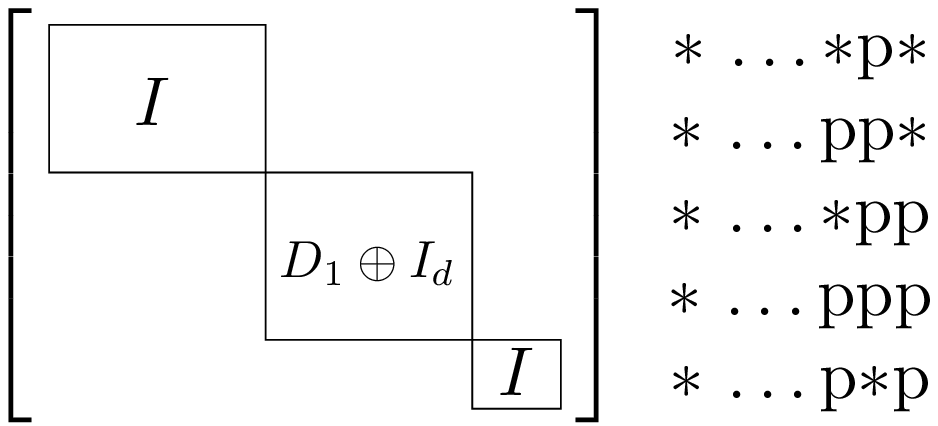}
\end{array}.
\]
Thus
\[
M_1 M_2 = \begin{array}{l}
 \includegraphics[width=2in]{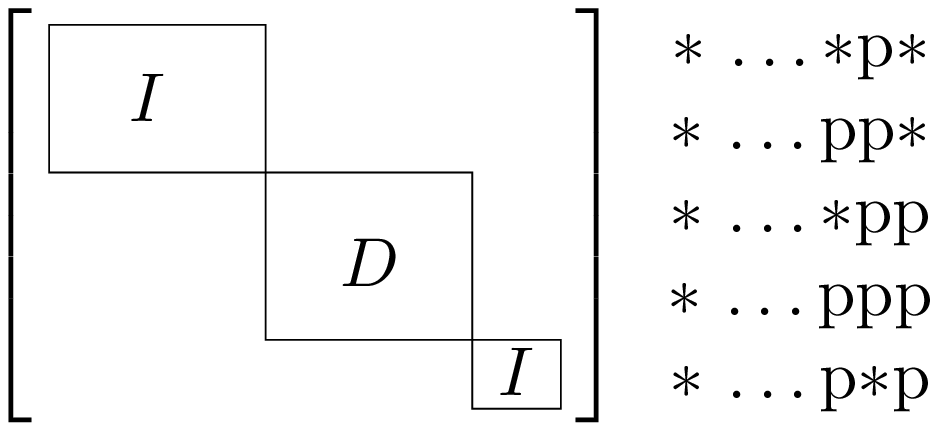}
 \end{array}.
\]
As a special case of equation \ref{generic} we can obtain
\[
M_U = \begin{array}{l}
 \includegraphics[width=2in]{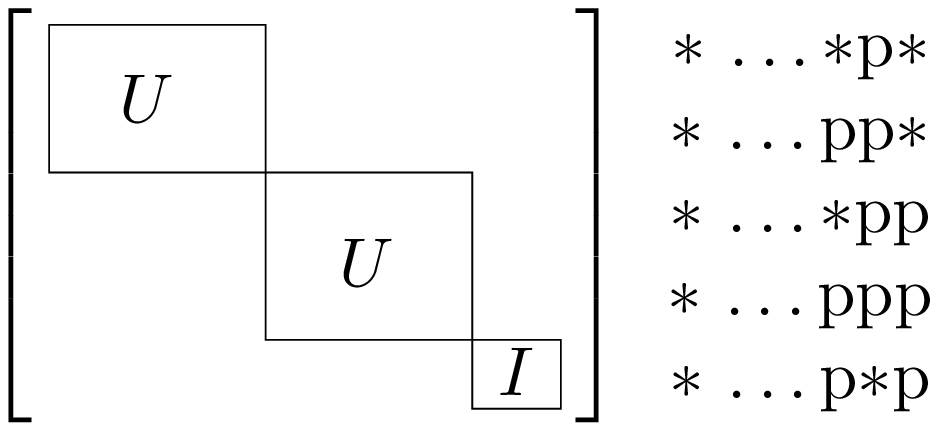}
 \end{array}.
\]
Thus we obtain $M_C$ by the construction $M_C = M_U M_1 M_2
M_U^{-1}$. By multiplying together $M_C$ and $M_{n-1}$ we can control
the three blocks independently. For arbitrary unitaries $A,B,C$ of
appropriate dimension we can obtain
\begin{equation}
\label{indep}
M_{ACB} = \begin{array}{l}
 \includegraphics[width=2in]{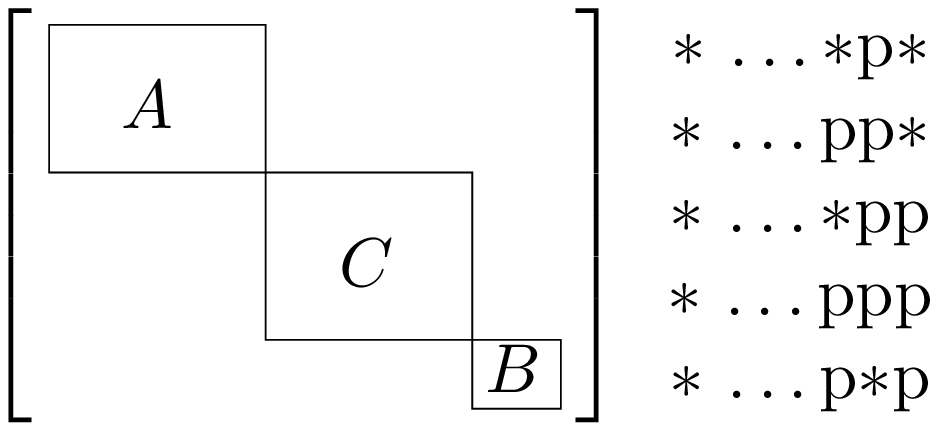}
 \end{array}.
\end{equation}
As a special case of equation \ref{indep} we can obtain
\[
M_{\mathrm{unphase}} =
\left[ \begin{array}{ccccc}
\phi_1^* &          &          &   &   \\
         & \phi_2^* &          &   &   \\
         &          & \phi_2^* &   &   \\
         &          &          & 1 &   \\
         &          &          &   & 1
\end{array} \right]
\begin{array}{c}
\textrm{$*$ \ldots $*$p$*$} \\
\textrm{$*$ \ldots pp$*$} \\
\textrm{$*$ \ldots $*$pp} \\
\textrm{$*$ \ldots ppp} \\
\textrm{$*$ \ldots p$*$p}
\end{array}.
\]
Thus, we obtain a clean swap
\begin{equation}
\label{clean}
M_{\mathrm{clean}} = M_{\mathrm{unphase}} M_{\mathrm{swap}} =
\left[ \begin{array}{ccccc}
1 &   &   &   &   \\
  & 1 &   &   &   \\
  &   & 1 &   &   \\
  &   &   & 0 & 1 \\
  &   &   & 1 & 0
\end{array} \right]
\begin{array}{c}
\textrm{$*$ \ldots $*$p$*$} \\
\textrm{$*$ \ldots pp$*$} \\
\textrm{$*$ \ldots $*$pp} \\
\textrm{$*$ \ldots ppp} \\
\textrm{$*$ \ldots p$*$p}
\end{array}.
\end{equation}
We'll now use $M_{\mathrm{clean}}$ and $M_{ACB}$ as our building
blocks to create the maximally general unitary
\begin{equation}
\label{maxgen}
M_{\mathrm{gen}}(V,W) = \begin{array}{l}
 \includegraphics[width=2in]{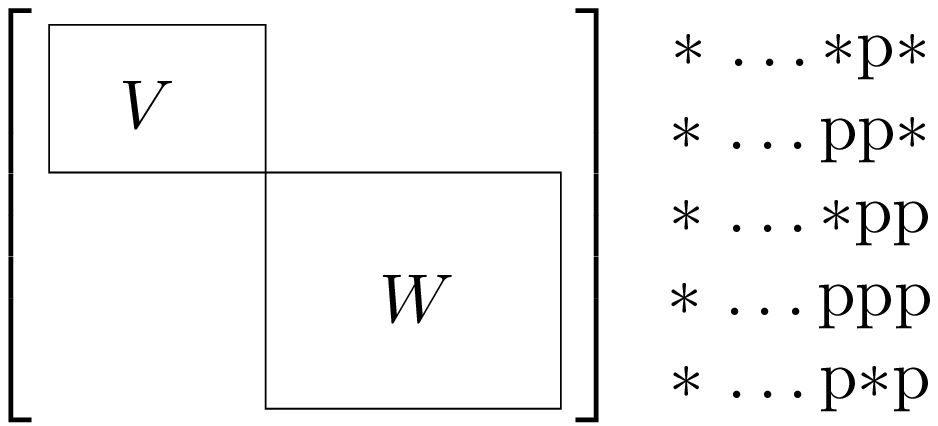}
 \end{array}.
\end{equation}

For $n+1$ symbols, the $* \ldots *pp$
subspace has dimension $f_{n-3}$, and the $* \ldots  p*p$ and $*
\ldots ppp$ subspaces each have dimension $f_{n-2}$. Thus, in equation
\ref{indep}, the block $C$ has dimension $f_{n-2} + f_{n-3} =
f_{n-1}$, and the block $B$ has dimension $f_{n-2}$. To construct
$M_{\mathrm{gen}}(V,W)$ we will choose a subset of the $f_n$ basis
states acted upon by the $B$ and $C$ blocks and permute them
into the $C$ block. Then using $M_{ACB}$, we'll perform an arbitrary unitary
on these basis states. At each such step we can act upon a subspace
whose dimension is a constant fraction of the dimension of the entire
$f_n$ dimensional space on which we wish to apply an arbitrary
unitary. Furthermore, this constant fraction is more than
half. Specifically, $f_n/f_{n-1} \simeq 1/\phi \simeq 0.62$ for large
$n$. We'll show that an arbitrary unitary can be built up as a product
of a constant number of unitaries each of which act only on half the
basis states. Thus our ability to act on approximately $62\%$ of the
basis states at each step is more than sufficient. 

Before proving this, we'll show how to permute an arbitrary set of
basis states into the $C$ block of $M_{ACB}$. Just use
$M_{\mathrm{clean}}$ to  swap the $B$ block into the $* \ldots ppp$
subspace of the $C$ block. Then, as a special case of equation
\ref{indep}, choose $A$ and $B$ to be the identity, and $C$ to be a
permutation which swaps some states between the $* \ldots *pp$ and $*
\ldots ppp$ subspaces of the $C$ block. The states which we swap up
from the $* \ldots ppp$ subspace are the ones from $B$ which we wish
to move into $C$. The ones which we swap down from the $* \ldots *pp$
subspace are the ones from $C$ which we wish to move into $B$. This
process allows us to swap a maximum of $f_{n-3}$ states between the
$B$ block and the $C$ block. Since $f_{n-3}$ is more than half the
dimension of the $B$ block, it follows that any desired permutation of
states between the $B$ and $C$ blocks can be achieved using two
repetitions of this process.

We'll now show the following.
\begin{lemma}
Let $m$ by divisible by 4. Any $m \times m$ unitary can be obtained as
a product of seven unitaries, each of which act only on the space
spanned by $m/2$ of the basis states, and leave the rest of the basis
states undisturbed.
\end{lemma}
It will be obvious from the proof that even if the dimension of the
matrix is not divisible by four, and the fraction of the basis states on
which the individual unitaries act is not exactly $1/2$ it will still
be possible to obtain an arbitrary unitary using a constant number of
steps independent of $m$. Therefore, we will not explicitly work out this
straightforward generalization.
\\
\begin{proof}
In \cite{Nielsen_Chuang} it is shown that for any unitary $U$, one can always
find a series of unitaries $L_n, \ldots, L_1$ which each act on only
two basis states such that $L_n \ldots L_1 U$ is the identity. Thus
$L_n \ldots L_1 = U^{-1}$. It follows that any unitary can be
obtained as a product of such two level unitaries. The individual
matrices $L_1, \ldots, L_n$ each perform a (unitary) row
operation on $U$. The sequence $L_n \ldots L_1$ reduces $U$ to
the identity by a method very similar to Gaussian elimination. We will
use a very similar construction to prove the present lemma. The
essential difference is that we must perform the two level unitaries
in groups. That is, we choose some set of $m/2$ basis states, perform
a series of two level unitaries on them, then choose another set of
$m/2$ basis states, perform a series of two level unitaries on them,
and so on. After a finite number of such steps (it turns out that
seven will suffice) we will reduce $U$ to the identity.

Our two-level unitaries will all be of the same type. We'll fix our
attention on two entries in $U$ taken from a particular column:
$U_{ik}$ and $U_{jk}$. We wish to perform a unitary row operation,
\emph{i.e.} left multiply by a two level unitary, to set
$U_{jk}=0$. If $U_{ik}$ and $U_{jk}$ are not both zero, then the
two-level unitary which acts on the rows $i$ and $j$ according to
\begin{equation}
\label{row_op}
\frac{1}{\sqrt{|U_{ik}|^2 + |U_{jk}|^2}} 
\left[ \begin{array}{cc}
U_{ik}^* & U_{jk}^* \\
U_{jk} & - U_{ik}
\end{array} \right]
\end{equation}
will achieve this. If $U_{ik}$ and $U_{jk}$ are both zero there is
nothing to be done.

We can now use this two level operation within groups of basis states
to eliminate matrix elements of $U$ one by one. As in Gaussian
elimination, the key is that once you've obtained some zero matrix
elements, your subsequent row operations must be chosen so that they
do not make these nonzero again, undoing your previous work.

As the first step, we'll act on the top $m/2$ rows in order to reduce
the upper-left quadrant of $U$ to upper triangular form. We can do
this as follows. Consider the first and second entries in the first
column. Using the operation \ref{row_op} we can make the second entry
zero. Next consider the first and third entries in the first
column. By operation \ref{row_op} we can similarly make the third
entry zero. Repeating this procedure, we get all of the entries in the
top half of the first column to be zero other than the top
entry. Next, we perform the same procedure on the second column except
leaving out the top row. These row operations will not alter the first
column since the rows being acted upon all have zero in the first
column. We can then repeat this procedure for each column in the left
half of $U$ until the upper-left block is upper triangular. 

We'll now think of U in terms of 16 blocks of size
$(m/4) \times (m/4)$. In the second step we'll eliminate the matrix
elements in the third block of the first column. The second step is
shown schematically as
\[
\begin{array}{l} \includegraphics[width=2.0in]{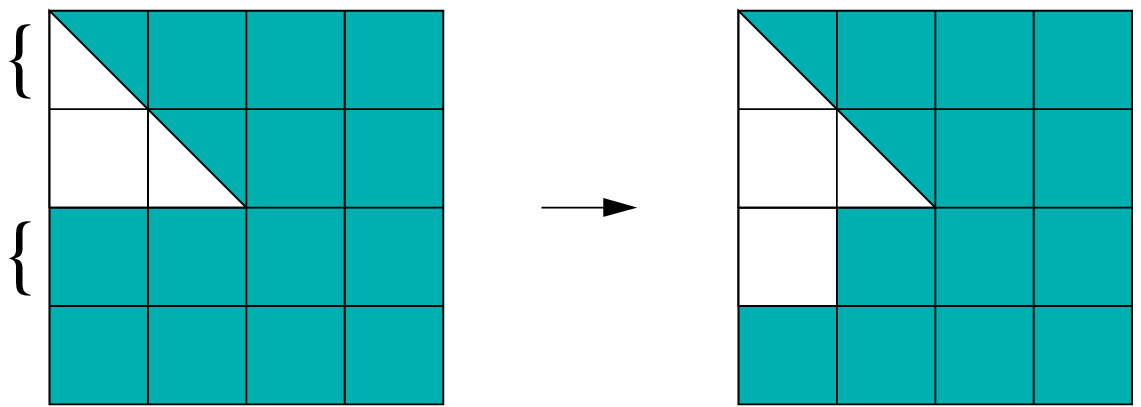} \end{array}
\]
The curly braces indicate the rows to be acted upon, and the unshaded
areas represent zero matrix elements. This step can be performed very
similarly to the first step. The nonzero matrix elements in the bottom
part of the first column can be eliminated one by one by interacting
with the first row. The nonzero matrix elements in the bottom part of
the second column can then be eliminated one by one by interacting
with the second row. The first column will be undisturbed by this
because the rows being acted upon in this step have zero matrix
elements in the first column. Similarly acting on the remaining
columns yields the desired result.

The next step, as shown below, is nearly identical and can be done the
same way.
\begin{equation}
\label{step3_eq}
\begin{array}{l} \includegraphics[width=2.0in]{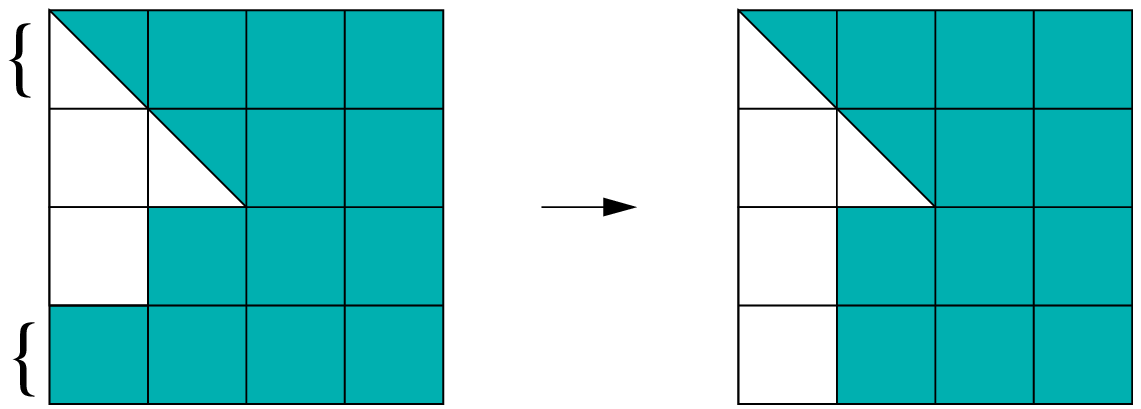} \end{array}
\end{equation}
The matrix on the right hand side of \ref{step3_eq} is unitary. It
follows that it must be of the form
\[
\begin{array}{l} \includegraphics[width=0.7in]{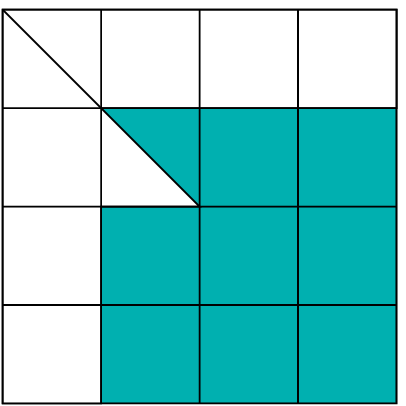} \end{array}
\]
where the upper-leftmost block is a diagonal unitary. We can next
apply the same sorts of steps to the lower $3 \times 3$ blocks, as
illustrated below.
\[
\begin{array}{l} \includegraphics[width=2.9in]{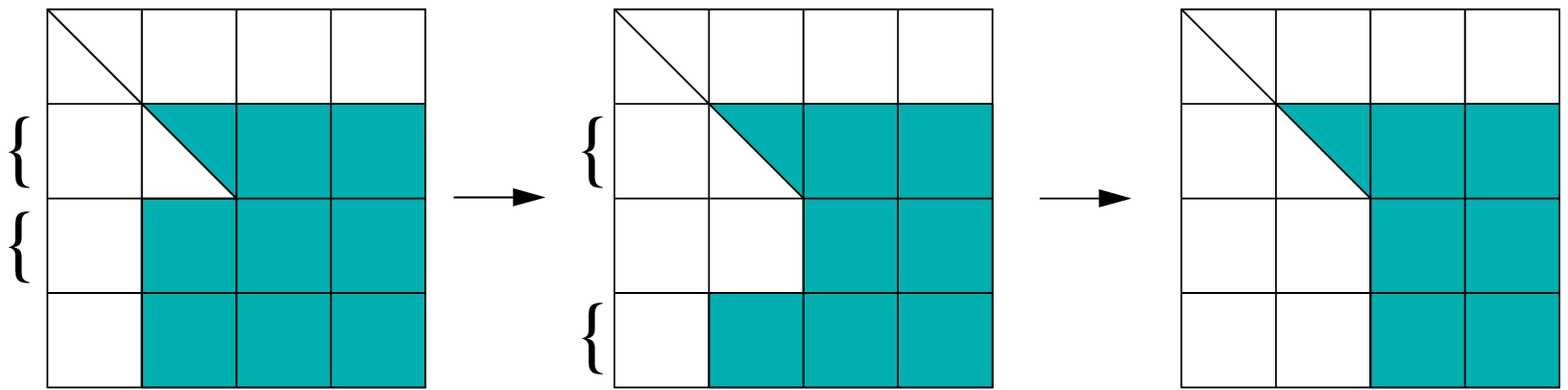} \end{array}
\]
By unitarity the resulting matrix is actually of the form
\[
\begin{array}{l} \includegraphics[width=0.7in]{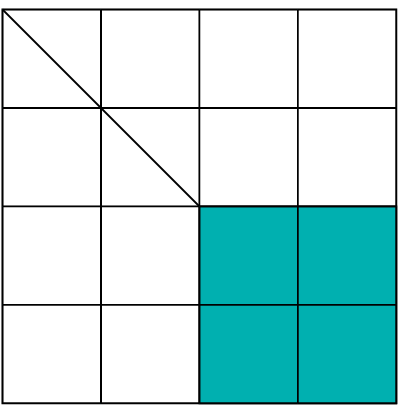} \end{array}
\]
where the lower-right quadrant is an $(m/2) \times (m/2)$ unitary matrix,
and the upper-left quadrant is an $(m/2) \times (m/2)$ diagonal unitary
matrix. We can now apply the inverse of the upper-left quadrant to the
top $m/2$ rows and then apply the inverse of the lower-right quadrant
to the bottom $m/2$ rows. This results in the identity matrix, and we
are done. In total we have used seven steps.
\end{proof}

Examining the preceding construction, we can see the recursive step
uses a constant number of the $M_{n-1}$ operators from the next lower
level of recursion, plus a constant number of $M_n$ operators. Thus,
the number of crossings in the braid grows only exponentially in the
recursion depth. Since each recursion adds one more symbol, we see
that to construct $M_{\mathrm{gen}}(V,W)$ on logarithmically many
symbols requires only polynomially many crossings in the corresponding
braid.

The main remaining task is to work out the base case on which the
recursion rests. Since the base case is for a fixed set of generators
on a fixed number of symbols, we can simply use the Solovay-Kitaev
theorem\cite{Kitaev}.
\begin{theorem}[Solovay-Kitaev]
\label{Solovay-Kitaev}
Suppose matrices $U_1, \ldots, U_r$ generate a dense subgroup in
$SU(d)$. Then, given a desired unitary $U \in SU(d)$, and a precision
parameter $\delta > 0$, there is an algorithm to find a product $V$ of
$U_1, \ldots, U_r$ and their inverses such that $\| V - U \| \leq
\delta$. The length of the product and the runtime of the algorithm
are both polynomial in $\log(1/\delta)$.
\end{theorem}

Because the total complexity of the process is polynomial, it
is only necessary to implement the base case to polynomially small
$\delta$ in order for the final unitary $M_{\mathrm{gen}}(V,W)$ to
have polynomial precision. This follows from simple error propagation.
An analogous statement about the precision of gates needed in
quantum circuits is worked out in \cite{Nielsen_Chuang}. This completes the
proof of proposition \ref{efficiency}.

\section{Zeckendorf Representation}
\label{bijective}

Following \cite{Kauffman}, to construct the Fibonacci representation
of the braid group, we use strings of p and $*$ symbols such that no
two $*$ symbols are adjacent. There exists a bijection $z$ between
such strings and the integers, known as the Zeckendorf
representation. Let $P_n$ be the set of all such strings of length
$n$. To construct the map $z:P_n \to \{0,1,\ldots,f_{n+2} \}$ 
we think of $*$ as one and $p$ as zero. Then, for a given 
string $s = s_n s_{n-1} \ldots s_1$ we associate the integer
\begin{equation}
\label{recap}
z(s) = \sum_{i=1}^n s_i f_{i+1},
\end{equation}
where $f_i$ is the $i\th$ Fibonacci number: $f_1 = 1, f_2=1, f_3 = 2$,
and so on. In this section we'll show the following.
\begin{proposition}
\label{bij}
For any $n$, the map $z: P_n \to \{0, \ldots, f_{n+2} \}$ defined by
  $z(s) = \sum_{i=1}^n s_i f_{i+1}$ is bijective.
\end{proposition}

\begin{proof}
We'll inductively show that the following two statements
are true for every $n \geq 2$. 
\\ \\
$\mathbf{A_n:}\quad$
\emph{$z$ maps strings of length $n$ starting with p bijectively to
  $\{0,\ldots, f_{n+1}-1 \}$.}
\\ \\
$\mathbf{B_n:}\quad$
\emph{$z$ maps strings of length $n$ starting with $*$ bijectively to
  $\{f_{n+1}, \ldots, f_{n+2}-1 \}$.}
\\ \\
Together, $A_n$ and $B_n$ imply that $z$ maps $P_n$ bijectively to
$\{0, \ldots, f_{n+2}-1 \}$. As a base case, we can look at $n=2$. 
\begin{eqnarray*}
pp & \leftrightarrow & 0 \\
p* & \leftrightarrow & 1 \\
*p & \leftrightarrow & 2
\end{eqnarray*}
Thus $A_2$ and $B_2$ are true. Now for the induction. Let $s_{n-1} \in
P_{n-1}$. By equation 
\ref{recap},
\[
z(p s_{n-1}) = z(s_{n-1}).
\]
Since $s_{n-1}$ follows a p symbol, it can be any element of
$P_{n-1}$. By induction, $z$ is bijective on $P_{n-1}$, thus $A_n$ is
true. Similarly, by equation 
\ref{recap}
\[
z(* s_{n-1}) = f_{n+1} + z(s_{n-1}).
\]
Since $s_{n-1}$ here follows a $*$, its allowed values are exactly
those strings which start with p. By induction, $A_{n-1}$ tells us
that $z$ maps these bijectively to $\{0, \ldots, f_n-1 \}$. Since
$f_{n+1} + f_n = f_{n+2}$, this implies $B_n$ is true. Together, $A_n$
and $B_n$ for all $n \geq 2$, along with the trivial $n=1$ case,
imply proposition \ref{bij}.
\end{proof}

\chapter{Perturbative Gadgets}
\label{gadgets}

\section{Introduction}
\label{intro}

Perturbative gadgets were introduced to construct a two-local Hamiltonian
whose low energy effective Hamiltonian corresponds to a desired
three-local Hamiltonian. They were originally
developed by Kempe, Kitaev, and Regev in 2004 to prove the
QMA-completeness of the 2-local Hamiltonian problem and to simulate 
3-local adiabatic quantum computation using 2-local adiabatic quantum
computation\cite{Kempe}. Perturbative gadgets have subsequently been
used to simulate spatially nonlocal Hamiltonians using spatially local
Hamiltonians\cite{Oliveira}, and to find a minimal set of set of
interactions for universal adiabatic quantum
computation\cite{Biamonte}. It was also pointed out in \cite{Oliveira}
that perturbative gadgets can be used recursively to obtain $k$-local
effective interactions using a 2-local Hamiltonian. Here we generalize
perturbative gadgets to directly obtain arbitrary $k$-local effective
interactions by a single application of $k\th$ order perturbation
theory. Our formulation is based on a perturbation expansion due to
Bloch\cite{Bloch}.

A $k$-local operator is one consisting of interactions between at most
$k$ qubits. A general $k$-local Hamiltonian on $n$ qubits can always
be expressed as a sum of $r$ terms,
\begin{equation}
\label{hcomp}
H^{\mathrm{comp}} = \sum_{s=1}^r c_s H_s
\end{equation}
with coefficients $c_s$, where each term $H_s$ is a $k$-fold
tensor product of Pauli operators. That is, $H_s$ couples some set of
$k$ qubits according to
\begin{equation}
H_s = \sigma_{s,1} \ \sigma_{s,2} \ \ldots \ \sigma_{s,k},
\end{equation}
where each operator $\sigma_{s,j}$ is of the form
\begin{equation}
\sigma_{s,j} = \hat{n}_{s,j} \cdot \vec{\sigma}_{s,j},
\end{equation} 
where $\hat{n}_{s,j}$ is a unit vector in $\mathbb{R}^3$, and 
$\vec{\sigma}_{s,j}$ is the vector of Pauli matrices operating on the
$j\th$ qubit in the set of $k$ qubits acted upon by $H_s$.

We wish to simulate $H^{\mathrm{comp}}$ using only 2-local
interactions. To this end, for each term $H_s$, we introduce $k$
ancilla qubits, generalizing the technique of \cite{Kempe}. There are
then $rk$ ancilla qubits and $n$ computational qubits, and we choose
the gadget Hamiltonian\footnote{For $H^{\mathrm{gad}}$ to be 
  Hermitian, the coefficient of $V_s$ must be real. We therefore
  choose the sign of each $\hat{n}_{s,j}$ so that all $c_s$ are
  positive.} as
\begin{equation}
H^{\mathrm{gad}} = \sum_{s=1}^r H_s^{\mathrm{anc}} + \lambda
\sum_{s=1}^r \sqrt[k]{c_s} V_s,
\end{equation}
where
\begin{equation}
\label{nonpert}
H_s^{\mathrm{anc}} = \sum_{1 \leq i < j \leq k } \frac{1}{2} 
(I - Z_{s,i} Z_{s,j}),
\end{equation}
and
\begin{equation}
\label{pert}
V_s = \sum_{j = 1}^k \sigma_{s,j} \otimes X_{s,j}.
\end{equation}
For each $s$ there is a corresponding register of $k$ ancilla
qubits. The operators $X_{s,j}$ and $Z_{s,j}$ are Pauli $X$ and
$Z$ operators acting on the $j\th$ ancilla qubit in the ancilla
register associated with $s$. For each ancilla register, the ground
space of $H_s^{\mathrm{anc}}$ is the span of $\ket{000\ldots}$ and
$\ket{111\ldots}$. $\lambda$ is the small parameter in which the
perturbative analysis is carried out.

\begin{figure}
\begin{center}
\includegraphics[width=0.15\textwidth]{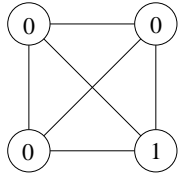}
\caption{The ancilla qubits are all coupled together
  using $ZZ$ couplings. This gives a unit energy penalty for each pair
  of unaligned qubits. If there are $k$ bits, of which $j$ are in the
  state $\ket{1}$ and the remaining $k-j$ are in the state $\ket{0}$,
  then the energy penalty is $j(k-j)$. In the example shown in this
  diagram, the 1 and 0 labels indicate that the qubits are in the
  state $\ket{0001}$, which has energy penalty 3.}
\end{center}
\end{figure}

For each $s$, the operator 
\begin{equation}
X_s^{\otimes k} = X_{s,1} \otimes X_{s,2} \otimes \ldots \otimes X_{s,k}
\end{equation}
acting on the $k$ ancilla qubits in the 
register $s$ commutes with $H^{\mathrm{gad}}$. Since there are $r$
ancilla registers, $H^{\mathrm{gad}}$ can be block diagonalized into
$2^r$ blocks, where each register is in either the $+1$ or $-1$
eigenspace of its $X_s^{\otimes k}$. In this paper, we analyze only the 
block corresponding to the $+1$ eigenspace for every register. This
$+1$ block of the gadget Hamiltonian is a Hermitian operator,
that we label $H_+^{\mathrm{gad}}$. We show that the
effective Hamiltonian on the low energy eigenstates of
$H_+^{\mathrm{gad}}$ approximates $H^{\mathrm{comp}}$. For many
purposes this is sufficient. For example, suppose one wishes to
simulate a $k$-local adiabatic quantum computer using a $2$-local
adiabatic quantum computer. If the initial state of the computer lies
within the all $+1$ subspace, then the system will remain in this
subspace throughout its evolution. To put the initial state of the
system into the all $+1$ subspace, one can initialize each ancilla
register to the state
\begin{equation}
\ket{+} = \frac{1}{\sqrt{2}} (\ket{000\ldots} + \ket{111\ldots}),
\end{equation}
which is the ground state of $\sum_s H_s^{\mathrm{anc}}$ within the
$+1$ subspace. Given the extensive experimental literature on the
preparation of states of the form $\ket{+}$, also known as cat states,
a supply of such states seems a reasonable resource to assume.

The purpose of the perturbative gadgets is to obtain $k$-local
effective interactions in the low energy subspace. To quantify this,
we use the concept of an effective Hamiltonian. We define this
to be
\begin{equation}
H_{\mathrm{eff}}(H,d) \equiv \sum_{j=1}^d E_j \ket{\psi_j}\bra{\psi_j},
\end{equation}
where $\ket{\psi_1},\ldots,\ket{\psi_d}$ are the $d$ lowest energy
eigenstates of a Hamiltonian $H$, and $E_1, \ldots , E_d$ are
their energies.

In section \ref{section_analysis}, we calculate
$H_{\mathrm{eff}}(H^{\mathrm{gad}}_+,2^n)$ perturbatively to $k\th$
order in $\lambda$. To do this, we write $H^{\mathrm{gad}}$ as
\begin{equation}
H^{\mathrm{gad}} = H^{\mathrm{anc}} + \lambda V
\end{equation}
where
\begin{equation}
H^{\mathrm{anc}} = \sum_{s=1}^r H_s^{\mathrm{anc}}
\end{equation}
and
\begin{equation}
V = \sum_{s=1}^r \sqrt[k]{c_s} V_s
\end{equation}
We consider $H^{\mathrm{anc}}$ to be the unperturbed Hamiltonian and
$\lambda V$ to be the the perturbation. We find that $\lambda V$
perturbs the ground space of $H^{\mathrm{anc}}$ in two separate
ways. The first is to shift the energy of the entire space. The second
is to split the degeneracy of the ground space. This splitting arises
at $k\th$ order in perturbation theory, because the lowest power of
$\lambda V$ that has nonzero matrix elements within the ground space
of $H^{\mathrm{anc}}$ is the $k\th$ power. It is this splitting which
allows the low energy subspace of $H^{\mathrm{gad}}_+$ to mimic the
spectrum of $H^{\mathrm{comp}}$.

It is convenient to analyze the shift and the splitting
separately. To do this, we define 
\begin{equation}
\widetilde{H}_{\mathrm{eff}}(H,d,\Delta) \equiv H_{\mathrm{eff}}(H,d)
- \Delta \Pi,
\end{equation}
where $\Pi$ is the projector onto the support of
$H_{\mathrm{eff}}(H,d)$. Thus, $\widetilde{H}_{\mathrm{eff}}(H,d,\Delta)$
differs from $H_{\mathrm{eff}}(H,d)$ only by an energy shift of
magnitude $\Delta$. The eigenstates of
$\widetilde{H}_{\mathrm{eff}}(H,d,\Delta)$ are identical to the
eigenstates of $H_{\mathrm{eff}}(H,d)$, as are all the gaps between
eigenenergies. The rest of this paper is devoted to showing that,
for any $k$-local Hamiltonian $H^{\mathrm{comp}}$ acting on $n$
qubits, there exists some function $f(\lambda)$ such that 
\begin{equation}
\label{result}
\widetilde{H}_{\mathrm{eff}}(H^{\mathrm{gad}}_+,2^n,f(\lambda))
= \frac{-k (-\lambda)^k}{(k-1)!}
H^{\mathrm{comp}} \otimes P_+ + \mathcal{O}(\lambda^{k+1})
\end{equation}
for sufficiently small $\lambda$. Here $P_+$ is an operator acting on
the ancilla registers, projecting each one into the state
$\ket{+}$. To obtain equation \ref{result} we use a formulation of
degenerate perturbation theory due to Bloch\cite{Bloch, Messiah},
which we describe in the next section.

\section{Perturbation Theory}
\label{perturbation}

Suppose we have a Hamiltonian of the form 
\begin{equation}
H = H^{(0)} + \lambda V,
\end{equation}
where $H^{(0)}$ has a $d$-dimensional degenerate ground space
$\mathcal{E}^{(0)}$ of energy zero. As discussed in \cite{Kato, Messiah},
the effective Hamiltonian for the $d$ lowest eigenstates of $H$ can be
obtained directly as a perturbation series in $V$. However, for our
purposes it is more convenient to use an indirect method due to
Bloch \cite{Bloch, Messiah}, which we now describe. As shown in
section \ref{convergence}, the perturbative expansions converge
provided that
\begin{equation}
\label{conv_cond}
\| \lambda V \| < \frac{\gamma}{4},
\end{equation}
where $\gamma$ is the energy gap between the eigenspace in question
and the next nearest eigenspace, and $\| \cdot \|$ denotes the
operator norm\footnote{For any linear operator $M$,
\[
\| M \| \equiv \max_{|\braket{\psi}{\psi}| = 1} | \bra{\psi} M
\ket{\psi} |.
\]}.

Let $\ket{\psi_1}, \ldots, \ket{\psi_d}$ be the $d$ lowest energy
eigenstates of $H$, and let $E_1,\ldots,E_d$
be their energies. For small perturbations, these states lie
primarily within $\mathcal{E}^{(0)}$. Let 
\begin{equation}
\ket{\alpha_j} = P_0 \ket{\psi_j},
\end{equation}
where $P_0$ is the projector onto $\mathcal{E}^{(0)}$. For $\lambda$
satisfying \ref{conv_cond}, the  vectors $\ket{\alpha_1}, \ldots,
\ket{\alpha_d}$ are linearly independent, and there exists a
linear operator $\cu$ such that 
\begin{equation}
\ket{\psi_j} = \cu \ket{\alpha_j} \quad \textrm{for $j=1,2,\ldots, d$}
\end{equation}
and
\begin{equation}
\cu \ket{\phi} = 0 \quad \mathrm{for} \quad \ket{\phi} \in
\mathcal{E}^{(0)\perp}.
\end{equation}
Note that $\cu$ is in general nonunitary. Let
\begin{equation}
\label{afromu}
\ca = \lambda P_0 V \cu.
\end{equation}
As shown in \cite{Messiah, Bloch} and recounted in section
\ref{Blochapp}, the eigenvectors of $\ca$ are
$\ket{\alpha_1},\ldots,\ket{\alpha_d}$, and the corresponding
eigenvalues are $E_1, \ldots, E_d$. Thus,
\begin{equation}
\label{heffudag}
H_{\mathrm{eff}} = \cu \ca \cu^\dag.
\end{equation}
$\ca$ and $\cu$ have the following perturbative expansions. Let $S^l$
be the operator
\begin{equation}
\label{S}
S^l = \left\{ \begin{array}{ll}
\displaystyle \sum_{j \neq 0} \frac{P_j}{(-E^{(0)}_j)^l} & 
\textrm{if $l >  0$} \\ \\
\displaystyle -P_0 & \textrm{if $l = 0$} 
\end{array} \right.
\end{equation}
where $P_j$ is the projector onto the eigenspace of $H^{(0)}$ with
energy $E^{(0)}_j$. (Recall that $E^{(0)}_0 = 0$.) Then
\begin{equation}
\label{ca}
\ca = \sum_{m=1}^\infty \ca^{(m)},
\end{equation}
where
\begin{equation}
\label{caterms}
\ca^{(m)} = \lambda^m \sum_{(m-1)} P_0 V S^{l_1} V S^{l_2} \ldots V
S^{l_{m-1}} V P_0, 
\end{equation}
and the sum is over all nonnegative integers $l_1 \ldots l_{m-1}$
satisfying
\begin{eqnarray}
l_1 + \ldots + l_{m-1} & = & m-1 \\
l_1 + \ldots + l_p & \geq & p \quad (p = 1,2,\ldots,m-2).
\end{eqnarray}
Similarly, $\cu$ has the expansion
\begin{equation}
\label{cu}
\cu = P_0 + \sum_{m=1}^\infty \cu^{(m)},
\end{equation}
where
\begin{equation}
\label{cuterms}
\cu^{(m)} = \lambda^m \sum_{(m)} S^{l_1} V S^{l_2} V \ldots V
S^{l_m} V P_0,
\end{equation}
and the sum is over
\begin{eqnarray}
l_1 + \ldots + l_m & = & m \\
l_1 + \ldots + l_p & \geq & p \quad (p=1,2,\ldots,m-1).
\end{eqnarray}
In section \ref{Blochapp} we derive the expansions for $\cu$ and
$\ca$, and in section \ref{convergence} we prove that condition
\ref{conv_cond} suffices to ensure convergence. The advantage of the
method of \cite{Bloch} over the direct approach of \cite{Kato} is that
$\ca$ is an operator whose support is strictly within
$\mathcal{E}^{(0)}$, which makes some of the calculations more
convenient.

\section{Analysis of the Gadget Hamiltonian}
\label{section_analysis}

Before analyzing $H^{\mathrm{gad}}$ for a general $k$-local
Hamiltonian, we first consider the case where $H^{\mathrm{comp}}$
has one coefficient $c_s = 1$ and all the rest equal to zero. That
is,
\begin{equation}
H^{\mathrm{comp}} = \sigma_1 \sigma_2 \ldots \sigma_k, 
\end{equation}
where for each $j$, $\sigma_j = \hat{n}_j \cdot \vec{\sigma}_j$ for some
unit vector $\hat{n}_j$ in $\mathbb{R}^3$. The corresponding gadget
Hamiltonian is thus 
\begin{equation}
H^{\mathrm{gad}} = H^{\mathrm{anc}} + \lambda V, 
\end{equation}
where
\begin{equation}
\label{nonpert_recap}
H^{\mathrm{anc}} = \sum_{1 \leq i < j \leq k} \frac{1}{2}(I-Z_i Z_j),
\end{equation}
and
\begin{equation}
\label{pert_recap}
V = \sum_{j=1}^k \sigma_j \otimes X_j.
\end{equation}
Here $\sigma_j$ acts on the $j\th$ computational qubit, and $X_j$ and
$Z_j$ are the Pauli $X$ and $Z$ operators acting on the $j\th$ ancilla
qubit. We use $k\th$ order perturbation theory to show that
$\widetilde{H}^{\mathrm{eff}}(H^{\mathrm{gad}}_+,2^k,\Delta)$
approximates $H^{\mathrm{comp}}$ for appropriate $\Delta$.

We start by calculating $\ca$ for $H_+^{\mathrm{gad}}$. For
$H^{\mathrm{anc}}$, the energy gap is 
$\gamma = k-1$, and $\| V \| = k$, so by condition \ref{conv_cond}, we
can use perturbation theory provided $\lambda$ satisfies 
\begin{equation}
\lambda < \frac{k-1}{4k}.
\end{equation}
Because all terms in $\ca$ are sandwiched by $P_0$ operators, the
nonzero terms in $\ca$ are ones in which the $m$ powers of $V$ take a
state in $\mathcal{E}^{(0)}$ and return it to
$\mathcal{E}^{(0)}$. Because we are working in the $+1$ eigenspace of
$X^{\otimes k}$, an examination of equation \ref{nonpert_recap} shows that
$\mathcal{E}^{(0)}$ is the span of the states in which the ancilla qubits
are in the state $\ket{+}$. Thus, $P_0 = I \otimes P_+$, where $P_+$
acts only on the ancilla qubits, projecting them onto the state
$\ket{+}$. Each term in $V$ flips one ancilla qubit. To return to
$\mathcal{E}^{(0)}$, the powers of $V$ must either flip some ancilla
qubits and then flip them back, or they must flip all of them. The
latter process occurs at $k\th$ order and gives rise
to a term that mimics $H^{\mathrm{comp}}$. The former process occurs
at many orders, but at orders $k$ and lower gives rise only to terms
proportional to $P_0$.

As an example, let's examine $\ca$ up to second order for $k > 2$.
\begin{equation}
\ca^{(\leq 2)} = \lambda P_0 V P_0 + \lambda^2 P_0 V S^1 V P_0
\end{equation}
The term $P_0 V P_0$ is zero, because $V$ kicks the state out of
$\mathcal{E}^{(0)}$. By equation \ref{pert_recap} we see that applying
$V$ to a state in the ground space yields a state in the energy $k-1$
eigenspace. Substituting this denominator into $S^1$ yields
\begin{equation}
\ca^{(2)} = -\frac{\lambda^2}{k-1} P_0 V^2 P_0.
\end{equation} 
Because $V$ is a sum, $V^2$ consists of the squares of individual
terms of $V$ and cross terms. The cross terms flip two ancilla qubits,
and thus do not return the state to the ground space. The
squares of individual terms are proportional to the identity, thus
\begin{equation}
\ca^{(2)} = \lambda^2 \alpha_2 P_0
\end{equation}
for some $\lambda$-independent constant $\alpha_2$. Similarly, at any
order $m < k$, the only terms in $V^m$ which project back to
$\mathcal{E}^{(0)}$ are those arising from squares of individual
terms, which are proportional to the identity. Thus, up to order
$k-1$,
\begin{equation}
\ca^{(\leq k-1)} = \left( \sum_m \alpha_m \lambda^m \right)
P_0
\end{equation}
where the sum is over even $m$ between zero and $k-1$ and
$\alpha_0, \alpha_2,\ldots$ are the corresponding coefficients.

At $k\th$ order there arises another type of term. In $V^k$ there are
$k$-fold cross terms in which each of the terms in $V$ appears
once. For example, there is the term
\begin{equation}
\label{goodterm}
\lambda^k P_0 (\sigma_1 \otimes X_1)S^1 (\sigma_2
\otimes X_2) S^1 \ldots S^1 (\sigma_k \otimes
X_k) P_0
\end{equation}
The product of the energy denominators occurring in the $S^1$ operators
is
\begin{equation}
\prod_{j=1}^{k-1} \frac{1}{-j (k-j)} = \frac{(-1)^{k-1}}{((k-1)!)^2}.
\end{equation}
Thus, this term is
\begin{equation}
\frac{(-1)^{k-1} \lambda^k}{((k-1)!)^2} P_0 (\sigma_1 \otimes X_1) (\sigma_2
\otimes X_2) \ldots (\sigma_k \otimes X_k) P_0,
\end{equation}
which can be rewritten as
\begin{equation}
\frac{-(-\lambda)^k}{((k-1)!)^2} P_0 (\sigma_1 \sigma_2 \ldots \sigma_k
\otimes X^{\otimes k}) P_0.
\end{equation}
This term mimics $H^{\mathrm{comp}}$. The fact that all the $S$
operators in this term are $S^1$ is a general feature. Any term in
$\ca^{(k)}$ where $l_1 \ldots l_{k-1}$ are not all equal to 1
either vanishes or is proportional to $P_0$. This is because such 
terms contain $P_0$ operators separated by fewer than $k$ powers of
$V$, and thus the same arguments used for $m < k$ apply.

There are a total of $k!$ terms of the type shown in expression
\ref{goodterm}. 
Thus, up to $k\th$ order
\begin{equation}
\ca^{(\leq k)} = f(\lambda) P_0 + \frac{-k (-\lambda)^k}{(k-1)!} P_0
(\sigma_1 \sigma_2 \ldots \sigma_k \otimes X^{\otimes k}) P_0,
\end{equation}
which can be written as
\begin{equation}
\label{mpauli}
\ca^{(\leq k)} = f(\lambda) P_0 + \frac{-k (-\lambda)^k}{(k-1)!} P_0
(H^{\mathrm{comp}} \otimes X^{\otimes k}) P_0
\end{equation}
where $f(\lambda)$ is some polynomial in $\lambda$. Note
that, up to $k\th$ order, $\ca$ happens to be Hermitian.
The effective Hamiltonian is $\cu \ca \cu^\dag$, thus by equation
\ref{mpauli},
\begin{eqnarray}
H_{\mathrm{eff}}(H_+^{\mathrm{gad}},2^k) & = & \cu f(\lambda) P_0
                       \cu^\dag + \cu \left[ \frac{-k (-\lambda)^k}{(k-1)!} P_0
                       (H^{\mathrm{comp}} \otimes X^{\otimes k}) P_0 
                       +\mathcal{O}(\lambda^{k+1}) \right] \cu^\dag
                       \nonumber \\ 
                       & = & f(\lambda) \Pi + \cu 
                       \left[\frac{-k (-\lambda)^k}{(k-1)!} 
                       P_0 (H^{\mathrm{comp}} \otimes X^{\otimes k}) P_0
                       + \mathcal{O}(\lambda^{k+1}) \right] \cu^\dag
\end{eqnarray}
since $\cu P_0 \cu^\dag = \Pi$. Thus,
\begin{equation}
\widetilde{H}_{\mathrm{eff}}(H_+^{\mathrm{gad}},2^k,f(\lambda)) = \cu \left[
                       \frac{-k (-\lambda)^k}{(k-1)!}
                       P_0(H^{\mathrm{comp}} \otimes X^{\otimes k})P_0
                       +\mathcal{O}(\lambda^{k+1}) \right] \cu^\dag.
\end{equation}
To order $\lambda^k$, we can approximate $\cu$ as $P_0$ since
the higher order corrections to $\cu$ give rise to terms of order 
$\lambda^{k+1}$ and higher in the expression for
$\widetilde{H}_{\mathrm{eff}}(H_+^{\mathrm{gad}},2^k,f(\lambda))$. Thus,
\begin{equation}
\widetilde{H}_{\mathrm{eff}}(H_+^{\mathrm{gad}},2^k,f(\lambda)) =
                        \frac{-k (-\lambda)^k}{(k-1)!} P_0 
                        (H^{\mathrm{comp}} \otimes X^{\otimes k}) P_0 +
                        \mathcal{O}(\lambda^{k+1}).
\end{equation}
Using $P_0 = I \otimes P_+$ we rewrite this as
\begin{equation}
\widetilde{H}_{\mathrm{eff}}(H_+^{\mathrm{gad}},2^k,f(\lambda)) = 
                        \frac{-k (-\lambda)^k}{(k-1)!}
                        H^{\mathrm{comp}} \otimes P_+ +
                        \mathcal{O}(\lambda^{k+1}).
\end{equation}

Now let's return to the general case where $H^{\mathrm{comp}}$ is a
linear combination of $k$-local terms with arbitrary coefficients
$c_s$, as described in equation \ref{hcomp}. Now that we have
gadgets to obtain $k$-local effective interactions, it is tempting to
eliminate one $k$-local interaction at a time, by introducing
corresponding gadgets one by one. However, this approach does not lend
itself to simple analysis by degenerate perturbation theory. This is
because the different $k$-local terms in general act on
overlapping sets of qubits. Hence, we instead consider 
\begin{equation}
\label{vgad}
V^{\mathrm{gad}} = \sum_{s=1}^r \sqrt[k]{c_s} V_s
\end{equation}
as a single perturbation, and work out the effective Hamiltonian in
powers of this operator. The unperturbed part of the total gadget
Hamiltonian is thus
\begin{equation}
H^{\mathrm{anc}} = \sum_{s=1}^r H_s^{\mathrm{anc}},
\end{equation}
which has energy gap $\gamma = k-1$. The full Hamiltonian is
\begin{equation}
H^{\mathrm{gad}} = H^{\mathrm{anc}} + \lambda V^{\mathrm{gad}},
\end{equation}
so the perturbation series is guaranteed to converge under the
condition
\begin{equation}
\lambda < \frac{k-1}{4 \| V^{\mathrm{gad}} \|}
\end{equation}
As mentioned previously, we will work only within the simultaneous
$+1$ eigenspace of the $X^{\otimes k}$ operators acting on each of the
ancilla registers. In this subspace, $H^{\mathrm{anc}}$ has degeneracy
$2^n$ which gets split by the perturbation $\lambda V$ so that it
mimics the spectrum of $H^{\mathrm{comp}}$.

Each $V_s$ term couples to a different ancilla register. Hence, any
cross term between different $V_s$ terms flips some ancilla
qubits in one register and some ancilla qubits in another. Thus, at
$k\th$ order, non-identity cross terms between different $s$ cannot
flip all $k$ ancilla qubits in any given ancilla register, and they
are thus projected away by the $P_0$ operators appearing in the
formula for $\ca$. Hence the perturbative analysis proceeds just as it
did when there was only a single nonzero $c_s$, and one finds,
\begin{equation}
\widetilde{H}_{\mathrm{eff}}(H_+^{\mathrm{gad}},2^n,f(\lambda)) =
                        \frac{-k (-\lambda)^k}{(k-1)!}
                        P_0 \left( \sum_{s=1}^r c_s H_s 
                        \otimes X_s^{\otimes k} \right) P_0 +
                        \mathcal{O}(\lambda^{k+1}), 
\end{equation}
where $X_s^{\otimes k}$ is the operator $X^{\otimes k}$ acting
on the register of $k$ ancilla qubits corresponding to a given
$s$, and $f(\lambda)$ is some polynomial in $\lambda$ of degree at
most $k$. Note that coefficients in the polynomial $f(\lambda)$ depend
on $H^{\mathrm{comp}}$. As before, this can be rewritten as
\begin{equation}
\label{finalheff}
\widetilde{H}_{\mathrm{eff}}(H_+^{\mathrm{gad}},2^n,f(\lambda)) =  
                \frac{-k (-\lambda)^k}{(k-1)!} H^{\mathrm{comp}}
                \otimes P_+ + \mathcal{O}(\lambda^{k+1}),
\end{equation}
where $P_+$ acts only on the ancilla registers, projecting them
all into the $\ket{+}$ state. Hence, as asserted in section
\ref{intro}, the 2-local gadget Hamiltonian $H^{\mathrm{gad}}$
generates effective interactions which mimic the $k$-local Hamiltonian
$H^{\mathrm{comp}}$. We expect that this technique may find many
applications in quantum computation, such as in proving
QMA-completeness of Hamiltonian problems, and constructing physically
realistic Hamiltonians for adiabatic quantum computation. 

\section{Numerical Examples}

In this section we numerically examine the performance of
perturbative gadgets in some small examples. As shown in section
\ref{section_analysis}, the shifted effective Hamiltonian is that given in
equation \ref{finalheff}. We define
\begin{equation}
H^{\mathrm{id}} \equiv \frac{-k (-\lambda)^k}{(k-1)!}
H^{\mathrm{comp}} \otimes P_+.
\end{equation}
$\widetilde{H}_{\mathrm{eff}}$ consists of the ideal piece
$H^{\mathrm{id}}$, which is of order $\lambda^k$, plus an error term of
order $\lambda^{k+1}$ and higher. For sufficiently small $\lambda$ these error
terms are therefore small compared to the $H^{\mathrm{id}}$ term which
simulates $H^{\mathrm{comp}}$. Indeed, by a calculation very similar
to that which appears in section \ref{convergence}, one can easily
place an upper bound on the norm of the error terms. However, in
practice the actual size of the error terms may be smaller than this
bound. To examine the error magnitude in practice, we plot
$ \frac{\| H^{\mathrm{id}} - \widetilde{H}_{\mathrm{eff}}
  \|}{\|H^{\mathrm{id}}\|}$ in figure \ref{ratio} using direct 
numerical computation of $\widetilde{H}_{\mathrm{eff}}$ without
perturbation theory. $f(\lambda)$ was calculated analytically for
these examples. In all cases the ratio of $\| H^{\mathrm{id}} -
  \widetilde{H}_{\mathrm{eff}}\|$ to $\|H^{\mathrm{id}}\|$ scales
approximately linearly with $\lambda$, as one expects since the
error terms are of  order $\lambda^{k+1}$ and higher, whereas
$H^{\mathrm{id}}$ is of order $\lambda^k$.

\begin{figure}
\begin{center}
\includegraphics[width=0.65\textwidth]{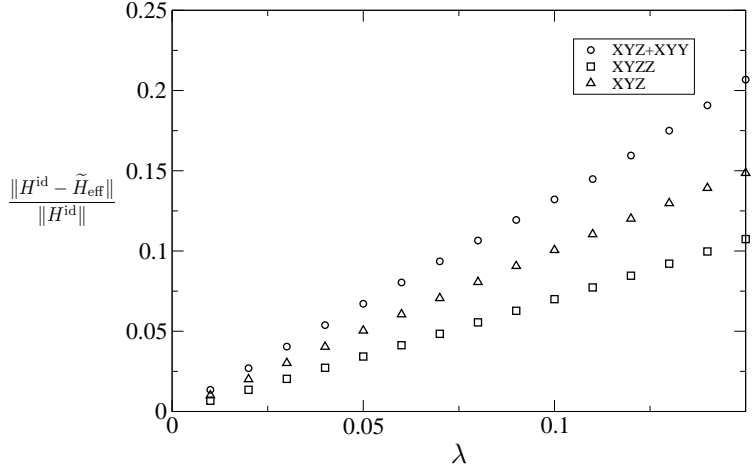}
\caption{\label{ratio} Here the ratio of the error terms to the ideal
  Hamiltonian $H^{\mathrm{id}} \equiv \frac{-k (-\lambda)^k}{(k-1)!}
  H^{\mathrm{comp}}$ is plotted. We examine three examples, a third
  order gadget simulating a single $XYZ$ interaction, a third order
  gadget simulating a pair of interactions $XYZ+XYY$, and a fourth
  order gadget simulating a fourth order interaction
  $XYZZ$. Here $\widetilde{H}_{\mathrm{eff}}$ is calculated by direct
  numerical computation without using perturbation theory. As expected
  the ratio of the norm of the error terms to $H^{\mathrm{id}}$ goes
  linearly to zero with shrinking $\lambda$.}
\end{center}
\end{figure}

\section{Derivation of Perturbative Formulas}
\label{Blochapp}

In this section we give a self-contained presentation of the
derivations for the method of degenerate perturbation theory used in
this paper. We closely follow Bloch\cite{Bloch}. Given a Hamiltonian
of the form
\begin{equation}
H = H^{(0)} + \lambda V
\end{equation}
we wish to find the effective Hamiltonian induced by the perturbation
$\lambda V$ on the ground space of $H^{(0)}$. In what follows, we
assume that the ground space of $H^{(0)}$ has energy zero. This
simplifies notation, and the generalization to nonzero ground energy
is straightforward. To further simplify notation we define
\begin{equation}
\hv = \lambda V.
\end{equation}

Suppose the ground space of $H^{(0)}$ is $d$-dimensional and denote it
by $\mathcal{E}^{(0)}$. Let
$\ket{\psi_1}, \ldots, \ket{\psi_d}$ be the perturbed eigenstates
arising from the splitting of this degenerate ground
space, and let $E_1, \ldots, E_d$ be their energies. Furthermore,
let $\ket{\alpha_j} = P_0 \ket{\psi_j}$ where $P_0$ is the projector
onto the unperturbed ground space of $H^{(0)}$. If $\lambda$ is
sufficiently small, $\ket{\alpha_1},\ldots,\ket{\alpha_d}$ are
linearly independent, and we can define an operator $\cu$ such that
\begin{equation}
\label{cudef}
\cu \ket{\alpha_j} = \ket{\psi_j}
\end{equation}
and
\begin{equation}
\label{uproj}
\cu \ket{\phi} = 0 \quad \forall \ket{\phi} \in \mathcal{E}^{(0)\perp}.
\end{equation}

Now let $\ca$ be the operator
\begin{equation}
\label{pvcu}
\ca = P_0 \hv \cu.
\end{equation}
$\ca$ has $\ket{\alpha_1}, \ldots, \ket{\alpha_d}$ as its eigenstates,
and $E_1, \ldots, E_d$ as its corresponding energies. To see this,
note that since $H^{(0)}$ has zero ground state energy
\begin{equation}
\label{zerotrick}
P_0 \hv = P_0 (H^{(0)}+\hv) = P_0 H.
\end{equation}
Thus,
\begin{eqnarray}
\ca \ket{\alpha_j} & = & P_0 \hv \cu \ket{\alpha_j} \nonumber \\
& = & P_0 \hv \ket{\psi_j} \nonumber \\
& = & P_0 H \ket{\psi_j} \nonumber \\
& = & P_0 E_j \ket{\psi_j} \nonumber \\
& = & E_j \ket{\alpha_j}.
\end{eqnarray}
The essential task in this formulation of
degenerate perturbation theory is to find a perturbative expansion for
$\cu$. From $\cu$ one can obtain $\ca$ by equation \ref{pvcu}. By
diagonalizing $\ca$ one obtains $E_1,\ldots, E_d$, and
$\ket{\alpha_1},\ldots,\ket{\alpha_d}$. Then, by applying $\cu$ to
$\ket{\alpha_j}$ one obtains $\ket{\psi_j}$. So, given a
perturbative formula for $\cu$, all quantities of interest can be
calculated. Rather than diagonalizing $\ca$ to obtain individual
eigenstates and eigenenergies, one can instead compute an effective
Hamiltonian for the entire perturbed eigenspace, defined by
\begin{equation}
H_{\mathrm{eff}}(H,d) \equiv \sum_{j=1}^d E_j \ket{\psi_j} \bra{\psi_j}.
\end{equation}
This is given by
\begin{equation}
H_{\mathrm{eff}}(H,d) = \cu \ca \cu^\dag.
\end{equation}

To derive a perturbative formula for $\cu$, we start with
Schr\"odinger's equation:
\begin{equation}
\label{schro}
H \ket{\psi_j} = E_j \ket{\psi_j}.
\end{equation}
By equation \ref{zerotrick}, left-multiplying this by $P_0$ yields
\begin{equation}
\label{pnaught}
P_0 \hv\ket{\psi_j} = E_j \ket{\alpha_j}.
\end{equation}
By equation \ref{uproj},
\begin{equation}
\label{icup}
\cu P_0 = \cu.
\end{equation}
Thus left-multiplying equation \ref{pnaught} by $\cu$ yields
\begin{equation}
\label{cuv}
\cu \hv \ket{\psi_j} = E_j \ket{\psi_j}.
\end{equation}
By subtracting \ref{cuv} from \ref{schro} we obtain
\begin{equation}
(H - \cu \hv) \ket{\psi_j} = 0.
\end{equation}
The span of $\ket{\psi_j}$ we call $\mathcal{E}$. For any state
$\ket{\beta}$ in $\mathcal{E}$ we have
\begin{equation}
(H - \cu \hv) \ket{\beta} = 0.
\end{equation}
Since $\cu \ket{\gamma} \in \mathcal{E}$ for any state $\ket{\gamma}$,
it follows that
\begin{equation}
(H - \cu \hv) \cu = 0.
\end{equation}
This equation can be rewritten as
\begin{equation}
\label{ho}
H^{(0)} \cu = -\hv \cu + \cu \hv \cu.
\end{equation}
Defining $Q_0 = \id - P_0$ we have
\begin{equation}
\label{icuq}
\cu = P_0 \cu + Q_0 \cu.
\end{equation}
Substituting this into the left side of \ref{ho} yields
\begin{equation}
H^{(0)} Q_0 \cu = - \hv \cu + \cu \hv \cu,
\end{equation}
because $H^{(0)} P_0 = 0$. In $\mathcal{E}^{(0)\perp}$, $H^{(0)}$ has a well
defined inverse and one can write
\begin{equation}
Q_0 \cu = -\frac{1}{H^{(0)}} Q_0 (\hv \cu - \cu \hv \cu). 
\end{equation}
Using equation \ref{icuq}, one obtains
\begin{equation}
\cu = P_0 \cu - \frac{1}{H^{(0)}} Q_0 ( \hv \cu - \cu \hv \cu ).
\end{equation}
By the definition of $\cu$ it is apparent that $P_0 \cu = P_0$, thus
this equation simplifies to
\begin{equation}
\label{imp}
\cu = P_0 - \frac{1}{H^{(0)}} Q_0 ( \hv \cu - \cu \hv \cu ).
\end{equation}

We now expand $\cu$ in powers of $\lambda$ (equivalently, in powers
of $\hv$), and denote the $m\th$ order term by
$\cu^{(m)}$. Substituting this expansion into equation 
\ref{imp} and equating terms at each order yields the following
recurrence relations.
\begin{eqnarray}
\label{rec1}
\cu^{(0)} & = & P_0 \\
\label{rec2}
\cu^{(m)} & = & -\frac{1}{H^{(0)}} Q_0 \left[ \hv \cu^{(m-1)} -
  \sum_{p=1}^{m-1} \cu^{(p)} \hv \cu^{(m-p-1)} \right] \quad (m=1,2,3\ldots)
\end{eqnarray}
Note that the sum over $p$ starts at $p=1$, not $p=0$. This is because
\begin{equation}
\frac{1}{H^{(0)}} Q_0 \cu^{(0)} = \frac{1}{H^{(0)}} Q_0 P_0 = 0.
\end{equation}
Let 
\begin{equation}
S^l = \left\{ \begin{array}{ll}
\frac{1}{\left(-H^{(0)}\right)^l} Q_0 & \textrm{if $l > 0$} \\
- P_0 & \textrm{if $l = 0$}
\end{array} \right. .
\end{equation}
$\cu^{(m)}$ is of the form
\begin{equation}
\cu^{(m)} = {\sum}' S^{l_1} \hv S^{l_2} \hv \ldots S^{l_m} \hv P_0,
\end{equation}
where $\sum'$ is a sum over some subset of $m$-tuples $(l_1, l_2,
\ldots, l_m)$ such that
\begin{equation}
l_i \geq 0 \quad (i=1,2,\ldots, m)
\end{equation}
\begin{equation}
l_1 + l_2 + \ldots + l_m = m.
\end{equation}
The proof is an easy induction. $\cu^{(0)}$ clearly satisfies this,
and we can see that if $\cu^{(j)}$ has these properties for all
$j < m$, then by recurrence \ref{rec2}, $\cu^{(m)}$ also has
these properties. 

All that remains is to prove that the subset of allowed $m$-tuples
appearing in the sum $\sum'$ are exactly those which satisfy
\begin{equation}
\label{convex}
l_1 + \ldots + l_p \geq p \quad (p=1,2,\ldots,m-1).
\end{equation}

Following \cite{Bloch}, we do this by introducing stairstep
diagrams to represent the $m$-tuples, as shown in figure \ref{stairs}.
\begin{figure}
\begin{center}
\includegraphics[width=0.22\textwidth]{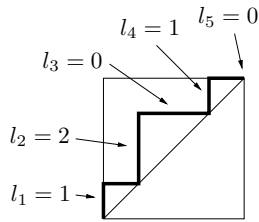}
\caption{\label{stairs} From a given $m$-tuple $(l_1,l_2,\ldots, l_m)$
  we construct a corresponding stairstep diagram by making the $j\th$
  step have height $l_j$, as illustrated above.}
\end{center}
\end{figure}
The $m$-tuples with property \ref{convex} correspond to diagrams in
which the steps lie above the diagonal. Following \cite{Bloch} we call
these convex diagrams. Thus our task is to prove that the sum $\sum'$
is over all and only the convex diagrams. To do this, we consider
the ways in which convex diagrams of order $m$ can be constructed from
convex diagrams of lower order. We then relate this to the way
$\cu^{(m)}$ is obtained from lower order terms in the recurrence
\ref{rec2}.

In any convex diagram, $l_1 \geq 1$. We now consider the two cases
$l_1=1$ and $l_1 > 1$. In the case that $l_1 = 1$, the diagram is as
shown on the left in figure \ref{cases}.
\begin{figure}
\begin{center}
\includegraphics[width=0.5\textwidth]{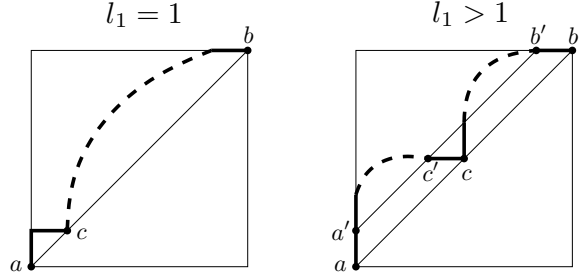}
\caption{\label{cases} A convex diagram must have either $l_1 = 1$ or
  $l_1 > 1$. In either case, the diagram can be decomposed as a
  concatenation of lower order convex diagrams.}
\end{center}
\end{figure}
In any convex diagram of order $m$ with $l_1 = 1$, there is an
intersection with the diagonal after one step, at the point
that we have labelled $c$. The diagram from $c$ to $b$ is a convex
diagram of order $m-1$. Conversely, given any convex diagram of order
$m-1$ we can construct a convex diagram of order $m$ by adding one
step to the beginning. Thus the convex diagrams of order $m$ with $l_1
= 1$ correspond bijectively to the convex diagrams of order $m-1$.

The case $l_1 > 1$ is shown in figure \ref{cases} on the right. Here
we introduce the line from $a'$ to $b'$, which is parallel to the
diagonal, but higher by one step. Since the diagram must end at $b$,
it must cross back under $a'b'$ at some point. We'll label the first
point at which it does so as $c'$. In general, $c'$ can equal
$b'$. The curve going from $a'$ to $c'$ is a convex diagram of order
$p$ with $1 \leq p \leq m-1$, and the curve going from $c$ to $b$ is a
convex diagram of order $n-p-1$ (which may be order zero if $c'$ =
$b'$). Since $c'$ exists and is unique, this establishes a bijection
between the convex diagrams of order $m$ with $l_1 > 1$, and the set
of the pairs of convex diagrams of orders $p$ and $n-p-1$, for $1 \leq
p \leq n-1$.

Examining the recurrence \ref{rec2}, we see that the $l_1 = 1$
diagrams are exactly those which arise from the term
\begin{equation}
\frac{Q_0}{H^{(0)}} \hv \cu^{(m-1)}
\end{equation}
and the $l_1 > 1$ diagrams are exactly those which arise from the term
\begin{equation}
\frac{Q_0}{H^{(0)}} \sum_{p=1}^{m-1} \cu^{(p)} \hv \cu^{(n-p-1)}.
\end{equation}
which completes the proof that $\sum'$ is over the $m$-tuples
satisfying equation \ref{convex}.

\section{Convergence of Perturbation Series}
\label{convergence}

Here we show that the perturbative expansion for $\cu$
given in equation \ref{cu} converges for
\begin{equation}
\| \lambda V \| < \frac{\gamma}{4}.
\end{equation}
By equation \ref{afromu}, the convergence of $\cu$ also implies the
convergence of $\ca$. Applying the triangle inequality to equation
\ref{cu} yields
\begin{equation}
\label{tri1}
\| \cu \| \leq 1 + \sum_{m=1}^{\infty} \| \cu^{(m)} \|.
\end{equation}
Substituting in equation \ref{cuterms} and applying the triangle
inequality again yields
\begin{equation}
\| \cu \| \leq 1 + \sum_{m=1}^{\infty} \lambda^m \sum_{(m)} \| S^{l_1}
\ldots V S^{l_m} V P_0 \|.
\end{equation}
By the submultiplicative property of the operator norm,
\begin{equation}
\label{intermed}
\| \cu \| \leq 1 + \sum_{m=1}^{\infty} \lambda^m \sum_{(m)} \| S^{l_1}
\| \cdot \| V \| \ldots \| V \| \cdot \| S^{l_m} \| \cdot \| V \|
\cdot \| P_0 \|.
\end{equation}
$\| P_0 \| = 1$, and by equation \ref{S} we have
\begin{equation}
\| S^l \| = \frac{1}{(E_1^{(0)})^l} = \frac{1}{\gamma^l}.
\end{equation}
Since the sum in equation \ref{intermed} is over $l_1 + \ldots +
l_m = m$, we have
\begin{equation}
\| \cu \| \leq 1 + \sum_{m=1}^\infty \sum_{(m)} \frac{\| \lambda V
  \|^m}{\gamma^m}. 
\end{equation}
The sum $\sum_{(m)}$ is over a subset of the $m$-tuples
adding up to $m$. Thus, the number of terms in this sum is less than
the number of ways of obtaining $m$ as a sum of $m$ nonnegative
integers. By elementary combinatorics, the number of ways to obtain
$n$ as a sum of $r$ nonnegative integers is $\binom{n+r-1}{n}$, thus
\begin{equation}
\label{notgeom}
\| \cu \| \leq 1 + \sum_{m=1}^\infty \binom{2m-1}{m} \frac{\| \lambda
  V \|^m}{\gamma^m}.
\end{equation}
Since
\begin{equation}
\sum_{j=0}^{2m-3} \binom{2m-1}{j} = 2^{2m-1},
\end{equation}
we have
\begin{equation}
\binom{2m-1}{m} \leq 2^{2m-1}.
\end{equation}
Substituting this into equation \ref{notgeom} converts it into a
convenient geometric series:
\begin{equation}
\| \cu \| \leq 1 + \sum_{m=1}^\infty 2^{2m-1} \frac{\| \lambda V
  \|^m}{\gamma^m}.
\end{equation}
This series converges for
\begin{equation}
\frac{4 \| \lambda V \|}{\gamma} < 1.
\end{equation}

\chapter{Multiplicity and Unity}

\begin{minipage}[c]{\textwidth}
\noindent
\emph{The fox knows many things but the hedgehog knows one big
  thing.}\\
-Archilochus \\
\end{minipage}

In his essay ``The Hedgehog and the Fox,'' Isaiah Berlin, echoing
Archilochus, classified thinkers into two categories, hedgehogs, who
look for unity and generality, and foxes, who revel in the variety and
complexity of the universe. In this chapter I will revisit the content
of my thesis from each of these points of view.

\section{Multiplicity}

\begin{minipage}[c]{\textwidth}
\noindent
\emph{...you shall not add to the misery and sorrow of the world, but
  shall smile to the infinite variety and mystery of it.} \\
-Sir William Saroyan \\
\end{minipage}

In this thesis I have spoken about many models of computation,
especially the adiabatic, topological, and circuit models. Each of
these models of computation can solve exactly the same set of problems
in polynomial time. Thus one might ask: ``why bother''? Why not just
stick to the quantum circuit model? 

As the many examples in this thesis have shown, sticking to only one
model of quantum computation would be a
mistake. The most obvious justification for considering alternative
models of quantum computation is that they provide promising
alternatives for the physical implementation of quantum computers. As
discussed in chapters \ref{introchap} and \ref{FT_AQC}, the main
barrier to practical quantum computation is the effect of noise. Three
types of quantum computer show promise for overcoming this barrier:
quantum circuits, by means of active error correction, adiabatic
quantum computers, by their inherent indifference to local properties,
and topological quantum computers by their energy gap, as discussed in
chapter \ref{FT_AQC}. It is not yet clear which of these strategies
will prove most useful, and this provides justification for
investigating all of them.

A skeptic might protest that this only justifies
investigation of physically realistic models of quantum
computation. Some of the models discussed in the literature, and in
this thesis, are not very physically realistic. For example, it seems
unlikely that adiabatic computation with 4-local Hamiltonians, or
quantum walks on exponentially many nodes will be implemented in
laboratories. However, even models slightly removed from physical
practicality have proven useful in the development of practical
physical implementations. For example, we now know that adiabatic
quantum computation with 2-local Hamiltonians is universal. This was
originally proven by showing that 5-local Hamiltonians are
computationally universal, and then making a series of reductions from
3-local down to 2-local. Even now, the best known \emph{direct}
universality proof is for 3-local adiabiatic quantum
computation\cite{Daniels_thesis}. 

There is a second, less obvious justification for considering multiple
models of quantum computation. Although all of quantum algorithms
could \emph{in principle} be formulated using the quantum circuit
model, this is not how quantum computation actually developed. The
quantum algorithms for factoring and searching were discovered using
the quantum circuit model. The quantum speedups for NAND tree
evaluation and simulated annealing were discovered using quantum
walks. The quantum algorithm for approximating Jones polynomials was discovered
using topological quantum computation. History has shown us that
a particular model of quantum computation gives us a perspective
from which certain quantum algorithms are easily visible, while they
remain obscure from the perspective of other models.

Having argued the virtues of having a multiplicity of models of
quantum computation, the time has come to put this principle into
action. That is, I will propose a direction for further research based
on the premise that formulating additional models of quantum
computation is useful even if the models are not directly practical.

One striking thing about the topological model of quantum computation
is its indifference to the details of the manipulations of the
particles. Only the topology of the braiding matters, and the specific
geometry is irrelevant. Let's now push this a step further and
consider a model where even the topology is irrelevant, and the only
thing that matters is how the particles were permuted. Just as the
braiding of anyons induces a representation of the braid group, we
analogously expect that the permutation of the particles induces a
representation of the symmetric group.

Such a model is not entirely without physical motivation. The exchange
statistics of Bosons and Fermions are exactly the two one-dimensional
representations of the symmetric group. Particles with exchange
statistics given by a higher dimensional representation of the
symmetric group have been proposed in the past. This is called
parastatistics. Such particles have never been observed. However,
based on the presently understood physics, they remain an exotic, but
not inconcievable possibility\cite{Peres}. Furthermore, many
Hamiltonians in nature are symmetric under permutation of the
particles. If a Hamiltonian has symmetry group $G$, then its
degenerate eigenspaces will transform as representations of $G$, and
these representations will generically be
irreducible\cite{Hamermesh}.

More precisely, suppose a Hamiltonian $H$ on $n$ particles has a
$d$-fold degenerate eigenspace spanned by
$\phi_1(x_1,\ldots,x_n),\ldots,\phi_d(x_1,\ldots,x_n)$. We start with a
given state within that space
\[
\psi(x_1,\ldots,x_n) = \sum_{j=1}^d \alpha_j \phi_j(x_1,\ldots,x_n)
\]
Then, if we permute the particles according to some permutation $\pi$
we obtain some other state $\psi(x_{\pi(1)},\ldots,x_{\pi(n)})$.
Because the Hamiltonian is permutation symmetric, this state will lie
within the same eigenspace. That is, there are some coefficients
$\beta_1,\ldots,\beta_d$ such that
\[
\psi(x_{\pi(1)},\ldots,x_{\pi(n)}) = \sum_{j=1}^d \beta_j \phi_j(x_1,\ldots,x_n).
\]
It is clear that the dependence of $\beta_1,\ldots,\beta_d$ on
$\alpha_1,\ldots,\alpha_d$ is linear. Thus there is some matrix
$M^{\pi}$ corresponding to permutation $\pi$:
\[
\beta_i = \sum_{j=1}^d M^{\pi}_{ij} \alpha_j.
\]
It is not hard to see that the mapping from permutations to their
corresponding matrices is a group homomorphism. That is, these
$d \times d$ matrices form a representation of the symmetric group
$S_n$. If the $H$ has no other symmetries, then this representation
will generically be irreducible\cite{Hamermesh}.

As advocated above, I will not worry too much about the physical
justification for the model, but instead consider just two
questions. First, does the model lend itself to the the rapid solution
of any interesting computational problems, and second, can it be
efficiently simulated by quantum circuits? If both answers are yes,
then we obtain new quantum algorithms.

The question of what problem this model of computation lends itself to
has an obvious answer: the computation of representations of
the symmetric group. There are many good reasons to restrict our
consideration to only the irreducible representations. Any finite
group has only finitely many irreducible representations but
infinitely many reducible representations. These reducible
representations are always the direct sum of multiple irreducible
representations. Thus, by performing a computation with a reducible
representation we would merely be performing a superposition of
computations with irreducible representations.  

In chapter \ref{Jones} we saw that if we can apply a representation of
the braid group then, using the Hadamard test, we can estimate the
matrix elements of this representation to polynomial
precision. Furthermore, by sampling over the diagonal matrix elements
we can estimate the normalized trace of the representation to
polynomial precision. If we instead have a representation of the
symmetric group, the situation is precisely analogous. The trace of a
group representation is called its character. Characters of group
representations have many uses not only in mathematics, but also in
physics and chemistry. Note that, unlike the matrix elements of a
representation, its character is basis independent.

Just as with anyonic quantum computation, we can imagine the particles
initially positioned along a line. We then allow ourselves to swap
neighbors. The runtime of an algorithm is the number of necessary
swaps. Interestingly, in the parastistical model of computation, no
algorithm has runtime more than $O(n^2)$, because there are only
finitely many permutations and each of them can be constructed using at
most $O(n^2)$ swaps.

To formulate the concrete computational problems we'll need to delve a
little bit into the specifics of the irreducible representations of
the symmetric group. The irreducible representations of $S_n$ are
indexed by the Young diagrams of $n$ boxes. These are all the possible
partitions of the $n$ boxes into rows, where the rows are arranged in
descending order of length. The example $n=4$ is shown in figure
\ref{young_diagrams}. The matrix elements of these representations
depend on a choice of basis. For our purposes it is essential that the
basis be chosen so that the representation is unitary. The most widely
used such basis is called the Young-Yamanouchi
basis\cite{Hamermesh}. In this basis the irreducible representations
are sparse orthogonal matricees. For the irreducible representation
corresponding to a Young diagram $\lambda$, the Young-Yamanouchi basis
vectors correspond to the set of standard Young tableaux compatible
with lambda. These are all the numberings of boxes so that if we added
the boxes in this order, the configuration would be a valid Young
diagram after every step. This is illustrated in figure
\ref{Young_tableau}. It is not hard to see that for some $\lambda$,
the number of standard tableaux, and hence the dimension of the
representation, is exponential in $n$.

\begin{figure}
 \begin{center}
\includegraphics[width=0.5\textwidth]{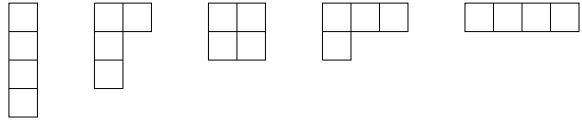}
\caption{\label{young_diagrams} The Young diagrams with four
  boxes. They correspond to the irreducible representations of $S_4$.}
\end{center}
\end{figure}

\begin{figure}
 \begin{center}
\includegraphics[width=0.6\textwidth]{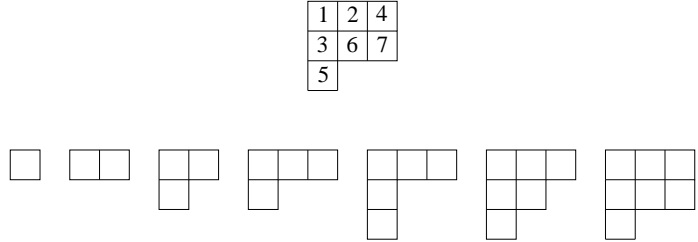}
\caption{\label{Young_tableau} Above we
  show an example Young tableau, and beneath it the corresponding
  sequence of Young diagrams from left to right.}
\end{center}
\end{figure}

Thus we can state the following computational problems regarding
$S_n$, which are solvable in $\mathrm{poly}(n)$ time in the
parastatistical model of computation. \\ \\
\noindent
\begin{minipage}[c]{\textwidth}
\textbf{Problem:} Calculate an irreducible representation for the symmetric
group $S_n$. \\
\textbf{Input:} A Young diagram specifying the irreducible
representation, a permutation from $S_n$, a pair of standard Young
tableaux indicating the desired matrix element, and a polynomially small
parameter $\epsilon$.\\
\textbf{Output:} The specified matrix element to within $\pm
\epsilon$. \\ \\
\end{minipage}
\noindent
\begin{minipage}[c]{\textwidth}
\textbf{Problem:} Calculate a character for the symmetric group $S_n$.\\
\textbf{Input:} A Young diagram $\lambda$ specifying the irreducible
representation, a permutation $\pi$ from $S_n$, and a polynomially
small parameter $\epsilon$.\\
\textbf{Output:} Let $\chi_\lambda(\pi)$ be the character, and let
$d_\lambda$ be the dimension of the irreducible representation. The
output is $\chi_\lambda(\pi)/d_\lambda$ to within $\pm \epsilon$.\\ \\
\end{minipage}
Next, we must find out how hard these problems are classically. If
classical polynomial time algorithms for these problems are known then
we have not discovered anything interesting. Without looking into the
classical computation literature there are already two things we can
say. First, as noted above, for some Young diagrams of $S_n$, the
corresponding representation has dimension exponential in $n$. Thus
the above problems cannot be solved classically in polynomial
time by directly using matrix multiplication. Second, the irreducible
representations have both positive and negative matrix elements in all
of the bases discussed in standard references\cite{Hamermesh,
  Boerner}. Thus interference effects are important, so naive Markov
chain methods will not work.

To go beyond these simple observations, we must consult the relevant
experts and body of literature. As shown in \cite{Hepler}, the problem
of exactly evaluating the characters of the irreducible
representations of the symmetric group is \#P-complete. This rules out
the possibility that some ingenious closed form expression for the
characters is buried in the math literature somewhere. The characters
are obtained in the parastatistical model of computation to only
polynomial precision. This corresponds to the precision one could
obtain by classically sampling from the diagonal matrix elements of
the representation. Thus, the most likely scenario by which these
algorithms could fail to be interesting is if the individual matrix
element of the irreducible representations of the symmetric group are
easy to compute classically, and the only reason that computing the
characters is hard is that there are exponentially many of these
matrix elements to add up. However, the sources I have consulted do
not provide any way to obtain matrix elements of the irreducible
representations of the $S_n$ in $\mathrm{poly}(n)$ time for arbitrary
permutations\cite{Moore_personal, Hamermesh, Boerner}. Furthermore,
there is a body of work on how to improve the efficiency of
exponential time algorithms for this problem\cite{Wu_Zhang1,
Wu_Zhang2, Egecioglu, Clifton}. Unless this entire body of work is
misguided, no polynomial-time methods for computing such matrix
elements were known when these papers were written.

A completeness result would provide even stronger evidence of the
difficulty of these problems. The problems of estimating Jones
polynomials discussed in section \ref{Jones} were each BQP or DQC1
complete. In general one could conjecture that for any representation
dense in a unitary group of exponentially large dimension, the problem
of estimating matrix elements to polynomial precision is BQP-complete
and the problem of estimationg normalized characters to polynomial
precision is DQC1-complete. However, because the symmetric group is
finite, no representation of it can be dense in a continuous
group. Thus it seems unlikely that the problems of estimating the
matrix elements and characters of the symmetric group are BQP-complete
or DQC1-complete. Furthermore, the fact that no parastatistical
algorithm on $n$ particles requires more than $O(n^2)$ computational
steps makes it seem unlikely\footnote{Probably one could prove a
  precise no-go theorem along these lines using the Heirarchy theorem
  for BQP. (See chapter \ref{introchap}.)} that this model is
universal.

Next we must see whether the parastatistical model of computation can
be simulated in polynomial time by standard quantum computers. If so
we obtain two new polynomial time quantum algorithms apparently
providing exponential speedup over known classical
algorithms. Normally the search for quantum algorithms is a pursuit
fraught with frustration. However, in this case we have a win-win
situation. If the parastatistical model cannot be simulated by 
quantum computers in polynomial time, then instead of quantum
algorithms we have a physically plausible model of computation not
contained in BQP, which is also very exciting.

A detailed examination of the Young-Yamanouchi matrices in
\cite{Hamermesh} makes me fairly convinced that the parastatistical
model of computation can be simulated in polynomial time by quantum
circuits. Specifically, it appears that this can be done very
analogously to the implementation of the Fibonacci representation of
the braid group in chapter \ref{Jones}. However, this is not the place
to discuss such details. Instead I will now switch teams, and take the
side of the hedgehog.

\section{Unity}

\begin{minipage}[c]{\textwidth}
\noindent
\emph{\ldots the world will somehow come clearer and we will grasp the
  true strangeness of the universe. And the strangeness will all prove
  to be connected and make sense.} \\
-E.O. Wilson \\
\end{minipage}

Upon examining the catalogue of quantum algorithms in chapter
\ref{introchap}, a striking pattern emerges. Although the quantum
algorithms are superficially widely varied, the exponential speedups
for non-oracular problems generally fall into two broad families:
the speedups obtained by reduction to hidden subgroup problems, and the
speedups related to knot invariants. Interestingly, both of these
families of speedups rely on representation theory. The speedups for
the hidden subgroup related problems are based on the Fourier
transform over groups. Such a Fourier transform goes between the
computational basis, and a basis made from matrix elements of
irreducible representations. The speedups for the evaluation of knot
invariants are based on implementating a unitary representation of the
braid group using quantum circuits. It is therefore tempting to look
for some grand unification to unite all exponential speedups into one
representation-theoretic framework.

I do not know whether such a grand unification is possible, nevermind
how to carry it out. However, I will offer a bold speculation as to a
possible route forward. Rather than starting with the Fourier
transform, lets first consider the Schur transform. Like the Fourier
transform, the Schur transform can be efficiently implented using
quantum circuits \cite{Aramsthesis}, and has applications in quantum
information processing. Furthermore, it has a more obvious connection
to multiparticle physics than does the Fourier transform. For example,
suppose we have two spin-1/2 particles. The Hilbert space of quantum
states for these spins is four dimensional. One basis we can choose
for this Hilbert space is obtained by taking the tensor product of
$\sigma_z$-basis for each spin:
\[
\begin{array}{l}
\ket{\uparrow} \ket{\uparrow} \\
\ket{\uparrow} \ket{\downarrow} \\
\ket{\downarrow} \ket{\uparrow} \\
\ket{\downarrow} \ket{\downarrow}
\end{array}
\]
Another basis can be obtained as the simultaneous eigenbasis of the
total angular momentum $\sigma_x^{(1)} \sigma_x^{(2)} + \sigma_y^{(1)}
\sigma_y^{(2)} + \sigma_z^{(1)} \sigma_z^{(2)}$ and the total
azimuthal angular momentum $\sigma_z^{(1)} + \sigma_z^{(2)}$. This
basis is
\[
\begin{array}{l}
\ket{\uparrow} \ket{\uparrow} \\
\frac{1}{\sqrt{2}} \left( \ket{\uparrow} \ket{\downarrow} + \ket{\downarrow}
\ket{\uparrow} \right) \\
\ket{\downarrow} \ket{\downarrow} \\
\frac{1}{\sqrt{2}} \left( \ket{\uparrow} \ket{\downarrow} + \ket{\downarrow}
\ket{\uparrow} \right)
\end{array}
\]
The Schur transform in this case is just the unitary change of basis
between these two bases. The general case is explained in
\cite{Aramsthesis}, and is fairly analogous.

The matrix elements appearing in this change of basis are known as
Clebsch-Gordon coefficients, and are tabulated in most undergraduate
quantum mechanics textbooks. According to \cite{Nayak}, the
coefficients appearing in the fusion rules of a topological quantum
field theory are essentially a generalization of Clebsch-Gordon
coefficients. Thus, I offer the conjecture that one could implement the
Schur transform directly using the fusion rules of some TQFT. (Since
some TQFTs are universal for quantum computation, and Schur transforms
can be efficiently computed, one could always take the circuit for
finding Schur transforms and convert it into some extremely
complicated braiding of anyons. This is not what we are looking for.)

The Schur transform is related to the Fourier transform over the
symmetric group\cite{Aramsthesis}. Thus, if a direct TQFT implementation of
the Schur transform were found then one could next look for
a direct anyonic implementations of Fourier transforms. I therefore
propose the conjecture that the two classes of quantum algorithms
correspond to the two components of a topological quantum field
theory: the knot invariant algorithms correspond to the braiding
rules, and the hidden subgroup problems correspond to the fusion
rules.

\appendix
\chapter{Classical Circuit Universality}
\label{classical_universality}

Consider an arbitrary function $f$ from $n$ bits to one bit. $f$ can
be specified by listing each bitstrings $x \in \{0,1\}^n$ such that
$f(x) = 1$. To show the universality of the 
$\{ \mathrm{AND, NOT, OR, FANOUT} \}$ gate set, we wish to construct a 
circuit out of these gates that implements $f$. We'll diagrammatically
represent these gates using
\begin{center}
\includegraphics[width=0.6\textwidth]{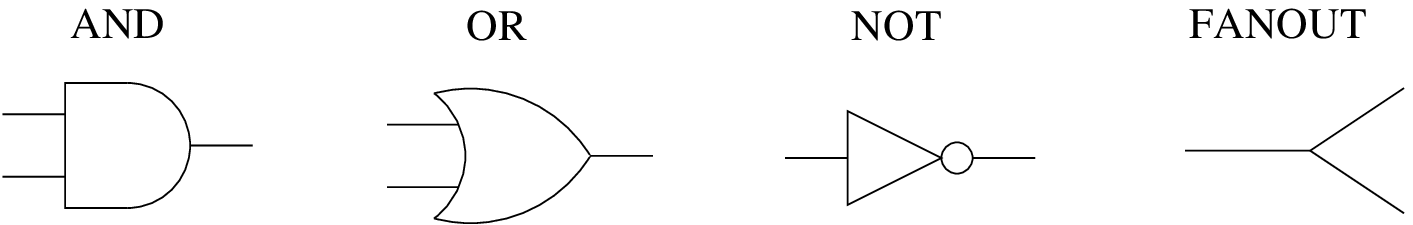}
\end{center}
From two-input AND gates we can construct an $m$-input AND gate for
any $m$. The example $m=4$ is shown below.
\begin{center}
\includegraphics[width=0.25\textwidth]{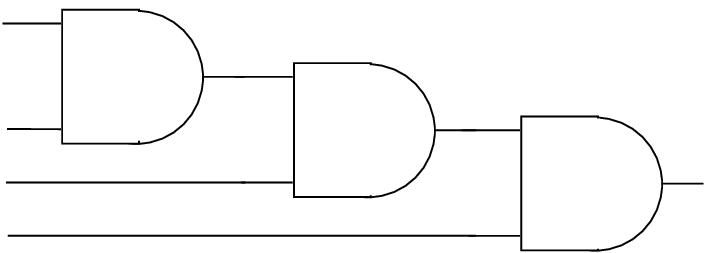}
\end{center}
An $m$-input OR gate can be constructed from 2-input OR gates, and an
$m$-output FANOUT gate can be constructed from 2-output FANOUT gates
similarly. The $m$-input AND gate accepts only the string
$111\ldots$. To accept a different string, one can simply attach NOT
gates to all the inputs which should be 0. Using AND and NOT gates,
one can thus make a circuit to accept each bitstring accepted by
$f$. These can then be joined together using a multi-input OR
circuit. FANOUT gates are used to supply the inputs to each of
these. The resulting circuit simulates $f$ as shown in figure
\ref{universal_circuit}.

\begin{figure}
\begin{center}
\includegraphics[width=0.9\textwidth]{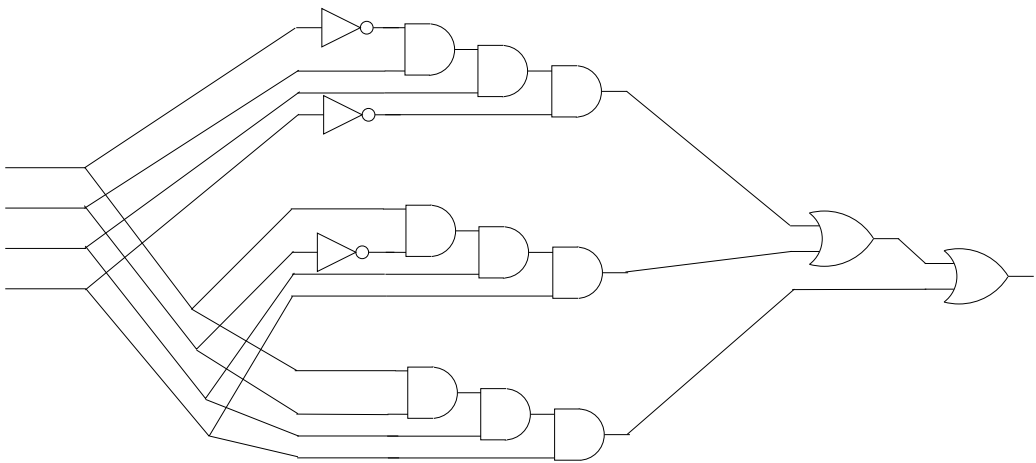}
\caption{\label{universal_circuit} The Boolean circuit with four bits
  of input and one bit of output which accepts only the inputs $0110$,
  $1011$, and $1111$ is implemented by the shown circuit. Any Boolean
  function can be implemented similarly, as described in the text.}
\end{center}
\end{figure}

To implement a function with multiple bits of output, one can consider
each bit of output to be a separate single-bit Boolean function of the
inputs, and construct a circuit for each bit of output accordingly.

\chapter{Optical Computing}
\label{optical}

Quantum mechanics bears a close resemblance to classical optics. An
optical wave associates an amplitude with each point in space. A
quantum wavefunction associates an amplitude with each possible
configuration of the system. Whereas an optical wave is a function of
at most three spatial variables, a quantum wavefunction of a system of
particles is a function of all the parameters needed to describe the
configuration of the particles. Thus classical optics lacks the
exponentially high-dimensional state space of quantum mechanics. The
similarity between classical optics and quantum mechanics is of course
no coincidence, as photons are governed by quantum mechanics. Essentially,
classical optics treats the case where these photons are independent
and unentangled, so that the intensity of light on at a given point on
the detector is simply proportional to the probability density for a
single photon to be detected at that point if it were sent through
the optical apparatus by itself.

\begin{figure}
\begin{center}
\includegraphics[width=0.45\textwidth]{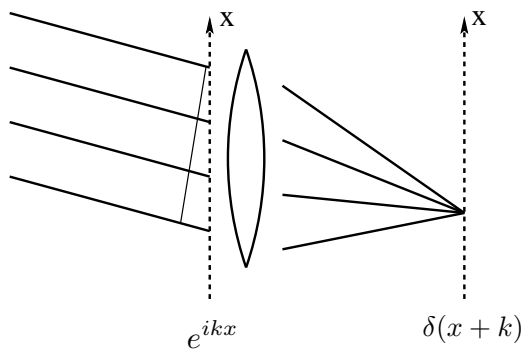}
\caption{\label{lens} A collimated beam shined at an angle toward a
  lens will have phases varying linearly across the face of the lens,
  due to the extra distance that some rays must travel. Such a beam
  gets focused to a point on the focal plane whose location depends on
  the angle at which the beam was shined onto the lens. Thus, the
  amplitude across the lens of $e^{ikx}$ gets transformed to an
  amplitude across the focal plane of $\delta(x+k)$. Since a lens is a
  linear optical device, any superposition of plane waves will be
  mapped to the corresponding superposition of delta functions. Thus
  the amplitude across the focal plane will be the Fourier transform
  of the amplitude across the face of the lens.}
\end{center}
\end{figure}

Fourier transforms are an important primitive in quantum computing,
and lie at the heart of the factoring algorithm, and several other
quantum algorithms. As shown in \cite{Shor_factoring}, quantum
computers can perform Fourier transforms on $n$-qubits using $O(n^3)$
gates\footnote{This has subsequently been improved. As shown in
  \cite{Coppersmith}, approximate Fourier transforms
  can be performed on $n$ qubits using $O(n)$ gates.}. Consider the
computational basis states of $n$-qubits as corresponding to the
numbers $\{0,1,\ldots,2^n-1\}$ via place value. The quantum Fourier
transform on $n$ qubits is the following unitary transformation.
\[
U_F \sum_{x=0}^{2^n-1} a(x) \ket{x} = 
\frac{1}{\sqrt{2^n}} \sum_{k=0}^{2^n-1} \sum_{x=0}^{2^n-1} e^{i 2 \pi
  x k/2^n} a(x) \ket{k}.
\]

The Fourier can be performed optically with a single lens, as shown in
figure \ref{lens}. Before digital computers became as powerful as they
are today, people used to develop optical schemes for analog
computation. Many of them were based in some way on the optical
Fourier transform.

Another important primitive in quantum algorithms is phase
kickback. Typically, an oracle $U_f$ for a give function $f$ acts
according to
\[
U_f \ket{x} \ket{y} = \ket{x} \ket{(y + f(x)) \ \mathrm{mod} \ 2^{n_o}},
\]
where $n_o$ is the number of output bits. Normally, one chooses $y=0$
so that the output register contains $f(x)$ after applying the
oracle. However, if one instead prepares the output register in the
state
\[
\ket{\psi} = \frac{1}{\sqrt{2^{n_o}}} \sum_{y=0}^{2^{n_o}} e^{i 2\pi
  y/2^{n_o}} \ket{y},
\]
$f(x)$ will be written into the phase instead of the qubits:
\[
U_f \ket{x} \ket{\psi} = e^{i 2 \pi f(x)/2^{n_o}} \ket{x} \ket{\psi}.
\]
This is because, for the process of adding $z$ modulo $2^{n_o}$,
$\ket{\psi}$ is an eigenstate with eigenvalue 
$e^{i 2 \pi z/2^{n_o}}$.

Phase kickback also has an optical analogue. Simply constuct a sheet
of glass whose thickness is proportional to $f(x)$. Because light
travels more slowly through glass than through air, the beam
experiences a phase lag proportional to the thickness of glass it
passes through. Now consider the following optical ``algorithm''. For
simplicity we'll demonstrate it in two dimensions, althogh it could
also be done in three. We are given a piece of glass whose thickness
is given by a smooth function $f(x)$. Because it is smooth, $f(x)$ is
locally linear, and so the sheet looks locally like a prism. By
inserting this glass between two lenses, as shown in figure
\ref{prism}, we can determine the angle of this prism
(\emph{i.e.} $\frac{df}{dx}$) by the location of the spot made on the
detector.
\begin{figure}
\begin{center}
\includegraphics[width=0.45\textwidth]{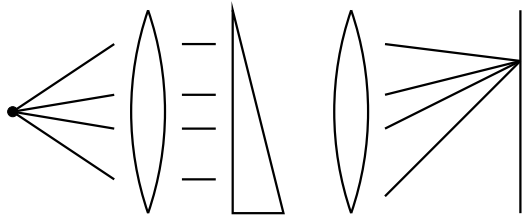}
\caption{\label{prism} At the left, a point source produces
  light. This is focused into a collimated beam by a lens. The beam
  then passes through a glass sheet which deflects the beam. The
  second lens then focuses the beam to a point on the detector whose
  location depends on the thickness gradient of the glass sheet.}
\end{center}
\end{figure}
By the analogies discussed above, this has an alternative description
in terms of Fourier transforms, and an analogous quantum algorithm, as
shown in figure \ref{quantum_prism}.

\begin{figure}
\begin{center}
\includegraphics[width=0.45\textwidth]{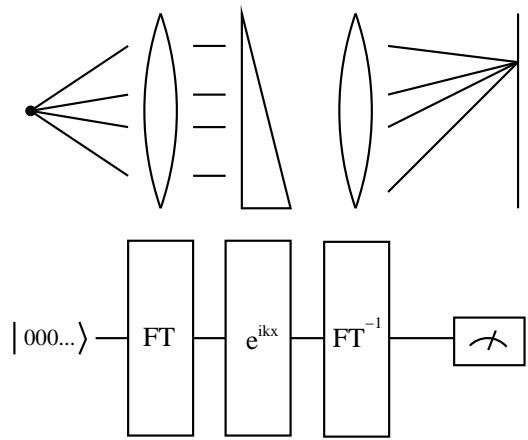}
\caption{\label{quantum_prism} Intially one starts with a given basis
  state. This is then Fourier transformed to yield the uniform
  superposition. A phase shift proportional to $f(x)$ is then
  performed. If $f(x)$ is locally linear than, the resulting shifted
  state is a plane wave $e^{i f'(x) x}$. The second Fourier transform
  then converts this to a basis state $\ket{f'(x)}$.}
\end{center}
\end{figure}

So far, this is an unimpressive accomplishment. Both determining the
angle of a prism and approximating the derivative of an oracular
function of one variable are easy tasks. Now recall the difference
between optical and quantum computing. Namely, the quantum analogue
can be extended to an arbitrary number of dimensions. In particular, a
phase proportional to $f(\vec{x})$ can easily be obtained using phase
kickback, and a $d$-dimensional quantum Fourier transform can be
achieved by applying a quantum Fourier transform to each vector
component. As a result, on a quantum computer, one can obtain the
gradient of a function of $d$ variables by generalizing the above
algorithm. This requires only a single query to the oracle to do the
phase kickback. In contrast, classically estimating the gradient of an
oracular function requires at least $d+1$ queries. This is described
in \cite{Jordan_gradient}.

A literature exists on analog optical computation. It might be
interesting to investigate whether some existing optical algorithms have
more powerful quantum analogues. Also, it seems that quantum computers
would be naturally suited to the simulation of classical optics
problems. Perhaps a quantum speedup could be obtained for optical
problems of practical interest, such as ray tracing.

\chapter{Phase Estimation}
\label{phase_estimation}

Phase estimation for unitary operators was introduced in \cite{Kit95},
and is explained nicely in \cite{Nielsen_Chuang}. Here I will describe
this method, and discuss how it can be used to measure in the
eigenbasis of physical observables.

Suppose you have a quantum circuit of $\mathrm{poly}(n)$ gates on $n$
qubits which implements the unitary transformation $U$. Suppose also
you are given some eigenstate $\ket{\psi_j}$ of $U$. You can always
efficiently estimate the corresponding eigenvalue to polynomial
precision using phase estimation. The first step in constructing the
phase estimation circuit is to construct a circuit for controlled-$U$.

Given an polynomial size quantum circuit for $U$, one can always
construct a polynomial size quantum circuit for controlled-$U$, which
is a unitary $\widetilde{U}$ on $n+1$ qubits defined by
\begin{eqnarray*}
\widetilde{U} \ket{0} \ket{\psi} & = & \ket{0} \ket{\psi} \\
\widetilde{U} \ket{1} \ket{\psi} & = & \ket{1} U \ket{\psi}
\end{eqnarray*}
for any $n$-qubit state $\ket{\psi}$. One can do this by taking the
quantum circuit for $U$ and replacing each gate with the corresponding
controlled gate. Each gate in the original circuit acts on at most $k$
qubits, where $k$ is some constant independent of $n$. Then, the
corresponding controlled gate acts on $k+1$ qubits. By general gate
universality results, any unitary on a constant number of qubits can
be efficiently implemented.

The phase estimation algorithm uses a series of controlled-$U^{2^m}$
operators for successive values of $m$. Given a quantum circuit for
$U$, one can construct a quantum circuit for $U^k$ by concatenating
$k$ copies of the original circuit. The resulting circuit for $U^k$
can then be converted into a circuit for controlled-$U^k$ as described
above.

Now consider the following circuit, which performs phase estimation to
three bits of precision.
\[
\mbox{ 
\Qcircuit @C=1em @R=.7em{
\lstick{\ket{0}}    & \qw    & \gate{H} & \ctrl{3}   & \qw        & \qw      & \multigate{2}{FT^{-1}} & \qw    & \qw \\
\lstick{\ket{0}}    & \qw    & \gate{H} & \qw        & \ctrl{2}   & \qw      & \ghost{FT^{-1}}        & \qw    & \qw \\
\lstick{\ket{0}}    & \qw    & \gate{H} & \qw        & \qw        & \ctrl{1} & \ghost{FT^{-1}}        & \qw    & \qw \\
\lstick{\ket{\psi_j}} & {/}\qw & \qw      & \gate{U^4} & \gate{U^2} & \gate{U} & \qw                    & {/}\qw & \qw
} }
\]
The top three qubits collectively form the control register. The
bottom $n$ qubits are initialized to the eigenstate $\ket{\psi_j}$ and
form the target register. The initial array of Hadamard gates puts the
control register into the uniform superposition. Thus the state of the
$n+3$ qubits after this step is 
\[
\frac{1}{\sqrt{8}} \sum_{x=0}^7 \ket{x} \ket{\psi_j}.
\]
Here we are using place value to make a correspondence between the
strings of 3 bits and the integers from 0 to 7. The
contolled-$U^{2^m}$ circuits then transform the state to
\[
\frac{1}{\sqrt{8}} \sum_{x=0}^7 \ket{x} U^x \ket{\psi_j}.
\]
$\ket{\psi_j}$ is an eigenstate of $U$, thus this is equal to
\[
\frac{1}{\sqrt{8}} \sum_{x=0}^7 \ket{x} e^{i x \theta_j} \ket{\psi_j} 
\]
where $e^{i \theta_j}$ is the eigenvalue of $\ket{\psi_j}$. Notice that
the control register is not entangled with the target
register. Performing an inverse Fourier transform on the control
register thus yields
\[
\ket{ \lceil \frac{8}{2 \pi} \theta_j \rfloor},
\]
where the notation $\lceil \cdot \rfloor$ indicates rounding to the
nearest integer. Thus we obtain $\theta$ with three bits of
precision.

Similarly, with $b$ qubits in the control register, one can obtain
$\theta_j$ to $b$ bits of precision. The necessary Hadamard and Fourier
transforms require only $\mathrm{poly}(b)$ gates to perform. However,
in order to obtain $b$ bits of precision, one must have a circuit for
controlled-$U^{2^b}$. The only known completely general way to
construct a circuit for $U^{2^b}$ from a circuit for $U$ is to
concatenate $2^b$ copies of the circuit for $U$. Thus, the total size
of the circuit for phase estimation will be on the order of $2^b g +
\mathrm{poly}(b)$, where $g$ is the number of gates in the circuit for
$U$. Thus, $\theta_j$ is obtained by the phase estimation algorithm to
within $\pm 2^{-b}$, and the dominant contribution to the total
runtime is proportional to $2^b$. In other words, $1/\mathrm{poly}(n)$
precision can be obtained in $\mathrm{poly}(n)$ time. 

For some special $U$, it is possible to construct a circuit for
$U^{2^b}$ using $\mathrm{poly}(b)$ gates. For example, this is true
for the operation of modular exponentiation. The quantum algorithm for
factoring can be formulated in terms of phase estimation of the
modular exponentiation operator, as described
in chapter five of \cite{Nielsen_Chuang}.

The phase estimation algorithm can be thought of as a special type of
measurement. Suppose you are given a superposition of eigenstates of
$U$. By linearity, the phase estimation circuit will perform the
following transformation
\[
\ket{00\ldots} \sum_j a_j \ket{\psi_j} \to \sum_j a_j \ket{\lceil \frac{8}{2 \pi}
  \theta_j \rfloor} \ket{\psi_j}.
\]
By measuring the control register in the computational basis, one
obtains the result $\lceil \frac{8}{2 \pi} \theta_j \rfloor$ with
probability $|a_j|^2$. If the eigenvalues of $U$ are separated by at
least $2^{-b}$, then one has thus performed a measurement in the
eigenbasis of $U$.

Often, one is interested in measuring in the eigenbasis of some
observable defined by a Hermition operator. For example, one may wish
to measure in the eigenbasis of a Hamiltonian, or an angular momentum
operator. This can be done by combining the techniques of phase
estimation and quantum simulation. As discussed in section
\ref{algorithms}, it is generally believed that any physically
realistic observable can be efficiently simulated using a quantum
circuit. More precisely, for any observable $H$, one can
construct a circuit for $U = e^{i H t}$, where the number of gates is
polynomial in $t$. Using this $U$ in the phase estimation algorithm,
one can measure eigenvalues of $H$ to polynomial precision.

\chapter{Minimizing Quadratic Forms}
\label{quadratic_forms}

A quadratic form is a function of the form $f(x) = x^{T} M x + b \cdot
x + c$, where $M$ is a $d \times d$ matrix, $x$ and $b$ are
$d$-dimensional vectors, and $c$ is a scalar. Here we consider
$M$, $x$, $b$, and $c$ to be real. Without loss generality we may
assume that $M$ is symmetric. If $M$ is also positive definite, then
$f$ has a unique minimum. In addition to its intrinsic interest, the
problem of finding this minimum by making queries to a blackbox for
$f$ can serve as an idealized mathematical model for numerical
optimization problems. 

Andrew Yao proved in 1975 that $\Omega(d^2)$ classical queries are
necessary to find the minimum of a quadratic form\cite{Yao}. In 2004,
I found that a single quantum query suffices to estimate the
gradient of a blackbox function, whereas a minimum of $d+1$ queries
are required classically\cite{Jordan_gradient}. One natural application for
this is gradient based numerical optimization. In 2005, David Bulger
applied quantum gradient estimation along with Grover search to the
problem of numerically finding the minimum to an objective function
with many basin-like local minima\cite{Bulger}. In the same paper he
suggested that it would be interesting to analyze the speedup
obtainable by applying quantum gradient estimation to the problem of
minimizing a quadratic form. In this appendix, I show that a simple
quantum algorithm based on gradient estimation can find the minimum of
a quadratic form using only $O(d)$ queries, thus beating the classical
lower bound.

For the blackbox for $f$ to be implemented on a digital computer, its
inputs and outputs must be discretized, and represented with some
finite number of bits. For present purposes, we shall ignore the
``numerical noise'' introduced by this discretization. In addition, on
quantum computers, blackboxes must be implemented as unitary
transformations. The standard way to achieve this is to let the
unitary transformation 
$U_f \ket{x}\ket{y} = \ket{x} \ket{(y + f(x)) \ \mathrm{mod} \ N_o}$
serve as the blackbox for $f$. Here $\ket{x}$ is the 
input register, $\ket{y}$ is the output register, and $N_o$ is the
allowed range of $y$. By choosing $y=0$, one obtains $f(x)$ in the
output register. As discussed in appendix \ref{optical}, by choosing
\[
\ket{y} \propto \sum_z e^{-i z} \ket{z},
\]
one obtains $U_f \ket{x}\ket{y} = e^{i f(x)} \ket{x}
\ket{y}$, since, in this case, $\ket{y}$ is an eigenstate of
addition modulo $N_o$. This is a standard technique in quantum
computation, known as phase kickback.

There are a number of methods by which one can find the minimum of a
quadratic form using $O(d)$ applications of quantum gradient
estimation. Since each gradient estimation requires only a single
quantum query to the blackbox, such methods only use $O(d)$ queries. 
One such method is as follows. First, we evaluate $\nabla f$ at $x =
0$. If $M$ is symmetric, $\nabla f = 2 M x
+ b$. Thus, this first gradient evaluation gives us $b$. Next we 
evaluate $\nabla f$ at $x = (1/2,0,0,\ldots)^T$. This yields $2 M
(1/2,0,0,\ldots)^T+b$. After subtracting $b$ we obtain the first column
of $M$. Similarly, we then evaluate $\nabla f$ at $(0,1/2,0,0,\ldots)^T$
and subtract $b$ to obtain the second column of $M$, and so on, until
we have full knowledge of $M$. Next we just compute $-\frac{1}{2}
M^{-1} b$ to find the minimum of $f$. This process uses a total of
$d+1$ gradient estimations, and hence $d+1$ quantum queries to the
blackbox for $f$. Note that $M$ is always invertible since by
assumption it is symmetric and positive definite.

\chapter{Principle of Deferred Measurement}
\label{deferred}

The principle of deferred measurement is a simple but conceptually
important fact which follows directly from the postulates of quantum
mechanics\cite{Nielsen_Chuang}. It says:
\\ \\
\emph{Any measurement performed in the course of a quantum computation
can always be deferred until the end of the computation. If any
operations are conditionally performed depending on the measurement
outcome they can be replaced by coherent conditional operations.}
\\ \\
For example, consider a process on two qubits, where one qubit is
measured in the computational basis, and if it is found to be in the
state $\ket{1}$ then the second qubit is flipped.
\[
\Qcircuit @C=1em @R=.7em {
& \qw & \meter \cwx[1] & \\
& \qw & \gate{X} & \qw
}
\] 
By the principle of deferred measurement, this is exactly equivalent
to
\[
\Qcircuit @C=1em @R=.7em {
& \qw & \ctrl{1} & \meter\\
& \qw & \targ & \qw
}
\]

The principle of deferred measurement allows us to immediately see
that the measurement based model of quantum computation (as described
in section \ref{measurement_based}) can be efficiently simulated by quantum
circuits. In addition, the principle of deferred measurement implies
that uniform families of quantum circuits generated in polynomial time
by quantum computers are no more powerful than uniform families of
quantum circuits generated in polynomial time by classical computers,
as shown below.

A quantum circuit can be described by series of bits corresponding to
possible quantum gates. If a bit is 1 then the corresponding gate is
present in the quantum circuit. Otherwise it is absent. Any quantum
circuit on $n$ bits with $\mathrm{poly}(n)$ gates chosen from a finite
set can be described by $\mathrm{poly}(n)$ classical bits in this
way. We can think of the quantum circuit as a series of classically
controlled gates, one for each bit in the description. Now suppose
these bits are generated by a quantum computer, which I'll call
the control circuit. By the principle of deferred measurement, these
classically controlled gates can be replaced by quantum controlled
gates. We now consider the control circuit and these controlled gates
together as one big quantum circuit. It is still of polynomial size,
and it is controlled by the classical computer that generated the
control circuit. Thus it is in BQP.



\chapter{Adiabatic Theorem}
\label{adiabatic_theorem}
\section{Main Proof}

This appendix give a proof of the adiabatic theorem due to Jeffrey
Goldstone \cite{Jeffrey}.

\begin{theorem}
Let $H(s)$ be a finite-dimensional twice differentiable Hamiltonian on
$0 \leq s \leq 1$ with a nondegenerate ground state $\ket{\phi_0(s)}$
separated by an energy energy gap $\gamma(s)$. Let $\ket{\psi(t)}$ be
the state obtained by Schr\"odinger  time evolution with Hamiltonian
$H(t/T)$ starting with state $\ket{\phi_0(0)}$ at $t = 0$. Then, with
appropriate choice of phase for $\ket{\phi_0(t)}$,
\[
\begin{array}{l}
\| \Ket{\psi(T)} - \Ket{\phi_0(T)} \| \leq
\frac{1}{T} \left[ \frac{1}{\gamma(0)^2} \left\| \frac{\ud H}{\ud s} \right\|_{s=0}
+ \frac{1}{\gamma(1)^2} \left\| \frac{\ud H}{\ud s} \right\|_{s=1}
+ \int_0^1 \ud s \left( \frac{5}{\gamma^3} \left\| \frac{\ud H}{\ud s}
\right\|^2 + \frac{1}{\gamma^2} \left\| \frac{\ud^2 H}{\ud s^2}
\right\| \right) \right]. 
\end{array}
\]
\end{theorem}

\begin{proof}
Let $E_0(t)$ be the ground state energy of $H(t)$. Let $H_f(t) = H(t)
- E_0(t) \id$. It is easy to see that if $\ket{\psi(t)}$ is the state
obtained by time evolving an initial state $\ket{\psi(0)}$ according
to $H(t)$ then
\[
\ket{\psi_f(t)} = e^{i \int_0^t E_0(\tau) \ud \tau} \ket{\psi(t)}
\] 
is the state obtained by time evolving the initial state
$\ket{\psi(0)}$ according to $H_f(t)$. Since these states differ by only
a global phase, the problem of proving the adiabatic theorem reduces
to the case where $E_0(t) = 0$ for $0 \leq t \leq T$. We assume this
from now on. Thus
\[
H(s) \ket{\phi_0(s)} = 0,
\]
so
\begin{equation}
\label{phidot}
\frac{\ud}{\ud s} \ket{\phi_0} = -G \frac{\ud H}{\ud s} \ket{\phi_0},
\end{equation}
where
\begin{equation}
\label{realG}
G = \sum_{j \neq 0} \frac{\ket{\phi_j} \bra{\phi_j}}{E_j},
\end{equation}
and $\ket{\phi_0(s)},\ket{\phi_1(s)},\ket{\phi_2(s)},\ldots$ is an
eigenbasis for $H(s)$ with corresponding energies
$E_0(s)=0,E_1(s),E_2(s),\ldots$.

We start with the Schr\"odinger equation
\[
i \frac{\ud}{\ud t} \ket{\psi} = H(t/T) \ket{\psi},
\]
and rescale the time coordinate to obtain
\[
\frac{i}{T} \frac{\ud}{\ud s} \ket{\psi} = H(s) \ket{\psi}.
\]
The corresponding unitary evolution operator $U_T(s,s')$ is the
solution of
\begin{eqnarray*}
& & \frac{i}{T} \frac{\ud}{\ud s} U_T(s,s') = H(s) U_T(s,s') \\
& & U(s,s)=1,
\end{eqnarray*}
and also satisfies
\[
-\frac{i}{T} \frac{\ud}{\ud s'} U_T(s,s') = U_T(s,s') H(s').
\]
We wish to bound the norm of
\begin{eqnarray*}
\delta_T & = & \ket{\phi_0(1)} - U_T(1,0) \ket{\phi_0(0)} \\
& = & \int_0^1 \frac{\ud}{\ud s} \left[ U_T(1,s) \ket{\phi_0(s)} \right]
\ud s.
\end{eqnarray*}
Integration by parts yields
\[
\delta_T = \frac{i}{T} \left[ U_T(1,s) G^2 \frac{\ud H}{\ud s}
\ket{\phi_0} \right]_{s=0}^{s=1} - \frac{i}{T} \int_0^1 \ud s \ U_T(1,s)
\frac{\ud}{\ud s} \left( G^2 \frac{\ud H}{\ud s} \ket{\phi_0} \right).
\]
Thus, by the triangle inequality
\begin{equation}
\label{normint}
\| \delta_T \| \leq  \frac{1}{T} \left( \left\| G^2 \frac{\ud
    H}{\ud s} \right\|_{s=1} + \left\| G^2 \frac{\ud H}{\ud s}
    \right\|_{s=0} + \int_0^1 \ud s \left\| \frac{\ud}{\ud s} \left( G^2
    \frac{\ud H}{\ud s} \ket{\phi_0} \right) \right\| \right).
\end{equation}
A straightforward calculation gives
\begin{eqnarray*}
\frac{\ud}{\ud s} \left( G^2 \frac{\ud H}{\ud s} \ket{\phi_0} \right)
& = & \ket{\phi_0} \bra{\phi_0} \frac{\ud H}{\ud s} G^3 \frac{\ud
  H}{\ud s} \ket{\phi_0} - G \frac{\ud H}{\ud s} G^2 \frac{\ud H}{\ud
  s} \ket{\phi_0} + G^3 \frac{\ud H}{\ud s} \ket{\phi_0} \bra{\phi_0}
\frac{\ud H}{\ud s} \ket{\phi_0} \\
& & + G^2 \frac{\ud^2 H}{\ud s^2} \ket{\phi_0} - 2 G^2 \frac{\ud
  H}{\ud s} G \frac{\ud H}{\ud s} \ket{\phi_0}
\end{eqnarray*}
as shown in section \ref{calc}. Thus, by the triangle inequality and
submultiplicativity,
\begin{equation}
\label{normintegrand}
\left\| \frac{\ud}{\ud s} \left( G^2
    \frac{\ud H}{\ud s} \ket{\phi_0} \right) \right\| \leq 5 \left\|
    G^3 \right\| \left\| \frac{\ud H}{\ud s} \right\|^2 + \left\| G^2
    \right\| \left\| \frac{\ud^2 H}{\ud s^2} \right\|.
\end{equation}
Substituting equation \ref{normintegrand} into equation \ref{normint}
and noting that $\| G \| = 1/\gamma$ completes the proof. 
\end{proof}

\section{Supplementary Calculation}
\label{calc}

In this section we calculate
\[
\frac{\ud}{\ud s} \left( G^2 \frac{\ud H}{\ud s} \ket{\phi_0} \right).
\]
$G$ as described in equation \ref{realG} is not convenient to work
with. Roughly speaking, $G$ represents the ``operator'' $\frac{Q}{H}$,
where $Q$ is the projector
\[
Q = \id - \ket{\phi_0}\bra{\phi_0}.
\]
However, $H$ has a zero eigenvalue and is therefore not invertible. We
can define
\[
\widetilde{H} = H + \epsilon \ket{\phi_0} \bra{\phi_0},
\]
where $\epsilon$ is some arbitrary real constant. $\widetilde{H}$ is
invertible and furthermore, 
\[
G = Q \widetilde{H}^{-1} = \widetilde{H}^{-1}Q = Q \widetilde{H}^{-1} Q.
\]
Thus,
\begin{eqnarray*}
\frac{\ud}{\ud s} G 
& = & \frac{\ud}{\ud s} \left( Q \widetilde{H}^{-1} Q \right) \\
& = & \frac{\ud Q}{\ud s} \widetilde{H}^{-1} Q + Q
\frac{\ud \widetilde{H}^{-1}}{\ud s} Q + Q \widetilde{H}^{-1} 
\frac{\ud Q}{\ud s} \\
& = & \frac{\ud Q}{\ud s} G + Q \frac{\ud \widetilde{H}^{-1}}{\ud s} Q
+ G \frac{\ud Q}{\ud s}
\end{eqnarray*}
For any invertible operator $M$,
\[
\frac{\ud}{\ud s} M^{-1} = - M^{-1} \frac{\ud M}{\ud s} M^{-1}.
\]
Thus,
\[
\frac{\ud G}{\ud s} = \frac{\ud Q}{\ud s} G - Q \widetilde{H}^{-1}
\frac{\ud \widetilde{H}}{\ud s} \widetilde{H}^{-1} Q + G \frac{\ud
  Q}{\ud s}.
\]
$\epsilon$ is an $s$-independent constant, so $\frac{\ud
  \widetilde{H}}{\ud s} = \frac{\ud H}{\ud s}$, thus
\begin{equation}
\label{intermediate_eq}
\frac{\ud G}{\ud s} = \frac{\ud Q}{\ud s} G - G \frac{\ud H}{\ud s} G
+ G \frac{\ud Q}{\ud s}.
\end{equation}
Using
\[
\frac{\ud Q}{\ud s} = - \frac{\ud \ket{\phi_0}}{\ud s} \bra{\phi_0} -
\ket{\phi_0} \frac{\ud \bra{\phi_0}}{\ud s}
\]
and equation \ref{phidot} yields
\[
\frac{\ud Q}{\ud s} = G \frac{\ud H}{\ud s} \ket{\phi_0}\bra{\phi_0} +
\ket{\phi_0}\bra{\phi_0} \frac{\ud H}{\ud s} G.
\]
Substituting this into equation \ref{intermediate_eq} yields
\begin{equation}
\label{gdot}
\frac{\ud G}{\ud s} = G^2 \frac{\ud H}{\ud s} \ket{\phi_0}
\bra{\phi_0} + \ket{\phi_0} \bra{\phi_0} \frac{\ud H}{\ud s} G^2 - G
\frac{\ud H}{\ud s} G.
\end{equation}
With this expression for $\frac{\ud G}{\ud s}$ we can now easily
calculate $\frac{\ud}{\ud s} \left( G^2 \frac{\ud H}{\ud s}
\ket{\phi_0} \right)$. Specifically,
\begin{eqnarray*}
\frac{\ud}{\ud s} \left( G^2 \frac{\ud H}{\ud s}
\ket{\phi_0} \right) & = & \frac{\ud G}{\ud s} G \frac{\ud H}{\ud s}
\ket{\phi_0} + G \frac{\ud G}{\ud s} \frac{\ud H}{\ud s} \ket{\phi_0}
+ G^2 \frac{\ud^2 H}{\ud s^2} \ket{\phi_0} + G^2 \frac{\ud H}{\ud s}
\frac{\ud \ket{\phi_0}}{\ud s} \\
& = & \ket{\phi_0} \bra{\phi_0} \frac{\ud H}{\ud s} G^3 \frac{\ud
  H}{\ud s} \ket{\phi_0} - G \frac{\ud H}{\ud s} G^2 \frac{\ud H}{\ud
  s} \ket{\phi_0} + G^3 \frac{\ud H}{\ud s} \ket{\phi_0} \bra{\phi_0}
\frac{\ud H}{\ud s} \ket{\phi_0} \\
& & + G^2 \frac{\ud^2 H}{\ud s^2} \ket{\phi_0} - 2 G^2 \frac{\ud
  H}{\ud s} G \frac{\ud H}{\ud s} \ket{\phi_0}
\end{eqnarray*}

\begin{singlespace}
\bibliographystyle{plain}
\bibliography{thesis}
\end{singlespace}

\end{document}